\documentclass[a4paper,10pt]{article}
\usepackage{amssymb}
\usepackage{amsthm}
\usepackage{amsmath,empheq}

\newtheorem{tm}{Theorem}
\newtheorem{lm}[tm]{Lemma}
\newtheorem{pr}[tm]{Proposition}
\newtheorem{cor}[tm]{Corollary}

\theoremstyle{definition}
\newtheorem{df}[tm]{Definition}
\newtheorem*{ack}{Acknowledgements}

\theoremstyle{remark}
\newtheorem{rk}[tm]{Remark}

\newcommand{\dd}{\mathrm{d}}
\newcommand{\st}{$\ast$~}
\DeclareMathOperator{\wsc}{\hspace{1pt}\rightharpoonup\hspace{-11pt}{^\ast}\hspace{5pt}}
\DeclareMathOperator{\sint}{\int\hspace{-5.5pt}\ast\hspace{-5.5pt}\int}

\numberwithin{tm}{section}
\numberwithin{equation}{section}

\begin{document}
\title{On the transfer of energy towards infinity\\ in the theory of weak turbulence\\ for the nonlinear Schr\"odinger equation}
\author{A.H.M.~Kierkels\hspace{12pt} J.J.~L.~Vel\'azquez\\\small Institute for Applied Mathematics, University of Bonn\\\small Endenicher Allee 60, 53115 Bonn, Germany\\\small kierkels@iam.uni-bonn.de\hspace{12pt}velazquez@iam.uni-bonn.de}
\maketitle
\begin{abstract}
We study the mathematical properties of a kinetic equation, derived in \cite{EV14b}, which describes the long time behaviour of solutions to the weak turbulence equation associated to the cubic nonlinear Schr\"odinger equation. In particular, we give a precise definition of weak solutions and prove global existence of solutions for all initial data with finite mass. We also prove that any nontrivial initial datum yields the instantaneous onset of a condensate, i.e.~a Dirac mass at the origin for any positive time. Furthermore we show that the only stationary solutions with finite total measure are Dirac masses at the origin. We finally construct solutions with finite energy, which is transferred to infinity in a self-similar manner.\\


\end{abstract}
\section{Introduction}
The theory of weak turbulence, also termed wave turbulence, is a physical theory which describes the transfer of energy between different wave lengths in a large class of wave systems. As a rule the weak turbulence theory can be applied to problems that can be approximated to a leading order by a theory of linear waves that interact by means of weak nonlinearities. The early examples of such theories were introduced in \cite{H62,H63}, \cite{P29} and \cite{Z67}, in the contexts of water waves, phonon interaction and plasma physics respectively.\\

Most theories of weak turbulence share some general features which can be summarized as follows. The starting point is a wave equation in a homogeneous medium. This wave equation can be approximated in the range of parameters under consideration by means of a linear equation plus additional small nonlinear terms. The corresponding linear wave equation can be explicitly solved using Fourier analysis due to the fact that the homogeneity of the medium makes the problem invariant under translations. Moreover, the distribution of energy among the spectral modes corresponding to different wave numbers $k$, with $|k|=\frac{2\pi}{\lambda}$ and where $\lambda$ is the associated wave length, does not change in time for the solutions of the linear wave equation. However, the presence of small nonlinearities yields a slow transfer of energy between the different spectral modes.

The derivation of weak turbulence theories in the physical literature assumes that the nonlinear wave equation is solved with random initial data, which are chosen according to a specific probability distribution. In a suitable asymptotic limit it is then possible to derive equations describing the evolution of distributions of physical quantities, such as mass and energy density, among the different wave numbers $k$. Typically the resulting equation is a kinetic equation whose solutions exhibit irreversible behaviour, contrary to the starting wave equation which typically exhibits time reversible dynamics.

Notice that the whole approach has many analogies with the derivation of the Boltzmann equation, taking as starting point the Hamiltonian dynamics of a system of particles. Nevertheless, although rigorous derivations of the Boltzmann equation have been obtained for some specific particle interactions (cf.~\cite{GSRT14}, \cite{L75}, \cite{PSS14}), there so far has not been any fully rigorous derivation of a nonlinear weak turbulence equation starting from a wave system. This is in spite of the fact that multiple formal derivations, at the level of physical rigour, have been obtained in the physical literature for a large variety of wave systems (cf.~\cite{D06}, \cite{DNPZ92}, \cite{N11}, \cite{N68}, \cite{N69}, \cite{Z72}, \cite{ZF67b}). There are few mathematically rigorous results where a time irreversible kinetic equation is obtained starting from a time reversible wave system. One of these results is the one in \cite{BCEP08}, where a cubic kinetic equation is obtained from the Schr\"odinger equation for a system of many particles, interacting through a suitable potential. More precisely, the authors prove that the terms of series representing solutions of the original Hamiltonian system converge to the terms of a new series, representing a solution of the kinetic equation. Another result in this direction can be found in \cite{LS11}. In this paper the authors study a discretized nonlinear Schr\"odinger equation and, assuming the initial data to be a suitable probability measure, they prove that the evolution near the thermal equilibrium is given by a linearized kinetic equation of weak turbulence type.\\

One of the most extensively studied systems in the theory of weak turbulence in the physical literature is the one associated to the nonlinear Schr\"odinger equation $iu_t=-\Delta u+\varepsilon|u|^2u$, with $\varepsilon>0$ small (cf.~\cite{DNPZ92}, \cite{N11}, \cite{ZLF92}). In that case we denote by $F(t,k)=|\hat{u}(t,k)|^2$ the density distribution between the different Fourier modes, where $\hat{u}$ denotes the Fourier transform of $u$ with respect the $x$ variable. The corresponding theory of weak turbulence is then given by the following kinetic equation:
\begin{equation}\label{eq:S1E1}
\begin{split}
\partial_t F_1=\frac{\varepsilon^2}{\pi}\iiint_{(\mathbb{R}^3)^3}&\delta_0(k_1+k_2-k_3-k_4)\delta_0(|k_1|^2+|k_2|^2-|k_3|^2-|k_4|^2)\\
&\times\big[F_3F_4(F_1+F_2)-F_1F_2(F_3+F_4)\big]\dd k_2\dd k_3\dd k_4,
\end{split}
\end{equation}
where $F_i=F(t,k_i)$ for each $i\in\{1,2,3,4\}$. In the case of isotropic solutions, i.e.~assuming radial symmetry, then \eqref{eq:S1E1} reduces upto rescaling to
\begin{equation}\label{eq:S1E2}
\partial_t f_1=\frac12\iint_{[0,\infty)^2}W\big[(f_1+f_2)f_3f_4-(f_3+f_4)f_1f_2\big]\dd\omega_3\dd\omega_4,
\end{equation}
where $f(t,\omega):=F(t,k)$ with $\omega:=|k|^2$, where as before $f_i=f(t,\omega_i)$ for each $i\in\{1,2,3,4\}$, where $\omega_2=(\omega_3+\omega_4-\omega_1)_+$ and where $W=\min_i\{\sqrt{\omega_i}\}/\sqrt{\omega_1}$. The isotropic equation \eqref{eq:S1E2} is usually given without the factor $\frac12$, but for notational convenience we use this rescaling.\\

The mathematical theory of \eqref{eq:S1E2} has been studied in detail in \cite{EV14b} where several properties of the solutions of \eqref{eq:S1E2} have been obtained. In particular, a description of the long time asymptotics of the solutions for some classes of initial data, blow-up in finite time for some particular initial data, finite time condensation and other results have been established.

In the study of \eqref{eq:S1E2} it is convenient to reformulate this equation in terms of a mass density function of particles. Setting $g(t,\omega)=\sqrt{\omega}f(t,\omega)$, then
\begin{equation}\label{eq:S1E3}
\partial_t g_1=\frac12\iint_{[0,\infty)^2}\tilde{W}\bigg[\bigg(\frac{g_1}{\sqrt{\omega_1}}+\frac{g_2}{\sqrt{\omega_2}}\bigg)\frac{g_3g_4}{\sqrt{\omega_3\omega_4}}-\bigg(\frac{g_3}{\sqrt{\omega_3}}+\frac{g_4}{\sqrt{\omega_4}}\bigg)\frac{g_1g_2}{\sqrt{\omega_1\omega_2}}\bigg]\dd\omega_3\dd\omega_4,
\end{equation}
where now $g_i=g(t,\omega_i)$ for each $i\in\{1,2,3,4\}$, where $\omega_2=(\omega_3+\omega_4-\omega_1)_+$ and where $\tilde{W}=\min_i\{  \sqrt{\omega_i}\}$.

Several properties of \eqref{eq:S1E3} become more transparent using the weak formulation of this equation. Multiplying \eqref{eq:S1E3} by a smooth test function $\varphi=\varphi(\omega_1)$, integrating over $[0,\infty)$ and rearranging variables, we obtain (cf. \cite{EV14b})
\begin{equation}\label{eq:S1E4}
\partial_t\bigg(\int_{[0,\infty)}g\varphi\,\dd\omega\bigg)=\frac12\iiint_{[0,\infty)^3}\frac{g_1g_2g_3\tilde{W}}{\sqrt{\omega_1\omega_2\omega_3}}\big[\varphi_3+\varphi_4-\varphi_1-\varphi_2\big]\dd\omega_1\dd\omega_2\dd\omega_3,
\end{equation}
where $\varphi_i=\varphi(\omega_i)$ for each $i\in\{1,2,3,4\}$ and where $\omega_4:=(\omega_1+\omega_2-\omega_3)_+$.

The weak formulation \eqref{eq:S1E4} formally yields conservation laws for the total mass of the particles and for the total energy; indeed, using $\varphi\equiv1$ and $\varphi(\omega)=\omega$ we obtain
\begin{equation}\label{eq:S1E5}
\partial_t\bigg(\int_{[0,\infty)}g(t,\omega)\bigg[\begin{array}{c}1\\\omega\end{array}\bigg]\dd\omega\bigg)=0,
\end{equation}
assuming that $\int\!\!_{[0,\infty)}g(\omega)\,\dd\omega$ and $\int\!\!_{[0,\infty)}g(\omega)\omega\,\dd\omega$ are initially finite.\\

One of the issues considered in \cite{EV14b}, which is the main motivation for this work, is the rate of transfer of energy towards infinity, i.e.~the rate of transfer of energy towards modes with very large $\omega$.

Let us assume that \eqref{eq:S1E3} is solved with initial data having finite total mass; moreover, suppose without loss of generality that $\int\!\!_{[0,\infty)} g(0,\omega)\dd\omega=1$. It has been proved in \cite{EV14b} that a solution $g(t,\cdot)$ then converges in the sense of measures to a Dirac measure $\delta_a$ as $t\rightarrow\infty$, where $a\in[0,\infty)$ is the infimum of the smallest set $A^*\subset(0,\infty)$ such that (i) $A^*\cup\{0\}$ contains the support of $g(0,\cdot)$, and (ii) $\omega_1,\omega_2,\omega_3\in A^*$ and $\omega_1+\omega_2>\omega_3$ imply $(\omega_1+\omega_2-\omega_3)\in A^*$ (cf.~\cite[Thm.~3.2]{EV14b}). Note in particular that $a>0$ only if the support of $g(0,\cdot)$ is contained in a discrete lattice $L$ satisfying some rationality condition. In the following we will assume that this is not the case; we assume that $a=0$ and thus that $g(t,\cdot)$ converges in the sense of measures to the Dirac measure $\delta_0$ at zero as $t\rightarrow\infty$.

Supposing further that $g$ initially has nonzero finite energy, then \eqref{eq:S1E5} implies $\int\!\!_{[0,\infty)}g(t,\omega)\omega\,\dd\omega=\int\!\!_{[0,\infty)}g(0,\omega)\omega\,\dd\omega>0$ for all $t\in[0,\infty)$. Since now $g(t,\omega)\omega\,\dd\omega$ converges in the sense of measures to $\delta_0(\omega)\omega\,\dd\omega\equiv0$ as $t\rightarrow\infty$, it follows that most of the energy of the distribution is contained in $\omega\rightarrow\infty$ as $t\rightarrow\infty$. In \cite{EV14b} it has been conjectured that such a transfer of energy towards infinity takes place in a self-similar manner and a kinetic equation describing how such a transfer could take place has been suggested. The argument goes as follows.

Given that solutions $g(t,\cdot)$ tend to Dirac measures at zero as $t\rightarrow\infty$, it is natural to consider the evolution of particle distributions for which most of the mass is supported at zero. Letting
\begin{equation}\label{eq:S1E6}
g(t,\cdot)=\delta_0(\cdot)+G(t,\cdot)
\end{equation}
with $G$ a nonnegative distribution satisfying $\int\!\!_{[0,\infty)}G(t,\omega)\dd\omega\ll1$, then plugging \eqref{eq:S1E6} into \eqref{eq:S1E4} the cubic term $g_1g_2g_3$ yields terms containing no, one, two or three Dirac measures at zero. Using the fact that
\begin{equation}\label{eq:S1E7}
\frac{\min_i\{\sqrt{\omega_i}\}}{\sqrt{\omega_1\omega_2\omega_3}}\big[\varphi(\omega_3)+\varphi(\omega_4)-\varphi(\omega_1)-\varphi(\omega_2)\big],\text{ with }\omega_4=(\omega_1+\omega_2-\omega_3)_+,
\end{equation}
is continuous on $[0,\infty)^3$ for $\varphi\in C_c^2([0,\infty))$   (cf. \cite[Lm.~2.5]{EV14b}), it may be seen that the terms containing two or three Dirac masses at zero are identically zero. Given furthermore that the total mass of $G$ is assumed small, it can be expected that the contribution of the term containing no Dirac measures at zero, which is cubic in $G$, is significantly smaller than the terms containing one Dirac measure at zero, which are quadratic in $G$. Keeping then only these quadratic terms and noting that the variables $\omega_1$ and $\omega_2$ are interchangeable in \eqref{eq:S1E7}, we obtain the following equation for $G$:
\begin{equation}\label{eq:S1E8}
\begin{split}
&\partial_t\bigg(\int_{[0,\infty)}\varphi(\omega)G(t,\omega)\dd\omega\bigg)\\
&\indent\begin{split}
&=\iiint_{[0,\infty)^3}\frac{\delta_0(\omega_1)G_2G_3\tilde{W}}{\sqrt{\omega_1\omega_2\omega_3}}\big[\varphi_3+\varphi_4-\varphi_1-\varphi_2\big]\dd\omega_1\dd\omega_2\dd\omega_3\\
&\indent+\frac12\iiint_{[0,\infty)^3}\frac{G_1G_2\delta_0(\omega_3)\tilde{W}}{\sqrt{\omega_1\omega_2\omega_3}}\big[\varphi_3+\varphi_4-\varphi_1-\varphi_2\big]\dd\omega_1\dd\omega_2\dd\omega_3
\end{split}
\end{split}
\end{equation}
Integrating out the Dirac measures, we find that after relabelling the right hand side of \eqref{eq:S1E8} equals
\begin{equation}\label{eq:S1E9}
\begin{split}
&\iint_{\{\omega_1\geq\omega_2\geq0\}}\frac{G_1G_2}{\sqrt{\omega_1\omega_2}}\big[\varphi(\omega_2)+\varphi(\omega_1-\omega_2)-\varphi(0)-\varphi(\omega_1)\big]\dd\omega_1\dd\omega_2\\
&\indent+\frac12\iint_{[0,\infty)^2}\frac{G_1G_2}{\sqrt{\omega_1\omega_2}}\big[\varphi(0)+\varphi(\omega_1+\omega_2)-\varphi(\omega_1)-\varphi(\omega_2)\big]\dd\omega_1\dd\omega_2.
\end{split}
\end{equation}
Note now that the integrand in the first term of \eqref{eq:S1E9} is zero on the diagonal $\{\omega_1=\omega_2\geq0\}$, so that the integration can be restricted to the strictly sub\-diagonal set $\{\omega_1>\omega_2\geq0\}$. Also note that the integrand in the second term is symmetric, which allows for the splitting of integrals
\begin{equation}\label{eq:S1E10}
\iint_{[0,\infty)^2}[\cdots]\dd\omega_1\dd\omega_2=2\iint_{\{\omega_1>\omega_2\geq0\}}[\cdots]\dd\omega_1\dd\omega_2+\iint_{\{\omega_1=\omega_2\geq0\}}[\cdots]\dd\omega_1\dd\omega_2.\nonumber
\end{equation}
Combining then the integrations over $\{\omega_1>\omega_2\geq0\}$ we can rewrite \eqref{eq:S1E8} as
\begin{equation}\label{eq:S1E11}
\begin{split}
&\partial_t\bigg(\int_{[0,\infty)}\varphi(\omega)G(t,\omega)\dd\omega\bigg)\\
&\indent\begin{split}
&=\iint_{\{\omega_1>\omega_2\geq0\}}\frac{G_1G_2}{\sqrt{\omega_1\omega_2}}\big[\varphi(\omega_1+\omega_2)-2\varphi(\omega_1)+\varphi(\omega_1-\omega_2)\big]\dd\omega_1\dd\omega_2\\
&\indent+\frac12\iint_{\{\omega_1=\omega_2\geq0\}}\frac{G_1G_2}{\sqrt{\omega_1\omega_2}}\big[\varphi(\omega_1+\omega_2)-2\varphi(\omega_1)+\varphi(\omega_1-\omega_2)\big]\dd\omega_1\dd\omega_2.
\end{split}\end{split}
\end{equation}
The previous heuristic derivation suggests that equation \eqref{eq:S1E11} can be used to describe the transfer of energy towards infinity for solutions to \eqref{eq:S1E3}.

Note here that the right hand side of \eqref{eq:S1E11} is well-defined for arbitrary finite measures $G$ if $\varphi\in C^1([0,\infty])$, since then the mapping
\begin{equation}\label{eq:S1E12}
(\omega_1,\omega_2)\mapsto\frac{\varphi(\omega_1+\omega_2)-2\varphi(\omega_1)+\varphi(\omega_1-\omega_2)}{\sqrt{\omega_1\omega_2}}
\end{equation}
is uniformly continuous on $\{\omega_1\geq\omega_2\geq0\}$ (cf.~Remark \ref{rk:Taylorargument}). As furthermore \eqref{eq:S1E12} is zero on $\{\omega_1\geq\omega_2=0\}$, the second term on the right hand side of \eqref{eq:S1E11} vanishes if $G$ is nonatomic on $(0,\infty)$, i.e.~if $\int\!\!_{\{\omega_0\}}G\,\dd\omega=0$ for each $\omega_0\in(0,\infty)$.\\

In \cite{EV14b}, equation \eqref{eq:S1E11} was derived with slightly different arguments that emphasize the fact that the energy is transported towards large values of $\omega$. Assuming further sufficient smoothness of $G$, it was also noticed that \eqref{eq:S1E11} is the weak formulation of
\begin{equation}\label{eq:S1E13}
\begin{split}
\partial_tG(\omega)&=\frac{G(\omega)}{\sqrt{\omega}}\int_\omega^\infty\frac{G(\xi)\dd\xi}{\sqrt{\xi}}-\frac{G(\omega)}{\sqrt{\omega}}\int_0^\omega\frac{G(\xi)\dd\xi}{\sqrt{\xi}}+\int_0^\infty\frac{G(\omega+\xi)G(\xi)\dd\xi}{\sqrt{(\omega+\xi)\xi}}\\
&\indent+\frac12\int_0^\omega\frac{G(\omega-\xi)G(\xi)\dd\xi}{\sqrt{(\omega-\xi)\xi}}-\frac{G(\omega)}{\sqrt{\omega}}\int_0^\infty\frac{G(\xi)\dd\xi}{\sqrt{\xi}},
\end{split}
\end{equation}
which provides a link to the theory of coagulation-fragmentation equations; indeed, \eqref{eq:S1E13} can be written as
\begin{equation}\label{eq:S1E14}
\begin{split}
\partial_tG(\omega)&=\textstyle\frac12\int_0^\omega K(\omega-\xi,\xi)G(\omega-\xi)G(\xi)\dd \xi-G(\omega)\int_0^\infty K(\omega,\xi)G(\xi)\dd \xi\\
&\indent-\textstyle\frac12G(\omega)\int_0^\omega B(\omega,\xi,\xi)\dd \xi+\int_\omega^\infty G(\xi)B(\xi-\omega,\omega)\dd \xi,
\end{split}
\end{equation}
with $K(a,b)=(ab)^{-1/2}$ and $B(a,b)=K(a+b,a)G(a)+K(a+b,b)G(b)$. The first two terms on the right hand side of \eqref{eq:S1E14} represent coagulation with kernel $K$, and the last two terms correspond to fragmentation with rate $B$. The only difference between \eqref{eq:S1E13} and the usual coagulation-fragmentation equations arises from the fact that the fragmentation rate $B$ depends on the solution $G$ itself. In fact, it is an easy computation to show that any coagulation-fragmentation equation with this kind of nonlinear fragmentation formally conserves initially finite zeroth and first moments.\\

A different particle interpretation of \eqref{eq:S1E11} will be given later (cf.~Remark \ref{rk:particleinterpretation}).

\subsection{Preliminary definitions and notation}
\begin{df}\label{df:measurespace}
By $\mathcal{M}_+([0,\infty])$ and $\mathcal{M}([0,\infty])$ we denote the spaces of finite nonnegative Radon measures and of signed nonnegative Radon measures on $[0,\infty]$ respectively. Given furthermore any interval $I$ of the form $[a,\infty]$, $(a,\infty]$, $[a,\infty)$ or $(a,\infty)$ with $0\leq a<\infty$, we denote by $\mathcal{M}_+(I)$ (resp.~$\mathcal{M}(I)$) the space of measures $\mu\in\mathcal{M}_+([0,\infty])$ (resp.~$\mu\in\mathcal{M}([0,\infty])$) for which $\mu\equiv0$ on $[0,\infty]\setminus I$.
\end{df}
\begin{rk}
In our notation of integrals with respect to measures, we will always write $\mu(x)\dd x$, even if $\mu$ is not absolutely continuous with respect to Lebesgue measure. Also, for any $\mu\in\mathcal{M}([0,\infty])$, we write $\|\mu\|=\int\!\!_{[0,\infty]}|\mu(x)|\dd x$.
\end{rk}
\begin{df}
By $C(I)$, with $I$ one of the intervals in Definition \ref{df:measurespace}, we denote the set of functions that are continuous on $I$, by $C^k(I)$, with $k\in\mathbb{N}\cup\{0\}$, the set of functions in $C(I)$ for which the derivatives of order upto $k$ exist and are in $C(I)$, and by $C_c^k(I)$ the set of functions in $C^k(I)$ supported in a compact $K\subset I$. Finally, we write $C(I)=C^0(I)$ and $C_c(I)=C_c^0(I)$.
\end{df}
\begin{rk}
Note that if $I$ is open ended on any side, then near that open end we do not impose any boundedness condition on functions, or derivatives thereof, in $C^k(I)$. Also, if $I$ is closed at infinity and $\varphi\in C^k(I)$, then $\lim_{x\rightarrow\infty}\varphi(x)$ exists and $\lim_{x\rightarrow\infty}\varphi^{(l)}(x)=0$ for all $l\in\{1,\ldots,k\}$.
\end{rk}
\begin{rk}
Given an open set $U\subset\mathbb{R}$ and $\alpha\in(0,1)$, we denote by $C^{0,\alpha}(U)$ the set of functions that are $\alpha$-H\"older continuous on any compact $K\subset U$.
\end{rk}
\begin{rk}
Given $\varphi\in C(I)$, we will write $\|\varphi\|_\infty=\|\varphi\|_{C(I)}$.
\end{rk}
\begin{df}
By the weak-\st topology on $\mathcal{M}_+(I)$ we mean the smallest topology such that the mapping $\mu\in\mathcal{M}_+(I)\mapsto\int\!\!_I\varphi(x)\mu(x)\dd x$ is continuous for all test functions $\varphi\in C_0(I):=\{\varphi\in C(\bar{I}):\varphi\equiv0\text{ on }\bar{I}\setminus I\}$.
\end{df}
\begin{rk}
Since the space $C_0(I)$, when endowed with the supremum norm, is a separable Banach space, it follows through Banach-Alaoglu (cf.~\cite{B11}) that the unit ball in $\mathcal{M}_+(I)$ is compact with respect to the weak-\st topology. Moreover, the weak-$\ast$~topology is metrizable (cf.~\cite[Thm.~3.28]{B11}) and can thus be uniquely characterized by means of convergence of sequences; a sequence $\{\mu_n\}$ in $\mathcal{M}_+(I)$ is said to converge to $\mu$ with respect to the weak-\st topology (notation: $\mu_n\wsc\mu$) if and only if $\int\!\!_{I}\varphi(x)\mu_n(x)\dd x\rightarrow\int\!\!_{I}\varphi(x)\mu(x)\dd x$ for all $\varphi\in C_0(I)$.
\end{rk}
\begin{rk}
Whenever we use the spaces of measures $\mathcal{M}_+(I)$, it will be endowed with the weak-\st topology, unless stated otherwise.
\end{rk}
\begin{rk}
By $\mathcal{M}_+((0,\infty);(1+x))$ we denote the set of nonnegative measures $\mu\in\mathcal{M}_+((0,\infty))$ for which $\int\!\!_{(0,\infty)}(1+x)\mu(x)\dd x<\infty$. By $L^1(0,\infty;(1+x))$ we denote the set of functions $f\in L^1(0,\infty)$ for which $\int\!\!_{(0,\infty)}(1+x)|f(x)|\dd x<\infty$.
\end{rk}
\begin{rk}
We use the notations $a\vee b=\max\{a,b\}$ and $a\wedge b=\min\{a,b\}$.
\end{rk}
\begin{df}
Given any $\varphi\in C([0,\infty))$ we define $\Delta_\varphi:\{x\geq y\geq0\}\rightarrow\mathbb{R}$ as
\begin{equation}\label{eq:def;Delta}
\Delta_\varphi(x,y):=\varphi(x+y)-2\varphi(x)+\varphi(x-y).
\end{equation}
Note that on the domain of definition of $\Delta_\varphi(x,y)$, the first argument is always the larger one.
\end{df}
\begin{rk}\label{rk:Taylorargument}
Given a test function $\varphi\in C^2([0,\infty))$, we observe that for $x\geq y\geq0$ we can write
\begin{align}
\Delta_\varphi(x,y)&\textstyle=\int_x^{x+y}\varphi'(z)\dd z-\int_{x-y}^x\varphi'(z)\dd z\label{eq:Deltaestimate1}\\
&\textstyle=\int_x^{x+y}\int_x^z\varphi''(w)\dd w\dd z-\int_{x-y}^x\int_x^z\varphi''(w)\dd w\dd z\nonumber\\
&\textstyle=\int_x^{x+y}(x+y-w)\varphi''(w)\dd w+\int_{x-y}^x(w-x+y)\varphi''(w)\dd w.\label{eq:Deltaestimate2}
\end{align}
Estimating then the individual terms in the right hand sides of \eqref{eq:def;Delta}, \eqref{eq:Deltaestimate1} and \eqref{eq:Deltaestimate2}, we obtain the estimate
\begin{equation}\label{eq:S1E18}
|\Delta_\varphi(x,y)|\leq\min\big\{4\|\varphi\|_{C([0,2x])},2y\|\varphi'\|_{C([0,2x])},y^2\|\varphi''\|_{C([0,2x])}\big\}.
\end{equation}
For a function $\varphi\in C^1([0,\infty])$, using only the estimates on the right hand sides of \eqref{eq:def;Delta} and \eqref{eq:Deltaestimate1}, we now find that
\begin{equation}\label{eq:S1E19}
|\Delta_\varphi(x,y)|\leq\min\big\{4\|\varphi\|_\infty,2y\|\varphi'\|_\infty\big\},
\end{equation}
and moreover, we also have that
\begin{equation}\label{eq:S1E20}
\textstyle\lim_{y\rightarrow0}\frac1y\sup_{\{x\geq y\geq 0\}}\Delta_\varphi(x,y)=\varphi'(x)-\varphi'(x)=0.
\end{equation}
In particular, the above implies that the mapping $(x,y)\mapsto(xy)^{-1/2}\Delta_\varphi(x,y)$ is uniformly continuous on $\{x\geq y\geq 0\}$ for any $\varphi\in C^1([0,\infty])$.
\end{rk}
\begin{df}
Given any two measures $\mu_1,\mu_2\in\mathcal{M}_+([0,\infty])$ we write
\begin{equation}\label{eq:integralnotation}
\begin{split}\nonumber
\sint[\cdots]\mu_1(x)\mu_2(y)\dd x\dd y&:=\iint_{\{x>y\geq0\}}[\cdots]\mu_1(x)\mu_2(y)\dd x\dd y\\
&\indent+\frac12\iint_{\{x=y\geq0\}}[\cdots]\mu_1(x)\mu_2(y)\dd x\dd y.
\end{split}
\end{equation}
\end{df}
\subsection{Notion of weak solution}
\begin{df}\label{def:notionofsolution}
A function $G\in C([0,\infty):\mathcal{M}_+([0,\infty)))$ that for all $t\in[0,\infty)$ and all $\varphi\in C^1([0,\infty):C_c^1([0,\infty)))$ satisfies
\begin{equation}\label{eq:def:solution}
\begin{split}
&\int_{[0,\infty)}\varphi(t,x)G(t,x)\dd x-\int_{[0,\infty)}\varphi(0,x)G(0,x)\dd x\\
&\indent\begin{split}=\int_0^t\bigg[&\int_{[0,\infty)}\varphi_s(s,x)G(s,x)\dd x\\&\indent+\sint\frac{G(s,x)G(s,y)}{\sqrt{xy}}\Delta_{\varphi(s,\cdot)}(x,y)\dd x\dd y\bigg]\dd s\end{split}
\end{split}
\end{equation}
will be called a {\em weak solution} to \eqref{eq:def:solution}.
\end{df}
\begin{rk}
We would like to emphasize that the notion of solution introduced in Definition \ref{def:notionofsolution} implicitly requires the total measure of a solution to be finite for all times. In fact, it can be seen that $\int\!\!_{[0,\infty)}G(\cdot,x)\dd x$ is uniformly bounded on $[0,T]$ for any $T\in[0,\infty)$ and any function $G\in C([0,\infty):\mathcal{M}_+([0,\infty)))$ (cf.~Lemma \ref{lm:BanachSteinh}). This additionally implies that all terms in \eqref{eq:def:solution} are well defined for arbitrary $G$. Note that in this paper we generally do not consider solutions with infinite total measure. The only exception to this is Remark \ref{rk:3_14}.
\end{rk}
\begin{lm}\label{lm:BanachSteinh}
Given any $G\in C([0,\infty):\mathcal{M}_+([0,\infty)))$ it holds that
\begin{equation}\label{eq:BS}
\textstyle\sup_{t\in[0,T]}\|G(t,\cdot)\|<\infty\text{ for all }T\in[0,\infty).
\end{equation}
\end{lm}
\begin{proof}
Fixing $T\in[0,\infty)$ we find that $\sup_{t\in[0,T]}\int\!\!_{[0,\infty)}\varphi(x)G(t,x)\dd x<\infty$ for all $\varphi\in C_0([0,\infty))$. By Banach-Steinhaus (cf.~\cite{B11}) then
\begin{equation}\label{eq:BS_2}\nonumber
\textstyle\sup_{t\in[0,T]}\sup_{\varphi\in C_0([0,\infty)),\|\varphi\|_\infty\leq1}\int_{[0,\infty)}\varphi(x)G(t,x)\dd x<\infty,
\end{equation}
and taking an increasing sequence $\{\varphi_n\}$ in $C_0([0,\infty))$ converging to $\varphi\equiv1$ as $n\rightarrow\infty$, locally uniformly on $[0,\infty)$, \eqref{eq:BS} follows by monotone convergence.
\end{proof}
\begin{rk}\label{rk:remark1_16}
Notice that \eqref{eq:def:solution} implies that weak solutions $G\in C([0,\infty):\mathcal{M}_+([0,\infty)))$ satisfy a more general identity in which for $t_1\in[0,\infty)$ and $t_2\in[t_1,\infty)$ the left hand side of \eqref{eq:def:solution} is replaced by $\int\!\!_{[0,\infty)}\varphi(t_2,x)G(t_2,x)\dd x-\int\!\!_{[0,\infty)}\varphi(t_1,x)G(t_1,x)\dd x$ and where on the right hand side the integral $\int_0\hspace{-4pt}^t\dd s$ is replaced by $\int_{t_1}\hspace{-7pt}^{t_2}\dd s$.
\end{rk}
\begin{rk}\label{rk:particleinterpretation}
We can give a particle interpretation to weak solutions $G$ to \eqref{eq:def:solution}, differing from the one using the coagulation-fragmentation equation. Suppose that we have a collection of particles, with sizes $x$ distributed according to $G(t,x)\dd x$. The evolution of this system of particles is determined by choosing pairs of particles $\{x,y\}$ at a rate $(xy)^{-1/2}$, which produce with probability $\frac12$ the pair $\{x+y,x\wedge y\}$ or with probability $\frac12$ the pair $\{|x-y|,x\wedge y\}$. Note that this collision mechanism conserves the number of particles $\int\!\!_{[0,\infty)}G(\cdot,x)\dd x$, and for infinitely large systems also the total volume $\int\!\!_{[0,\infty)}xG(\cdot,x)\dd x$.
\end{rk}

\subsection{Statement of main results}
The body of this paper consists of three main parts. In Section \ref{sec:existence} we prove global existence of weak solutions.
\begin{tm}[Global existence]\label{tm:existence}
Given any $G_0\in\mathcal{M}_+([0,\infty))$ there exists at least one weak solution $G\in C([0,\infty):\mathcal{M}_+([0,\infty)))$ to \eqref{eq:def:solution} in the sense of Definition \ref{def:notionofsolution} that satisfies $G(0,\cdot)=G_0$.
\end{tm}
The strategy to prove Theorem \ref{tm:existence} is similar to the one used in \cite{EV14b} for a cubic weak turbulence equation, or in \cite{L04} for the quantum Boltzmann equation in the bosonic case. Namely, we regularize the kernel $(xy)^{-1/2}$ to make it bounded, prove existence for the regularized problem and then show convergence of these solutions to weak solutions to \eqref{eq:def:solution}.\\

In Section \ref{sec:properties} we prove certain properties of weak solutions to \eqref{eq:def:solution}. We start by showing that the total measure of solutions is preserved and that if the first moment is initially finite, then this moment is also constant in time (cf.~Proposition \ref{pr:conservationlaws}). Then we show that for all solutions $G$ the measure of the origin is monotonic, i.e.~that the mapping $t\mapsto\int\!\!_{\{0\}}G(t,x)\dd x$ is nondecreasing (cf.~Proposition \ref{pr:monotonicityoforigin}). Lastly we prove that all solutions, with finite mass, converge weakly-\st in $\mathcal{M}_+([0,\infty))$ to a Dirac measure at the origin as $t\rightarrow\infty$ (cf. Proposition \ref{pr:longtimebehaviour}).

In Subsection \ref{sec:instantdirac} we show that the measure of the origin is in fact strictly increasing; indeed, we prove the following result.
\begin{pr}[Instantaneous condensation]\label{pr:instantdirac}
Given any weak solution $G\in C([0,\infty):\mathcal{M}_+([0,\infty)))$ to \eqref{eq:def:solution}, then for any $\bar{t}\in[0,\infty)$ the following is true.
\begin{equation}\nonumber
\int_{(0,\infty)}G(\bar{t},x)\dd x>0 \Rightarrow \int_{\{0\}}G(t,x)\dd x>\int_{\{0\}}G(\bar{t},x)\dd x\text{ for all }t>\bar{t}.
\end{equation}
\end{pr}
\begin{rk}\label{rk:deftrivsol}
If a weak solution $G\in C([0,\infty):\mathcal{M}_+([0,\infty)))$ to \eqref{eq:def:solution} is supported in $[0,\infty)\times\{0\}$, we will refer to it as a {\em trivial solution}. Conservation of the total measure of weak solutions (cf.~Proposition \ref{pr:conservationlaws}) implies that trivial solutions have to be constant in time. As a consequence of Proposition \ref{pr:instantdirac}, we know that the only stationary weak solutions to \eqref{eq:def:solution}, with finite mass, are the trivial solutions (cf.~Corollary \ref{cor:uniquestationarysolutions}). In particular there are no stationary weak solutions with finite total measure and positive first moment.
\end{rk}
Following the terminology in \cite{EV14a,EV14b,L13,L14} we will say that weak solutions to \eqref{eq:def:solution} have a condensate if the measure of the origin is strictly positive. It was proven in \cite{EV14a} that there exist solutions $g$ to the quantum Boltzmann equation, which are bounded in a finite positive time interval - and in particular satisfying $\int\!\!_{\{0\}}g(\cdot,x)\dd x\equiv0$ on $[0,T]$ with $T\in(0,\infty)$ - but have $\int\!\!_{\{0\}}g(t,x)\dd x>0$ for some $t\in(T,\infty)$. Similar results have been obtained for the weak turbulence equation in \cite{EV14b}. Additionally it has been proven in \cite{L14} that there exist weak solutions $g$ to the quantum Boltzmann equation which are unbounded near the origin, but satisfying $\int\!\!_{\{0\}}g(\cdot,x)\dd x\equiv0$ on $[0,T]$ with $T\in(0,\infty)$ and $\int\!\!_{\{0\}}g(t,x)\dd x>0$ for some $t\in(T,\infty)$. Notice that the solutions constructed in \cite{EV14a,EV14b,L14} exhibit a behaviour that we can call condensation after finite time. Contrary to this, Proposition \ref{pr:instantdirac} shows that any nontrivial weak solutions to \eqref{eq:def:solution} has the property that $\int\!\!_{\{0\}}G(t,x)\dd x>0$ for all $t\in(0,\infty)$. We will refer to this behaviour as {\em instantaneous condensation}.\\

The remaining results in this paper concern self-similar behaviour of solutions. For a weak solution $G\in C([0,\infty):\mathcal{M}_+([0,\infty)))$ to \eqref{eq:def:solution} with initially finite and nonzero first moment $E$ the results in Section \ref{sec:properties} imply that (i) $\int\!\!_{[0,\infty)}G(t,x)\dd x=\int\!\!_{[0,\infty)}G(0,x)\dd x=M$ for all $t\in[0,\infty)$ and $G(t,\cdot)\wsc M\delta_0(\cdot)$ as $t\rightarrow\infty$, and (ii) $\int\!\!_{[0,\infty)}xG(t,x)\dd x=E$ for all $t\in[0,\infty)$ and $xG(t,x)\dd x\wsc0$ as $t\rightarrow\infty$. Therefore the main contribution to $\int\!\!_{[0,\infty)}xG(t,x)\dd x$ must be due to large values of $x$. In a physical context this can be interpreted as transfer of energy towards infinity and it has been conjectured in \cite{EV14b} that this transfer occurs in a self-similar manner, which should be described by self-similar solutions of \eqref{eq:S1E13}. In Section \ref{sec:self-similarity} we prove that there indeed exist weak solutions to \eqref{eq:def:solution} that exhibit such self-similar transfer of energy towards infinity.
\begin{tm}[Self-similar solutions]\label{tm:self-similarsolutions}
Given $E\in[0,\infty)$ there exists at least one nonnegative $\Phi\in C^{0,\alpha}((0,\infty))\cap L^1(0,\infty)$, $\alpha<\frac12$, with $\int\!\!_{[0,\infty)}x\Phi(x)\dd x=E$ such that given any $t_0\in(0,\infty)$, then for all $M\in[0,\infty)$ satisfying $M\sqrt{t_0}\geq\|\Phi\|_{L^1(0,\infty)}$, the function $G\in C([0,\infty):\mathcal{M}_+([0,\infty)))$ that is given by
\begin{equation}\label{eq:ssesolution}
G(t,x)\dd x=\bigg(M-\frac{\|\Phi\|_{L^1(0,\infty)}}{\sqrt{t+t_0}}\bigg)\delta_0(x)\dd x+\frac{1}{t+t_0}\Phi\Big(\frac{x}{\sqrt{t+t_0}}\Big)\dd x,
\end{equation}
is a weak solution to \eqref{eq:def:solution} in the sense of Definition \ref{def:notionofsolution}. This solution then satisfies $\|G(t,\cdot)\|=M$ and $\int\!\!_{[0,\infty)}xG(t,x)\dd x=E$ for all $t\in[0,\infty)$.
\end{tm}
Let us make precise what we mean by self-similar solutions. Elementary computations show that any weak solution $G$ to \eqref{eq:def:solution} gives rise to a two-parameter family of weak solutions
\begin{equation}\label{eq:resckl}
G_{\kappa,\lambda}(t,x)\dd x=\kappa G(\kappa\lambda t,\lambda x)\lambda\dd x\text{ for }\kappa,\lambda\in(0,\infty).
\end{equation}
In problems with two scaling parameters, like the one under consideration here, it is customary to relate the parameters using the conserved quantities of the system. In this particular problem we have two conserved quantities, namely the mass and the energy. The respective relations between the parameters $\kappa$ and $\lambda$ retaining these conservation laws are $\kappa=1$ and $\kappa=\lambda$; indeed, for $\lambda\in(0,\infty)$ we see that
\begin{align}\label{eq:resc1l}
\textstyle\int_{[0,\infty)}G_{1,\lambda}(t,x)\,\dd x&=\textstyle\int_{[0,\infty)}G(\lambda t,x)\,\dd x=M,\\\label{eq:rescll}\textstyle\int_{[0,\infty)}xG_{\lambda,\lambda}(t,x)\,\dd x&=\textstyle\int_{[0,\infty)}xG(\lambda^2t,x)\,\dd{x}=E.
\end{align}
Notice that for any rescaling, other than the trivial one ($\kappa=\lambda=1$), the rescaled solutions $G_{\kappa,\lambda}$ have either a different mass or a different energy than the original solution $G$. In the literature, self-similar solutions are usually referred to solutions that are invariant under some one-parameter rescaling group. Even though there is no rescaling group conserving at the same time both the mass and the energy, tentatively we can look for solutions which leave invariant either the mass or the energy. This amounts to looking for solutions satisfying either $G\equiv G_{1,\lambda}$ or $G\equiv G_{\lambda,\lambda}$. Following \eqref{eq:resckl} these solutions should then have the following forms respectively
\begin{align}\label{eq:selfsimmasssol}
\textstyle G_0(t,x)\dd x&\textstyle=\Phi_0\big(\frac{x}{t}\big)\frac{\dd x}{t},\\
\textstyle G_1(t,x)\dd x&\textstyle=\frac1{\sqrt{t}}\Phi_1\big(\frac{x}{\sqrt{t}}\big)\frac{\dd x}{\sqrt{t}}.\nonumber
\end{align}
As remarked before, nontrivial solutions with these functional forms can not satisfy conservation of mass and energy simultaneously (cf.~\eqref{eq:resc1l}, \eqref{eq:rescll}); indeed, we obtain
\begin{align}\label{eq:contradictingssmasssol}
\textstyle \int_{[0,\infty)} xG_0(t,x)\,\dd x&\textstyle=t\int_{[0,\infty)} xG_0(0,x)\,\dd x,\\
\textstyle \int_{[0,\infty)} G_1(t,x)\,\dd x&\textstyle=\frac1{\sqrt{t}}\int_{[0,\infty)} G_1(0,x)\,\dd x.\nonumber
\end{align}
There thus are no solutions satisfying either $G\equiv G_{1,\lambda}$ or $G\equiv G_{\lambda,\lambda}$ on $x\in[0,\infty)$. However, given that the energy is actually independent of the value of $G$ at the origin, we can look for solutions satisfying $G\equiv G_{\lambda,\lambda}$, i.e.~being of the form \eqref{eq:rescll}, on $x\in(0,\infty)$ alone. Such a solution $G_1$ then satisfies
\begin{equation}\nonumber
\textstyle \int_{(0,\infty)} G_1(t,x)\,\dd x\textstyle=\frac1{\sqrt{t}}\int_{(0,\infty)} G_1(0,x)\,\dd x,
\end{equation}
which is still decreasing. Nevertheless, since the function $G_1$ is thus far only determined on $x\in(0,\infty)$, we can compensate the loss of mass in $(0,\infty)$ by adding a Dirac measure at the origin. Therefore the resulting solutions $G_1$ should have the form
\begin{equation}\nonumber
G_1(t,x)\dd x=m(t)\delta_0(x)\dd x+\tfrac1{\sqrt{t}}\Phi_1\big(\tfrac{x}{\sqrt{t}}\big)\tfrac{\dd x}{\sqrt{t}},
\end{equation}
where $m(t)$ is chosen in order to have mass conservation. Conservation of energy follows automatically. In Theorem \ref{tm:self-similarsolutions} we show that such solutions exist, and in view of the above we call functions of the form \eqref{eq:ssesolution} {\em generalized self-similar energy solutions}. Notice that since \eqref{eq:def:solution} is invariant under translations in $t$, we have added in \eqref{eq:ssesolution} an arbitrary shift in the origin of time.

A similar way to compensate for the violation of conservation of energy (cf.~\eqref{eq:contradictingssmasssol}) is not available for solutions of the form \eqref{eq:selfsimmasssol}; indeed, we have the following result.
\begin{pr}\label{prop:selfsimmasssol}
If a solution $G\in C([0,\infty):\mathcal{M}_+([0,\infty)))$ to \eqref{eq:def:solution} satisfies $G\equiv G_{1,\lambda}$ on $x\in(0,\infty)$ for all $\lambda\in(0,\infty)$, then this solution is trivial in the sense of Remark \ref{rk:deftrivsol}, i.e.~$G(t,\cdot)\equiv M\delta_0$ for all $t\in[0,\infty)$ and some $M\in[0,\infty)$.
\end{pr}
Some of the ideas used to prove Theorem \ref{tm:self-similarsolutions} have been used in the study of other kinetic equations, in particular in the study of coagulation equations (cf.~\cite{EMR05}, \cite{FL05}, \cite{NV13}).

\section{Existence of weak solution}\label{sec:existence}
In this section we prove Theorem \ref{tm:existence}. We first prove existence of solutions to a regularized version of \eqref{eq:def:solution}.
\begin{lm}\label{lm:ex;regularizedexistence}
Let $\varepsilon\in(0,1)$ and $G_0\in\mathcal{M}_+([0,\infty))$ be arbitrary. Then there exists at least one function $G\in C([0,\infty):\mathcal{M}_+([0,\infty]))$ that for all $t\in[0,\infty)$ and all $\varphi\in C^1([0,\infty):C([0,\infty]))$ satisfies
\begin{equation}\label{eq:regex}
\begin{split}
&\int_{[0,\infty]}\varphi(t,x)G(t,x)\dd x-\int_{[0,\infty]}\varphi(0,x)G_0(x)\dd x\\
&\indent\begin{split}
=\int_0^t\bigg[&\int_{[0,\infty]}\varphi_s(s,x)G(s,x)\dd x\\&\indent+\sint \frac{G(s,x)G(s,y)}{\sqrt{(x+\varepsilon)(y+\varepsilon)}}\Delta_{\varphi(s,\cdot)}(x,y)\dd x\dd y\bigg]\dd s.
\end{split}
\end{split}
\end{equation}
In particular, any such function satisfies $\|G(t,\cdot)\|=\|G_0\|$ for all $t\in[0,\infty)$.
\end{lm}
\begin{proof}
We restrict ourselves to nonzero $G_0\in\mathcal{M}_+([0,\infty))$, since otherwise the proof is trivial. First, we shall prove short time existence by means of a fixed point argument. Global existence will then follow using conservation of the total measure.\\

Fix $\varepsilon\in(0,1)$ arbitrarily and let $T\in(0,\infty)$ to be determined later. Fix further some nonnegative $\phi\in C_c^\infty((-1,1))$ with $\|\phi\|_{L^1(-1,1)}$ and for small ${a}\in(0,1)$ let $\phi_{a}(x):=\frac1{a}\phi(\frac{x}{a})$. We then define $\mathcal{T}_{a}$ from $C([0,T]:\mathcal{M}_+([0,\infty]))$ into itself to be such that for all $G\in C([0,T]:\mathcal{M}_+([0,\infty]))$, all $t\in[0,T]$ and all $\varphi\in C([0,\infty])$ it satisfies
\begin{equation}\label{eq:S2E2}\nonumber
\begin{split}
\int_{[0,\infty]}\varphi(x)\mathcal{T}_{a}[G](t,x)\dd x&=\int_{[0,\infty)}\varphi(x)G_0(x)e^{-\int_0^tA_{a}[G(s,\cdot)](x)\dd s}\dd x\\
&\hspace{-32pt}+\int_0^t\int_{[0,\infty)}\varphi(x)e^{-\int_s^tA_{a}[G(\sigma,\cdot)](x)\dd\sigma}B_{a}[G(s,\cdot)](x)\dd x\,\dd s,
\end{split}
\end{equation}
with $A_{a}:\mathcal{M}_+([0,\infty])\rightarrow C_0([0,\infty))$ given by
\begin{equation}\label{eq:S2E3}\nonumber
A_{a}[G](x):=2\int_0^x\frac{(\phi_a\ast G)(y)\dd y}{\sqrt{(x+\varepsilon)(y+\varepsilon)}},
\end{equation}
and where $B_{a}:\mathcal{M}_+([0,\infty])\rightarrow\mathcal{M}_+([0,\infty))$ is such that for all $G\in\mathcal{M}_+([0,\infty])$ and all $\varphi\in C([0,\infty])$ it satisfies
\begin{equation}\label{eq:S2E4}\nonumber
\int_{[0,\infty)}\varphi(x)B_{a}[G](x)\dd x=\!\iint_{\{x>y\geq0\}}\!\frac{G(x)(\phi_{a}\ast G)(y)}{\sqrt{(x+\varepsilon)(y+\varepsilon)}}\big(\varphi(x+y)+\varphi(x-y)\big)\dd x\dd y.
\end{equation}
Note that the operators $A_{a}$ and $B_{a}$ are well defined. For any ${a}\in(0,1)$, it is now easy to check that $\mathcal{T}_{a}$ maps the space $C([0,T]:\mathcal{M}_+([0,\infty]))$ into itself; indeed, both operators $A_{a}$ and $B_{a}$ are continuous on $\mathcal{M}_+([0,\infty])$, which implies that the mapping $t\mapsto\mathcal{T}_{a}[G](t,\cdot)$ is continuous on $[0,T]$.

Writing next $\|\mu\|_T=\sup_{t\in[0,T]}\|\mu(t,\cdot)\|$ for $\mu\in C([0,T]:\mathcal{M}_+([0,\infty]))$ and recalling that $\|\nu\|=\sup_{\varphi\in C([0,\infty]),\|\varphi\|_{\infty}\leq1}\int\!\!_{[0,\infty]}\varphi(x)\nu(x)\dd x$ for $\nu\in\mathcal{M}_+([0,\infty])$, we find, using that $A_{a}\geq0$, that
\begin{equation}\label{eq:S2E5}\nonumber
\|\mathcal{T}_{a}[G]\|_T\leq\|G_0\|+\tfrac{2T}\varepsilon\|G\|_T^2.
\end{equation}
Therefore, the set $\mathcal{X}_T:=\{C\in C([0,T]:\mathcal{M}_+([0,\infty])):\|G\|_T\leq2\|G_0\|\}$ is invariant under $\mathcal{T}_{a}$ if $T\leq\frac\varepsilon{8\|G_0\|}$.

Let next $G\in C([0,\infty):\mathcal{M}_+([0,\infty]))$, let $t_1\in[0,T]$ and $t_2\in[t_1,T]$ and let $\varphi\in C([0,\infty])$ with $\|\varphi\|_\infty\leq1$. Then, using the fact that $|e^{-x}-1|\leq x$ for $x\in(0,\infty)$, we can see that
\begin{equation}\label{eq:S2E6}\nonumber
\begin{split}
\bigg|\int_{[0,\infty]}\varphi(x)\mathcal{T}_{a}[G](t_2,x)\dd x-\int_{[0,\infty]}\varphi(x)\mathcal{T}_{a}[G](t_1,x)\dd x\bigg|
&\textstyle\leq\int_{t_1}^{t_2}\|B_{a}[G(s,\cdot)]\|\dd s\\
&\textstyle\hspace{-192pt}+\big(\|G_0\|+\int_0^{t_1}\|B_{a}[G(s,\cdot)]\|\dd s\big)\,\big\|\int_{t_1}^{t_2}A_{a}[G(s,\cdot)](\cdot)\dd s\big\|_\infty,
\end{split}
\end{equation}
which, if $T\leq\frac\varepsilon{8\|G_0\|}$ and $G\in \mathcal{X}_T$, is bounded by
\begin{equation}\label{eq:S2E7}
\textstyle\Big(\frac8\varepsilon\|G_0\|^2+\big(\|G_0\|+\frac{8T}\varepsilon\|G_0\|^2\big)\frac4\varepsilon\|G_0\|\Big)\,\big|t_2-t_1\big|.
\end{equation}
By Arzel\`a-Ascoli (cf.~\cite{DS63}) it then follows that $\mathcal{T}_{a}$ is a compact operator on $\mathcal{X}_T$ and by Schauder's fixed point theorem there exists at least one $G_{a}\in\mathcal{X}_T$ such that $\mathcal{T}_{a}[G_{a}]\equiv G_{a}$ on $[0,T]\times[0,\infty]$. In particular, $G_{a}$ now satisfies
\begin{equation}\label{eq:mildsolution}
\begin{split}
\int_{[0,\infty]}\varphi(x)G_{a}(t,x)\dd x&=\int_{[0,\infty)}\varphi(x)G_0(x)e^{-\int_0^tA_{a}[G_{a}(s,\cdot)](x)\dd s}\dd x\\
&\hspace{-32pt}+\int_0^t\int_{[0,\infty)}\varphi(x)e^{-\int_s^tA_{a}[G_{a}(\sigma,\cdot)](x)\dd\sigma}B_{a}[G_{a}(s,\cdot)](x)\dd x\,\dd s,
\end{split}
\end{equation}
for all $\varphi\in C([0,\infty])$ and all $t\in[0,T]$. It then follows that
\begin{equation}\label{eq:S2E8}
\partial_s\bigg[{\int_{[0,\infty]}\varphi(x)G_a(s,x)\dd x}\bigg]\!\!=\!\!\iint_{\{x>y\geq0\}}\frac{G_{a}(s,x)(\phi_{a}\ast G_{a}(s,\cdot))(y)}{\sqrt{(x+\varepsilon)(y+\varepsilon)}}\Delta_\varphi(x,y)\dd x\dd y,
\end{equation}
and since $(\phi_{a}\ast G_{a}(s,\cdot))$ is smooth for all $s\in[0,T]$, we note that the integral over $\{x=y\geq0\}$ with respect to the product measure $G_{a}(s,\cdot)\times(\phi_{a}\ast G_{a}(s,\cdot))(\cdot)$ is zero. Therefore, we can write the integral of \eqref{eq:S2E8} with respect to $s$ from $0$ to $t\in[0,T]$ as
\begin{equation}\label{eq:S2E9}
\begin{split}
&\int_{[0,\infty]}\varphi(x)G_{a}(t,x)\dd x-\int_{[0,\infty]}\varphi(x)G_0(x)\dd x\\
&\hspace{8pt}=\frac12\int_0^t\iint_{[0,\infty]^2}\frac{G_{a}(s,x\vee y)(\phi_{a}\ast G_{a}(s,\cdot))(x\wedge y)}{\sqrt{(x+\varepsilon)(y+\varepsilon)}}\Delta_\varphi(x\vee y,x\wedge y)\dd x\dd y\,\dd s.
\end{split}
\end{equation}
Using now $\varphi\equiv1$ in \eqref{eq:S2E9}, it follows immediately that $\int\!\!_{[0,\infty]}G_{a}(t,x)\dd x=\|G_0\|$ for all $t\in[0,T]$. In order to prove global existence we can then repeat the fixed point argument above, starting at time $t=T$ and using initial datum $\tilde{G}_0:=G(T,\cdot)$. This will give us a solution to \eqref{eq:mildsolution} in the interval $[0,2T]$; more precisely, we can find a solution $G_{a}\in C([T,2T]:\mathcal{M}_+([0,\infty]))$, satisfying
\begin{equation}\label{eq:mildsolution2}
\begin{split}
\int_{[0,\infty]}\varphi(x)G_{a}(t,x)\dd x&=\int_{[0,\infty)}\varphi(x)G_{a}(T,x)e^{-\int_T^tA_{a}[G_{a}(s,\cdot)](x)\dd s}\dd x\\
&\hspace{-32pt}+\int_T^t\int_{[0,\infty)}\varphi(x)e^{-\int_s^tA_{a}[G_{a}(\sigma,\cdot)](x)\dd\sigma}B_{a}[G_{a}(s,\cdot)](x)\dd x\,\dd s,
\end{split}
\end{equation}
for all $\varphi\in C([0,\infty])$ and all $t\in[T,2T]$. The solvability of this last problem follows from a fixed point argument analogous to the one before, where we note that we can find a solution in $[T,2T]$ because of the fact that $T$ depends - for fixed $\varepsilon,a\in(0,1)$ - only on $\|\tilde{G}_0\|=\|G_0\|$.
We now set $t=T$ in \eqref{eq:mildsolution}, we replace the test function $\varphi(x)$ by $\varphi(x)\exp({-\int_T^tA[G_{a}(s,\cdot)](x)\dd s})$, with $t\in[T,2T]$ fixed, and we add the resulting equation to \eqref{eq:mildsolution2} to see that $G_{a}$ indeed satisfies \eqref{eq:mildsolution} for all $t\in[0,2T]$. Iterating this procedure we obtain for any $a\in(0,1)$ a function $G_{a}\in C([0,\infty):\mathcal{M}_+([0,\infty]))$ satisfying \eqref{eq:mildsolution} for all $\varphi\in C([0,\infty])$ and all $t\in[0,\infty)$.\\

To finish the proof of this Lemma, we shall consider the limit $a\rightarrow0$. To that end we note that the family of functions $\mathcal{G}=\{G_a\}_{a\in(0,1)}$ is precompact in $C([0,\infty):\mathcal{M}_+([0,\infty]))$. This follows, by Arzel\`a-Ascoli, from the fact that the family is bounded, as well as equicontinuous due to the estimate \eqref{eq:S2E7}, which is independent of $a$. There thus exist some $G\in C([0,\infty):\mathcal{M}_+([0,\infty]))$, with $\int\!\!_{[0,\infty]}G(t,x)\dd x=\|G_0\|$ for all $t\in[0,\infty)$, and a sequence ${a}_n\rightarrow0$ such that $G_{a_n}\wsc G$, locally uniformly in $t$ on $[0,\infty)$. The left hand side of \eqref{eq:S2E9} now converges by definition to the left hand side of \eqref{eq:regex}, with a time independent test function $\varphi$. We also have that $\phi_{a_n}\ast G_{a_n}(s,\cdot)\wsc G(s,\cdot)$, locally uniformly in $s$ on $[0,\infty)$, as may be seen by passing the convolution to the test functions. Therefore the right hand side of \eqref{eq:S2E9} converges to
\begin{equation}\label{eq:S2E12}\nonumber
\frac12\int_0^t\iint_{[0,\infty]^2}\frac{G(s,x)G(s,y)}{\sqrt{(x+\varepsilon)(y+\varepsilon)}}\big[\varphi(x+y)-2\varphi(x\vee y)+\varphi(|x-y|)\big]\dd x\dd y\,\dd s,
\end{equation}
which is equal to the second term on the right hand side of \eqref{eq:regex}, with $\varphi$ time independent.
Notice that the integral over the diagonal $\{x=y\geq0\}$ is now in general no longer negligible, since $G(t,\cdot)$ could contain atoms.

Note finally that if we had chosen time dependent test functions $\varphi\in C^1([0,\infty):C([0,\infty]))$, then there would have been an additional linear term in both \eqref{eq:S2E7} and \eqref{eq:S2E9}. The term in \eqref{eq:S2E9} would then without problems have converged to the first term on the right hand side of \eqref{eq:regex}, confirming that the function $G$ that we obtained indeed satisfies \eqref{eq:regex} for all $t\in[0,\infty)$ and all $\varphi\in C^1([0,\infty):C([0,\infty]))$.
\end{proof}
Notice that Lemma \ref{lm:ex;regularizedexistence} does not rule out the possibility that $\int\!\!_{\{\infty\}}G(t,x)\dd x>0$ for some $t\in(0,\infty)$. The following result will show that this is not the case. Moreover, the tightness result that we obtain is uniform in $\varepsilon\in(0,1)$ and will allow us to take the limit $\varepsilon\rightarrow0$ in \eqref{eq:regex}.
\begin{lm}\label{lm:ex;uniformtightnessreg}
Let $\varepsilon\in(0,1)$ and $G_0\in\mathcal{M}_+([0,\infty))$ be arbitrary and suppose that $G\in C([0,\infty):\mathcal{M}_+([0,\infty]))$ satisfies \eqref{eq:regex} for all $t\in[0,\infty)$ and all $\varphi\in C^1([0,\infty):C([0,\infty]))$. Then given $\eta,R\in(0,\infty)$ the following holds.
\begin{equation}\label{eq:untghtreg}
\int_{[0,R/\eta]}G(t,x)\dd x\geq(1-\eta)\int_{[0,R]}G_0(x)\dd x\text{ for all }t\in[0,\infty)
\end{equation}
\end{lm}
\begin{proof}
As $G\geq0$, we may restrict ourselves to $\eta\in(0,1)$. Fix $K\in(0,1/R)$ and set $\varphi(x)=(1-Kx)_+$. Then $\Delta_\varphi\geq0$ since $\varphi$ is convex (cf.~\eqref{eq:def;Delta}), and thus
\begin{equation}\label{eq:untghtreg_1}
\begin{split}
\int_{[0,1/K]}G(t,x)\dd x&\geq\int_{[0,\infty)}\varphi(x)G(t,x)\dd x\geq\int_{[0,\infty)}\varphi(x)G_0(x)\dd x\\&\geq(1-KR)\int_{[0,R]}G_0(x)\dd x\text{ for all }t\in[0,\infty).
\end{split}
\end{equation}
Choosing $K=\eta/R$ finishes the proof.
\end{proof}
The following is an immediate consequence.
\begin{cor}\label{cor:nomassatinfty}
Let $\varepsilon\in(0,1)$ and $G_0\in\mathcal{M}_+([0,\infty))$ be arbitrary and suppose that $G\in C([0,\infty):\mathcal{M}_+([0,\infty]))$ satisfies \eqref{eq:regex} for all $t\in[0,\infty)$ and all $\varphi\in C^1([0,\infty):C([0,\infty]))$. Then actually $G\in C([0,\infty):\mathcal{M}_+([0,\infty)))$, i.e.~$\int\!\!_{\{\infty\}}G(t,x)\dd x\equiv0$ for all $t\in[0,\infty)$. Moreover, for any $G_0\in\mathcal{M}_+([0,\infty))$ the collection $\mathcal{G}$ of functions in $C([0,\infty):\mathcal{M}_+([0,\infty)))$ that for some $\varepsilon\in(0,1)$ satisfy \eqref{eq:regex} for all $t\in[0,\infty)$ and all $\varphi\in C^1([0,\infty):C([0,\infty]))$ is uniformly tight, i.e.~for any $\eta\in(0,1)$ small there exists some $R\in(0,\infty)$, depending only on $\eta$ and $G_0$, such that $\int\!\!_{[R,\infty)}G(t,x)\dd x<\eta$ for all $t\in[0,\infty)$.
\end{cor}
\begin{proof}
Using $\varphi\equiv1$ in \eqref{eq:regex} we have $\int\!\!_{[0,\infty]}G(t,x)\dd x=\|G_0\|$ for all $t\in[0,\infty)$, while for any $t\in[0,\infty)$ if follows from Lemma \ref{lm:ex;uniformtightnessreg} that
\begin{equation}\nonumber
\|G_0\|\geq\int_{[0,\infty)}G(t,x)\dd x\geq\int_{[0,R^2]}G(t,x)\dd x\geq\frac{R-1}{R}\int_{[0,R]}G_0(x)\dd x.
\end{equation}
Since the right hand side now tends to $\|G_0\|$ as $R\rightarrow\infty$, we obtain that $\int\!\!_{[0,\infty)}G(t,x)\dd x=\|G_0\|$ for all $t\in[0,\infty)$. The first claim is then immediate and the second claim follows as \eqref{eq:untghtreg} is independent of $\varepsilon$.
\end{proof}
We are now ready to prove our first main result.

\begin{proof}[Proof of Theorem \ref{tm:existence}]
From Lemma \ref{lm:ex;regularizedexistence} and Corollary \ref{cor:nomassatinfty} we know that for any $\varepsilon\in(0,1)$ there exists at least one function $G_\varepsilon\in C([0,\infty):\mathcal{M}_+([0,\infty)))$ that satisfies \eqref{eq:regex} for all $t\in[0,\infty)$ and all $\varphi\in C^1([0,\infty):C([0,\infty]))$. We will next consider the collection $\mathcal{G}=\{G_\varepsilon\}_{\varepsilon\in(0,1)}$ of all of these functions.

Choosing $\varphi\in C^1([0,\infty])$ arbitrarily and $t_1,t_2\in[0,\infty)$, it follows from \eqref{eq:regex} that for any $G_\varepsilon\in\mathcal{G}$ we have
\begin{equation}\nonumber
\bigg|\int_{[0,\infty]}\varphi(x)G_\varepsilon(t_2,x)\dd x-\int_{[0,\infty]}\varphi(x)G_\varepsilon(t_1,x)\dd x\bigg|\leq|t_2-t_1|C_{\varphi,\varepsilon}\|G_0\|^2,
\end{equation}
where
\begin{equation}\nonumber
C_{\varphi,\varepsilon}={\textstyle\sup_{\{x\geq y\geq0\}}}\bigg|\frac{\Delta_\varphi(x,y)}{\sqrt{(x+\varepsilon)(y+\varepsilon)}}\bigg|\leq2\|\varphi'\|_\infty,
\end{equation}
due to \eqref{eq:S1E19}. Since $C^1([0,\infty])$ is dense in $C([0,\infty])$ it then follows that for any $\varphi\in C([0,\infty])$ the family of mappings $t\mapsto\int\!\!_{[0,\infty)}\varphi(x)G_\varepsilon(t,x)\dd x$, with $G_\varepsilon\in\mathcal{G}$, is uniformly continuous on $[0,\infty)$. It then follows, by Arzel\`a-Ascoli, that $\mathcal{G}$ is precompact in $C([0,\infty):\mathcal{M}_+([0,\infty]))$, and using also Corollary \ref{cor:nomassatinfty} there thus exist some function $G\in C([0,\infty):\mathcal{M}_+([0,\infty)))$, satisfying $\|G(t,\cdot)\|=\|G_0\|$ for all $t\in[0,\infty)$, and a sequence $\varepsilon_n\rightarrow0$ such that $G_{\varepsilon_n}(t,\cdot)\wsc G(t,\cdot)$ locally uniformly in $t$ on $[0,\infty)$.

By definition of weak-\st convergence it is now trivial that the left hand side and the first term on the right hand side of \eqref{eq:regex} converge to the corresponding terms in \eqref{eq:def:solution}. Convergence of the remaining term follows if we remark that $G_{\varepsilon_n}(s,\cdot)\times G_{\varepsilon_n}(s,\cdot)\wsc G(s,\cdot)\times G(s,\cdot)$ in $\mathcal{M}_+([0,\infty)^2)$, locally uniformly in $s\in[0,\infty)$, and that $((x+\varepsilon)(y+\varepsilon))^{-1/2}\Delta_\varphi(x,y)\rightarrow (xy)^{-1/2}\Delta_\varphi(x,y)$, uniformly in $\{x\geq y\geq 0\}$ for any $\varphi\in C_c^1([0,\infty))$.
\end{proof}

\section{Properties of weak solutions}\label{sec:properties}
\subsection{Basic properties}
This subsection is devoted to proving some basic properties of weak solutions to \eqref{eq:def:solution} in the sense of Definition \ref{def:notionofsolution}. We start with a tightness result that is similar to the one obtained in Section \ref{sec:existence} (cf.~Lemma \ref{lm:ex;uniformtightnessreg}).
\begin{pr}[Tightness]\label{pr:tightnessofsolutions}
Let $G\in C([0,\infty):\mathcal{M}_+([0,\infty)))$ be a weak solution to \eqref{eq:def:solution}. Then given $\eta,R\in(0,\infty)$ the following holds.
\begin{equation}\label{eq:tightnessofsolutions}\nonumber
\int_{[0,R/\eta]}G(t,x)\dd x\geq(1-\eta)\int_{[0,R]}G(0,x)\dd x\text{ for all }t\in[0,\infty)
\end{equation}
\end{pr}
\begin{proof}
The proof is the same as the one of Lemma \ref{lm:ex;uniformtightnessreg}. The only difference is that we need to approximate the test function $\varphi(x)=(1-Kx)_+$ in the uniform topology by convex functions in $C_c^1([0,\infty))$ to get the second inequality in \eqref{eq:untghtreg_1}.
\end{proof}
We will next show that weak solutions to \eqref{eq:def:solution} satisfy two conservation laws, provided that the first moment is initially finite. In order to do this, we need the following lemma.
\begin{lm}\label{lm:moretestfunctions}
Let $G\in C([0,\infty):\mathcal{M}_+([0,\infty)))$ be a weak solution to \eqref{eq:def:solution}. Then $G$ satisfies \eqref{eq:def:solution} for all $t\in[0,\infty)$ and all $\varphi\in C^1([0,\infty):C^1([0,\infty]))$.
\end{lm}
\begin{proof}
Let $\{\zeta_n\}$ be an increasing sequence, bounded in $C_c^1([0,\infty))$, such that $\zeta_n\equiv1$ on $[0,n]$ for all $n\in\mathbb{N}$. Given then any $\varphi\in C^1([0,\infty):C^1([0,\infty]))$ we set $\varphi_n(t,x):=\varphi(t,x)\zeta_n(x)$, which we use in \eqref{eq:def:solution}. As $n\rightarrow\infty$, then $\varphi_n(t,\cdot)$ and $\partial_t\varphi_n(t,\cdot)$ converge pointwise to $\varphi(t,\cdot)$ and $\varphi_t(t,\cdot)$ respectively for all $t\in[0,\infty)$ and since $\|G(t,\cdot)\|$ is bounded, locally uniformly in $t\in[0,\infty)$, the linear terms in \eqref{eq:def:solution} converge by dominated convergence. For convergence of the quadratic term, it is finally sufficient to notice that $\Delta_{\varphi_n(s,\cdot)}\equiv\Delta_{\varphi(s,\cdot)}$ on $\{x\leq\frac{n}{2}\}\cap\{x\geq y\geq0\}$, and that $|\Delta_{\varphi_n(s,\cdot)}(x,y)|\leq C(y\wedge1)$ for all $y\in[0,x]$ and all $s\in[0,\infty)$, with $C$ independent of $n$ (cf.~\eqref{eq:S1E19}). Therefore, as $n\rightarrow\infty$ we have $(xy)^{-1/2}\Delta_{\varphi_n(s,\cdot)}(x,y)\rightarrow(xy)^{-1/2}\Delta_{\varphi(s,\cdot)}(x,y)$ for all $(x,y)\in\{x\geq y\geq 0\}$ and all $s\in[0,\infty)$ and since $|(xy)^{-1/2}\Delta_{\varphi_n(s,\cdot)}(x,y)|\leq C$, the quadratic term in \eqref{eq:def:solution} converges by dominated convergence.
\end{proof}
\begin{pr}[Conservation laws]\label{pr:conservationlaws}
Let $G\in C([0,\infty):\mathcal{M}_+([0,\infty)))$ be a weak solution to \eqref{eq:def:solution}. The total measure of $G$ is then constant in $t\in[0,\infty)$ and, furthermore, if the first moment is initially finite, then
\begin{equation}\label{eq:conservationoffirstmoment}
\int_{[0,\infty)}xG(t,x)\dd x=\int_{[0,\infty)}xG(0,x)\dd x\text{ for all }t\in[0,\infty).
\end{equation}
\end{pr}
\begin{proof}
Conservation of the total measure follows immediately from Lemma \ref{lm:moretestfunctions}, choosing the test function $\varphi\equiv1$. To then prove conservation of finite first moments, suppose that $\int\!\!_{[0,\infty)}xG(0,x)\dd x=:E\in[0,\infty)$. For any $\varepsilon\in(0,1)$ we now use the function $\varphi_\varepsilon\in C^\infty([0,\infty])$, defined as $\varphi_\varepsilon(x)=x/(1+\varepsilon x)$, in \eqref{eq:def:solution}. Since $\varphi_\varepsilon$ is concave, it follows that $\Delta_{\varphi_\varepsilon}\leq0$ (cf.~\eqref{eq:def;Delta}) and we then have
\begin{equation}\nonumber
\int_{[0,\infty)}\varphi_\varepsilon(x)G(t,x)\dd x\leq\int_{[0,\infty)}\varphi_\varepsilon(x)G(0,x)\dd x\leq E\text{ for all }t\in[0,\infty).
\end{equation}
Taking the limit $\varepsilon\rightarrow0$, it follows by monotone convergence that the first moment of $G(t,\cdot)$ is bounded by $E$ for all $t\in[0,\infty)$. We furthermore notice that $-\Delta_{\varphi_\varepsilon}(x,y)\leq 2\varepsilon y^2$ on $\{x\geq y\geq 0\}$, so we may obtain
\begin{equation}\label{eq:S3E5}
\begin{split}
&\bigg|\sint\frac{G(s,x)G(s,y)}{\sqrt{xy}}\Delta_{\varphi_\varepsilon}(x,y)\dd x\dd y\bigg|
\leq\sint G(s,x)G(s,y)\frac{2\varepsilon y^2}{\sqrt{xy}}\dd x\dd y\\
&\indent\leq2\varepsilon\iint_{[0,\infty)^2}G(s,x)\,yG(s,y)\,\dd x\dd y
\leq2\varepsilon\|G(0,\cdot)\|E,
\end{split}
\end{equation}
whereby the left hand side of \eqref{eq:S3E5} vanishes as $\varepsilon\rightarrow0$. Taking finally the limit $\varepsilon\rightarrow0$ in \eqref{eq:def:solution}, with $\varphi=\varphi_\varepsilon$, we find, using monotone convergence in the linear terms, that indeed \eqref{eq:conservationoffirstmoment} holds.
\end{proof}
Now we prove that for any weak solution to \eqref{eq:def:solution} the measure of the origin is monotone.
\begin{pr}[Monotonicity of the measure of $\{0\}$]\label{pr:monotonicityoforigin}
Let $G\in C([0,\infty):\linebreak\mathcal{M}_+([0,\infty)))$ be a weak solution to \eqref{eq:def:solution}. Then the mapping $t\mapsto\int\!\!_{\{0\}}G(t,x)\dd x$ is nondecreasing on $[0,\infty)$.
\end{pr}
\begin{proof}
For $\mu\in\mathcal{M}_+([0,\infty))$ we have that $\int\!\!_{\{0\}}\mu(x)\dd x=\inf_\varphi\int\!\!_{[0,\infty)}\varphi(x)\mu(x)\dd x$, where we take the infimum over all convex functions $\varphi\in C^1([0,\infty])$ that satisfy $\varphi(0)=1$ and $\lim_{x\rightarrow\infty}\varphi(x)=0$, which is a consequence of outer regularity of Radon measures in $\mathbb{R}$. For any of these test functions we now have $\Delta_\varphi\geq0$, by convexity (cf.~\eqref{eq:def;Delta}), and the mappings $t\mapsto\int\!\!_{[0,\infty)}\varphi(x)G(t,x)\dd x$ are nondecreasing on $[0,\infty)$ (cf.~Remark \ref{rk:remark1_16}). We conclude that $t\mapsto\int\!\!_{\{0\}}G(t,x)\dd x$ is nondecreasing on $[0,\infty)$, since it is the infimum of a collection of nondecreasing functions.
\end{proof}
As a consequence of the previous propositions we obtain the following uniqueness result.
\begin{cor}[Trivial solutions]\label{cor:trivialsolutions}
Let $G\in C([0,\infty):\mathcal{M}_+([0,\infty)))$ be a weak solution to \eqref{eq:def:solution} that satisfies $G(0,\cdot)\equiv M\delta_0$ for some $M\in[0,\infty)$. Then $G$ is a trivial solution in the sense of Remark \ref{rk:deftrivsol}, i.e.~$G(t,\cdot)\equiv M\delta_0$ for all $t\in[0,\infty)$.
\end{cor}
\begin{proof}
Using Propositions \ref{pr:conservationlaws} and \ref{pr:monotonicityoforigin} we find that
\begin{equation}\nonumber
0\leq\int_{(0,\infty)}G(t,x)\dd x=M-\int_{\{0\}}G(t,x)\dd x\leq0\text{ for all }t\in[0,\infty),
\end{equation}
implying that $G\equiv0$ on $[0,\infty)\times(0,\infty)$ and hence $G$ is trivial.
\end{proof}
\begin{rk}
Note that it is a straightforward computation to check that the function $G(t,\cdot)\equiv M\delta_0$ is weak solution to \eqref{eq:def:solution}, using the observation \eqref{eq:S1E20}.
\end{rk}
The following result means that, given a weak solution to \eqref{eq:def:solution}, the part of the solution concentrated at the origin does not interact with the remaining part of the solution.
\begin{pr}\label{pr:invarianceunderaddingdirac}
Let $G_1\in C([0,\infty):\mathcal{M}_+([0,\infty)))$ be a weak solution to \eqref{eq:def:solution} and let $M\in[0,\infty)$. Then the function $G_2\in C([0,\infty):\mathcal{M}_+([0,\infty)))$ that is defined as the sum $G_2(t,x)\dd x=M\delta_0(x)\dd x+G_1(t,x)\dd x$ is another weak solution to \eqref{eq:def:solution}.
\end{pr}
\begin{proof}
Since both $G_1$ and $G_0(\cdot,x)\dd x\equiv M\delta_0(x)\dd x$ are weak solutions to \eqref{eq:def:solution} (cf.~Corollary \ref{cor:trivialsolutions}), it is sufficient to note that the cross terms arising from the quadratic term in \eqref{eq:def:solution}, i.e.~the terms containing the $G_0\times G_1$ and $G_1\times G_0$, vanish; indeed this is true, since $(xy)^{-1/2}\Delta_\varphi(x,y)=0$ on $\{x\geq y=0\}=:L$ (cf.~\eqref{eq:S1E20}), and $L$ contains the support of both product measures $G_0\times G_1$ and $G_1\times G_0$.
\end{proof}
\begin{rk}\label{rk:remark3_8}
Notice that the same argument, combined with Proposition \ref{pr:monotonicityoforigin}, implies that if $G_1\in C([0,\infty):\mathcal{M}_+([0,\infty)))$ is a weak solution to \eqref{eq:def:solution} that satisfies $\int\!\!_{\{0\}}G_1(0,x)\dd x=M\in[0,\infty)$, then all functions $G_2(t,x)\dd x=m\delta_0(x)\dd x+G_1(t,x)\dd x$ with $m\in[-M,\infty)$ are also weak solutions to \eqref{eq:def:solution}.
\end{rk}
\begin{rk}\label{rk:remark3_9}
Note that the result of Proposition \ref{pr:invarianceunderaddingdirac} can be intuitively understood using the particle interpretation given in Remark \ref{rk:particleinterpretation}, since both outcomes of an interaction of the pair of particles $\{x,0\}$, with $x\in[0,\infty)$, are again $\{x,0\}$.
\end{rk}
As a last result in this subsection, we show that any weak solution to \eqref{eq:def:solution} converges weakly-\st to a uniquely defined trivial solution as $t\rightarrow\infty$.
\begin{pr}[Long time behaviour]\label{pr:longtimebehaviour}
Let $G\in C([0,\infty):\mathcal{M}_+([0,\infty)))$ be a weak solution to \eqref{eq:def:solution}. Then $G(t,\cdot)\wsc \|G(0,\cdot)\|\delta_0$ as $t\rightarrow\infty$.
\end{pr}
\begin{proof}
We restrict ourselves to nontrivial solutions, since the stationary trivial solutions satisfy the claim automatically. Let $\varphi\in C^2([0,\infty])$ be a convex function, satisfying $\varphi(0)=1$, $\lim_{x\rightarrow\infty}\varphi(x)=0$ and $\varphi''>0$ on $[0,\infty)$. Since $\Delta_\varphi\geq0$, the mapping $t\mapsto H_\varphi(t):=\int\!\!_{[0,\infty)}\varphi(x)G(t,x)\dd x$ is then nondecreasing, and as it is further bounded by $M:=\|G(0,\cdot)\|\in(0,\infty)$, due to conservation of the total measure (cf.~Proposition \ref{pr:conservationlaws}). Therefore there exists a limit $\lim_{t\rightarrow\infty}H_\varphi(t)=:L\in[0,M]$. Now, since $\varphi<1$ on $(0,\infty)$ by convexity, we note that our claim holds if and only if $L=M$. We will prove this by contradiction and thereto suppose that $M-L=:\varepsilon>0$.

As a first step we will show that there exist $0<\delta<Q<\infty$ such that
\begin{equation}\label{eq:ltb_5}
\int_{(\delta,Q)}G(s,x)\dd x>\frac{\varepsilon}{3}\text{ for all }s\in[0,\infty).
\end{equation}
To that end we fix constants $\alpha\in[\frac23,1)$ and $\beta\in[\frac{L}{M}(1-\alpha),1-\alpha)$. For any $\delta\in[0,\infty)$ it then holds that
\begin{equation}\label{eq:ltb_1}\nonumber
\int_{[0,\delta]}G(s,x)\dd x\leq\int_{[0,\infty)}\frac{\varphi(x)}{\varphi(\delta)}G(s,x)\dd x\leq\frac{L}{\varphi(\delta)}\leq\frac{\alpha L+\beta M}{\varphi(\delta)}\text{ for all }s\in[0,\infty).
\end{equation}
Fixing now any $\delta\in(0,\varphi^{-1}(\alpha+\beta)]$ it follows that
\begin{equation}\label{eq:ltb_2}
\int_{(\delta,\infty)}G(s,x)\dd x\geq M-\frac{\alpha L+\beta M}{\alpha+\beta}>\frac{2\varepsilon}{3}\text{ for all }s\in[0,\infty).
\end{equation}
Next we fix $R\in(0,\infty)$ such that
\begin{equation}\label{eq:ltb_3}\nonumber
\int_{[0,R)}G(0,x)\dd x\geq\frac{L+5M}{6},
\end{equation}
and applying Proposition \ref{pr:tightnessofsolutions} with $\eta\in(0,\frac{M-L}{L+5M})$, we obtain
\begin{equation}\label{eq:ltb_4}
\int_{[0,\frac{R}{\eta})}G(s,x)\dd x>M-\frac\varepsilon3\text{ for all }s\in[0,\infty).
\end{equation}
Combining then \eqref{eq:ltb_2} and \eqref{eq:ltb_4}, we find that \eqref{eq:ltb_5} holds with $\delta$ as before and $Q:=\frac{R}{\eta}\vee2\delta$.

Using now \eqref{eq:ltb_5} in \eqref{eq:def:solution} and the fact that $G(0,\cdot)\geq0$, we obtain for all $t\in[0,\infty)$ that
\begin{equation}\label{eq:pr;pws_4_7}\nonumber
\begin{split}
H_\varphi(t)&\geq\int_0^t\sint\frac{G(s,x)G(s,y)}{\sqrt{xy}}\Delta_\varphi(x,y)\dd x\dd y\,\dd s\\
&\geq\Theta_\varphi(\delta,Q)\int_0^t\Big({\textstyle\int_{(\delta,Q)}G(s,x)\dd x}\Big)\!\Big.^2\dd s>\tfrac19\Theta_\varphi(\delta,Q)\varepsilon^2t,
\end{split}
\end{equation}
with $\Theta_\varphi(\delta,Q)=\frac12\inf_{\{Q>x\geq y>\delta\}}(xy)^{-1/2}\Delta_\varphi(x,y)\geq\frac{\delta^{3/2}}{2Q^{1/2}}\inf_{\xi\in[0,2Q]}\varphi''(\xi)>0$, whence $H_\varphi(t)\rightarrow\infty$ as $t\rightarrow\infty$, which is a contradiction. Therefore $L=M$ and the proposition follows.
\end{proof}
\subsection{Instantaneous Dirac formation for nontrivial data}
\label{sec:instantdirac}
In the previous subsection we have seen that any solution converges in the sense of measures to a Dirac measure at zero and that the measure of the origin is monotonically increasing. In this subsection we show that this monotonicity is strict as long as the measure of $(0,\infty)$ is positive (cf.~Proposition \ref{pr:instantdirac}). The following lemma will play a key role.
\begin{lm}\label{lm:id;keylemma}
Let $G\in C([0,\infty):\mathcal{M}_+([0,\infty)))$ be a weak solution to \eqref{eq:def:solution}. For any $c\in(0,\infty)$ and all $t\in[0,\infty)$ the following holds.
\begin{equation}\label{eq:lm;idkey}
\int_{[0,c]}G(t,x)\dd x\geq\frac{1}{2c}\int_0^t\iint_{[c,\infty)^2}\frac{G(s,x)G(s,y)}{\sqrt{xy}}\big(c-|x-y|\big)_+\dd x\dd y\ \dd s
\end{equation}
\end{lm}
\begin{proof}
Fix some $c\in(0,\infty)$ and let $\varphi\in C_c^1([0,\infty))$ be a nonnegative convex function that satisfies $\varphi(0)\leq c$ and $\varphi\equiv0$ on $[c,\infty)$. As $G(0,\cdot)$ is nonnegative by definition we may drop the second term on the left hand side of \eqref{eq:def:solution} to obtain for all $t\in[0,\infty)$ that
\begin{equation}\label{eq:lm;idkey_1}
\int_{[0,\infty)}\varphi(x)G(t,x)\dd x\geq\int_0^t\sint\frac{G(s,x)G(s,y)}{\sqrt{xy}}\Delta_\varphi(x,y)\dd x\dd y\ \dd s.
\end{equation}
Using now the facts that $\Delta_\varphi\geq0$ on $\{x\geq y\geq0\}$ and $\Delta_\varphi(x,y)=\varphi(x-y)$ for $x\geq c$, we obtain that $\Delta_\varphi(x,y)\geq\varphi(x-y)\mathbf{1}_{[c,\infty]}(x)$ on $\{x\geq y\geq 0\}$. On the other hand, we have that $\varphi\leq c\mathbf{1}_{[0,c]}$ on $[0,\infty)$. Combining these inequalities with a symmetrization argument on the right hand side of \eqref{eq:lm;idkey_1}, we find that
\begin{equation}\label{eq:lm;idkey_2}\nonumber
\int_{[0,c]}c\,G(t,x)\dd x\geq\int_0^t\frac12\iint_{[c,\infty)^2}\frac{G(s,x)G(s,y)}{\sqrt{xy}}\varphi(|x-y|)\dd x\dd y\ \dd s,
\end{equation}
for all $t\in[0,\infty)$ and all nonnegative convex $\varphi\in C_c^1([0,\infty))$ with $\varphi\leq c\mathbf{1}_{[0,c]}$. As the supremum over all of these test functions is $\sup_\varphi\varphi(x)=(c-x)_+$, it then follows by monotone convergence that \eqref{eq:lm;idkey} holds for all $t\in[0,\infty)$.
\end{proof}
For the convenience of the reader we further state the following result, which we will use frequently. It follows from either H\"older's or Jensen's inequality; the proof is omitted.
\begin{lm}[{see e.g.~\cite[Thm.~13]{HLP52}}]\label{lm:HoelderJensen}
For $n\in\mathbb{N}$, let $(a_j)_{j=1}^n$ and $(b_j)_{j=1}^n$ be nonnegative sequences. Supposing that $a_j>0$ for all $j\in\{1,\ldots,n\}$, then
\begin{equation}\label{eq:HJ1}\nonumber
\textstyle\sum_{j=1}^na_jb_j^2\geq\Big(\sum_{j=1}^na_j^{-1}\Big)\!\Big.^{-1}\Big(\sum_{j=1}^nb_j\Big)\!\Big.^2,
\end{equation}
and, in particular,
\begin{equation}\label{eq:HJ2}
\textstyle\sum_{j=1}^nb_j^2\geq\frac1n\Big(\sum_{j=1}^nb_j\Big)\!\Big.^2.
\end{equation}
\end{lm}
The proof of Proposition \ref{pr:instantdirac} will use a contradiction argument. First we show that, for $r$ small, $\int\!\!_{(0,r]}G(t,x)\dd x$ is roughly bounded from above by $\sqrt{r}$ (cf.~Lemma \ref{lm:averagedupperbound}). On the other hand we will also prove that for any $\tau>0$ we have $\int\!\!_{[0,r]}G(t,x)\dd x\geq C r^\alpha$ with $\alpha<\frac12$ for any $t\geq \tau$ (cf.~Proposition \ref{pr:finalconcentrationresult}). This will give a contradiction unless $\int\!\!_{\{0\}}G(t,x)\dd x>0$ for all $t>0$.

We split the proof in several steps.
\subsubsection{An upper bound for the measure of $G$ in $(0,r]$.}
\begin{lm}\label{lm:averagedupperbound}
Let $G\in C([0,\infty):\mathcal{M}_+([0,\infty)))$ be a weak solution to \eqref{eq:def:solution}. Then for all $r\in[0,\infty)$, all $t_1\in[0,\infty)$ and any $t_2\in[t_1,\infty)$ we have
\begin{equation}\label{eq:aub}
\int_{t_1}^{t_2}\int_{(0,r]}G(s,x)\dd x\,\dd s\leq A\sqrt{r},\text{ with }A=12\sqrt{(t_2-t_1)\|G(0,\cdot)\|}.
\end{equation}
\end{lm}
\begin{proof}
Fixing $t_1\in[0,\infty)$ and $t_2\in[t_1,\infty)$, then applying Lemma \ref{lm:id;keylemma} and using that $|x-y|\leq\frac{c}{2}$ on $(c,\frac32c]\times(c,\frac32c]$, we obtain
\begin{equation}\label{eq:aub_1}
\frac14\int_{t_1}^{t_2}\Big({\textstyle\int_{(c,\frac32c]}\frac{G(s,x)}{\sqrt{x}}\dd x}\Big)\!\Big.^2\dd s\leq\|G(0,\cdot)\|\text{ for all }c\in(0,\infty).
\end{equation}
Note further for any fixed $r\in(0,\infty)$ that
\begin{equation}\label{eq:aub_2}
\int_{(0,r]}G(s,x)\dd x\leq\sqrt{r}\ \sum_{j=0}^\infty\big(\tfrac23\big)^{j/2}\,\Big(\textstyle\int_{((\frac23)^{j+1}r,(\frac23)^jr]}\frac{G(s,x)}{\sqrt{x}}\dd x\Big).
\end{equation}
We then integrate \eqref{eq:aub_2} with respect to $s$ from $t_1$ to $t_2$ and use Cauchy-Schwarz in each of the terms in the sum on the right. Applying finally \eqref{eq:aub_1} with $c=(\frac23)^{j+1}r$, we obtain \eqref{eq:aub}.
\end{proof}
\begin{rk}\label{rk:3_14}
The estimate in the previous lemma has an interesting interpretation. Even though we only study solutions with finite total measure in this paper, the measure
\begin{equation}\label{eq:3_12}
G(t,x)\dd x=\big(\tfrac{\pi^2}{12}t\delta_0(x)+x^{-1/2}\big)\dd x
\end{equation}
satisfies \eqref{eq:def:solution} for all $\varphi\in C^1([0,\infty):C_c^1([0,\infty)))$ and all $t\in[0,\infty)$. This is a consequence of the fact that
\begin{equation}\label{eq:3_13}\nonumber
\iint_{\{x>y\geq0\}}\frac{\Delta_\varphi(x,y)}{xy}\dd x\dd y=\frac{\varphi(0)\,\pi^2}{12}\text{ for all }\varphi\in C_c^1([0,\infty)),
\end{equation}
which can be obtained by replacing $(xy)^{-1}$ with $((x+\varepsilon)(y+\varepsilon))^{-1}$ and taking the limit $\varepsilon\rightarrow0^+$.

The power law $x^{-1/2}$ in \eqref{eq:3_12} can be interpreted physically as a thermal equilibrium between particles with $x>0$ and the particles in the condensate. Analogous power laws have been obtained for the quantum Boltzmann equation (cf.~\cite{S10}). It is worth noticing that Lemma \ref{lm:averagedupperbound} provides an upper estimate with the same behaviour for arbitrary weak solutions to \eqref{eq:def:solution}. Another estimate of this kind can be found in Lemma \ref{lm:ssb;finitemoments}.
\end{rk}
\subsubsection{A lower bound on $\int\!\!_{[0,r]}G(t,x)\dd x$, linear in $r$ for $t>0$}
\begin{lm}\label{lm:fourthconcentrationresult}
Let $G\in C([0,\infty):\mathcal{M}_+([0,\infty)))$ be a nontrivial weak solution to \eqref{eq:def:solution}. There then exist constants $B_1,R_1,T_1\in(0,\infty)$, depending only on the initial measure $G(0,\cdot)\in\mathcal{M}_+([0,\infty))$, such that
\begin{equation}\label{eq:fourconres}
\int_{[0,r]}G(t,x)\dd x\geq B_1rt\text{ for all }r\in[0,R_1]\text{ and all }t\in[0,T_1].
\end{equation}
\end{lm}
\begin{proof}
Since $G$ is nontrivial we have $\int\!\!_{(0,\infty)}G(0,x)\dd x=:m_0\in(0,\infty)$, and we can define constants $R_1,R'\in(0,\infty)$ as follows.
\begin{align}\nonumber
R_1&\textstyle:=\frac12\inf\big\{r\in[0,\infty)\,|\,\int_{(0,r]}G(0,x)\dd x>\frac14m_0\big\}\\\nonumber
R'&\textstyle:=2\inf\big\{r\in[0,\infty)\,|\,\int_{(0,r]}G(0,x)\dd x\geq\frac34m_0\big\}
\end{align}
In particular this implies that $\int\!\!_{(2R_1,R'/2]}G(0,x)\dd x\geq\frac12m_0$. We then choose a test function $\varphi\in C_c^\infty([0,\infty))$ satisfying $\varphi\equiv0$ on $[0,\infty)\setminus(R_1,R']$, $\varphi\equiv 1$ on $(2R_1,\frac12R']$, $0\leq\varphi\leq 1$ and $|\varphi'|\leq C_1$ with $C_1$ depending only on $R_1$ and $R'$. Using then $\varphi$ in \eqref{eq:def:solution}, we see that $\int\!\!_{(R_1,R']}G(t,x)\dd x\geq\frac12m_0-C m_0^2 t$ with $C$ depending only on $R_1$ and $R'$. Therefore, taking $T_1\in(0,\infty)$ small enough we obtain
\begin{equation}\label{eq:fourconres_3}\nonumber
\int_{(R_1,R']}G(s,x)\dd x\geq\tfrac14m_0\text{ for all }s\in[0,T].
\end{equation}
Let now $r\in(0,R_1]$ be arbitrary and let $n\in\mathbb{N}$ be such that $2^nr\in(R',2R']$. Applying then Lemma \ref{lm:id;keylemma} with $c=r$ we obtain that for all $t\in[0,\infty)$
\begin{equation}\label{eq:fourconres_4}
\begin{split}
&\int_{[0,r]}G(t,x)\dd x\geq\frac{1}{2r}\!\int_0^t\iint_{[r,\infty)^2}\frac{G(s,x)G(s,y)}{\sqrt{xy}}\big(r-|x-y|\big)_+\dd x\dd y\,\dd s\\
&\indent\!\!\!\geq\frac{1}{2r}\!\int_0^t\bigg[\sum_{i=3}^{2^{n+1}}\iint_{(\frac{r}{2}(i-1),\frac{r}{2}i]^2}\frac{G(s,x)G(s,y)}{\sqrt{xy}}\big(r-|x-y|\big)_+\dd x\dd y\bigg]\dd s.
\end{split}
\end{equation}
We now note that $(r-|x-y|)_+\geq\frac12r$ for $x,y\in(\frac{r}{2}(i-1),\frac{r}{2}i]$ with $i\in\mathbb{N}$ and that $(xy)^{-1/2}\geq(2^jr)^{-1}$ for $x,y\in(2^{j-1}r,2^jr]$ with $j\in\{1,\ldots,n\}$. Using these estimates we can then bound the term between square brackets from below on the right hand side of \eqref{eq:fourconres_4} by 
\begin{equation}\label{eq:fourconres_4b}
\sum_{i=3}^{2^{n+1}}\big(\cdots\big)\geq\sum_{j=1}^n\frac1{2^{j+1}}\sum_{i=2^j+1}^{2^{j+1}}\Big({\textstyle\int_{(\frac{r}{2}(i-1),\frac{r}{2}i]}G(s,x)\dd x}\Big)\!\Big.^2.
\end{equation}
By \eqref{eq:HJ2} we now see that
\begin{equation}\label{eq:fourconres_5}
{\sum_{i=2^j+1}^{2^{j+1}}\Big({\textstyle\int_{(\frac{r}{2}(i-1),\frac{r}{2}i]}G(s,x)\dd x}\Big)\!\Big.^2}\geq\frac1{2^j}\Big({\textstyle\int_{(2^{j-1}r,2^jr]}G(s,x)\dd x}\Big)\!\Big.^2,
\end{equation}
so combining \eqref{eq:fourconres_4}, \eqref{eq:fourconres_4b} and \eqref{eq:fourconres_5} we obtain that
\begin{equation}\label{eq:fourconres_6}
\begin{split}
&\int_{[0,r]}G(t,x)\dd x\geq\frac{1}{2r}\!\int_0^t\bigg[\frac12\sum_{j=1}^n\frac1{4^j}\Big({\textstyle\int_{(2^{j-1}r,2^jr]}G(s,x)\dd x}\Big)\!\Big.^2\bigg]\dd s\\
&\indent\geq\frac{1}{4r}\Big({\textstyle\sum_{j=1}^n4^j}\Big)\!\Big.^{-1}\!\!\int_0^t\Big({\textstyle\int_{(r,2^nr]}G(s,x)\dd x}\Big)\!\Big.^2\dd s\text{ for all }t\in[0,\infty),
\end{split}
\end{equation}
where for the last inequality we have used Lemma \ref{lm:HoelderJensen} with $a_j=4^{-j}$. Since finally $(R_1,R']\subset(r,2^nr]$ and since the prefactor in the right hand side of \eqref{eq:fourconres_6} can be estimated from below by $\frac3{16}r^{-1}(2^n)^{-2}\geq\frac3{64}(R')^{-2}\,r$, the claim follows.
\end{proof}
Combining Lemma \ref{lm:fourthconcentrationresult} and Proposition \ref{pr:tightnessofsolutions} we then obtain the following.
\begin{lm}\label{lm:fifthconcentrationresult}
Let $G\in C([0,\infty):\mathcal{M}_+([0,\infty)))$ be a nonzero weak solution to \eqref{eq:def:solution}. For all $\tau_1\in(0,1)$ there then exist constants $B_2,R_2\in(0,\infty)$, depending only on $G(0,\cdot)$, and some $T_2\in(0,\tau_1]$ such that
\begin{equation}\label{eq:fifconres}
\int_{[0,r]}G(t,x)\dd x\geq B_2T_2\,r\text{ for all }r\in[0,R_2]\text{ and all }t\in[T_2,\infty).
\end{equation}
\end{lm}
\begin{proof}
Without loss of generality we can restrict ourselves to nontrivial solutions. By Lemma \ref{lm:fourthconcentrationresult} there then exist constants $B_1,R_1,T_1\in(0,\infty)$, depending only on $G(0,\cdot)$, such that \eqref{eq:fourconres} holds. Fixing next $T_2=T_1\wedge\tau_1$, then applying Proposition \ref{pr:tightnessofsolutions} with $\eta=\frac12$ implies that \eqref{eq:fifconres} holds with $B_2=\frac14B_1$ and $R_2=2R_1$.
\end{proof}
\subsubsection{The average $\frac1r\int\!\!_{[0,r]}G(t,x)\dd x$ is unbounded for all $t>0$.}
We prove the following result.
\begin{pr}\label{pr:seventhconcentrationresult}
Let $G\in C([0,\infty):\mathcal{M}_+([0,\infty)))$ be a nonzero weak solution to \eqref{eq:def:solution}, let $\tau_2\in(0,1)$ be arbitrary and let $L\in(0,\infty)$ be arbitrarily large. There then exists some $R_0\in(0,\infty)$, depending on $G$, $L$ and $\tau_2$, such that
\begin{equation}\label{eq:sevconres}
\int_{[0,R_0]}G(t,x)\dd x\geq LR_0\text{ for all }t\in[\tau_2,\infty)
\end{equation}
\end{pr}
\noindent The proof is given immediately after Lemma \ref{lm:sixthconcentrationresult} below.
\begin{lm}\label{lm:sixthconcentrationresult}
Let $G\in C([0,\infty):\mathcal{M}_+([0,\infty)))$ be a weak solution to \eqref{eq:def:solution}, suppose that for given $B,R\in(0,\infty)$ it holds that
\begin{equation}\label{eq:sixconres_a}\nonumber
\int_{[0,r]}G(t,x)\dd x\geq Br\text{ for all }r\in[0,R]\text{ and all }t\in[0,\infty),
\end{equation}
and let $\theta\in(0,\frac23)$ and $\tau_3\in(0,1)$ be arbitrary. Then at least one of the following alternatives is true.
\begin{enumerate}
\item[(a)]{There exists some $R_3\in(0,2\theta R]$ such that
\begin{equation}\label{eq:sixconres_c}
\int_{[0,R_3]}G(t,x)\dd x\geq\tfrac{c}{8\theta}R_3\text{ for all }t\in[\tau_3,\infty).
\end{equation}}
\item[(b)]{There exists a constant $C_\theta\in(0,\infty)$, depending explicitly on $\theta$, such that
\begin{equation}\label{eq:sixconres_b}\nonumber
\int_{[0,r]}G(t,x)\dd x\geq C_\theta B^2t r\log\!\big(\tfrac{\theta R}{r}\big)\!\text{ for all }r\in(0,\theta R]\!\text{ and all }t\in[0,\tau_3].
\end{equation}}
\end{enumerate}
\end{lm}
\begin{proof}
Suppose first that
\begin{equation}\label{eq:sixconres_1}
\int_{(\theta r,r]}G(s,x)\dd x>\frac{Br}{2}\text{ for all }r\in(0,R]\text{ and all }s\in[0,\tau_3].
\end{equation}
We now proceed as in the proof of Lemma \ref{lm:fourthconcentrationresult}. Fix $r\in(0,\theta R]$ arbi\-trarily and let $n\in\mathbb{N}$ be such that $r\theta^{-n}\in(\theta R,R]$. For any $j\in\{1,\ldots,n\}$ we then choose a partition
\begin{equation}\label{eq:sixconres_2}\nonumber
r\theta^{1-j}=\gamma_{j,0}<\gamma_{j,1}<\cdots<\gamma_{j,k_j}=r\theta^{-j}
\end{equation}
with $\gamma_{j,i}-\gamma_{j,i-1}\in(0,\frac{r}{2}]$ for all $i\in\{1,\ldots,k_j\}$ and $k_j=\lceil2\theta^{-j}(1-\theta)\rceil$, and by \eqref{eq:HJ2} we then have
\begin{equation}\label{eq:sixconres_3}
{\textstyle\sum_{i=1}^{k_j}\Big({\textstyle\int_{(\gamma_{j,i-1},\gamma_{j,i}]}G(s,x)\dd x}\Big)\!\Big.^2}\geq k_j^{-1}\Big({\textstyle\int_{(r\theta^{1-j},r\theta^{-j}]}G(s,x)\dd x}\Big)\!\Big.^2.
\end{equation}
Therefore, after first applying Lemma \ref{lm:id;keylemma} with $c=r$, then using the estimates $(xy)^{-1/2}\geq r^{-1}\theta^{j}$ for $x,y\in(r\theta^{1-j},r\theta^{-j}]$ and $(r-|x-y|)_+\geq\frac12r$ and finally also \eqref{eq:sixconres_3}, we find that for all $t\in[0,\infty)$
\begin{equation}\label{eq:sixconres_4}
\int_{[0,r]}G(t,x)\dd x\geq\frac{1}{2r}\!\int_0^t\bigg[\sum_{j=1}^n\frac{\theta^j}{2k_j}\Big({\textstyle\int_{(r\theta^{1-j},r\theta^{-j}]}G(s,x)\dd x}\Big)\!\Big.^2\bigg]\dd s.
\end{equation}
Noting finally that $k_j^{-1}\geq\frac14\theta^j(1-\theta)^{-1}$, it follows, using \eqref{eq:sixconres_1} with $r$ replaced by $r\theta^{-j}$ in the individual terms in the sum on the right hand side of \eqref{eq:sixconres_4}, that
\begin{equation}\nonumber
\int_{[0,r]}G(t,x)\dd x\geq\frac{B^2tnr}{64(1-\theta)}\text{ for all }t\in[0,\tau_3]\text{, with }n\geq\frac{\log(\frac{\theta R}{r})}{|\log\theta|},
\end{equation}
and alternative (b) holds with $C_\theta^{-1}=64(1-\theta)|\log\theta|$.\\

Suppose conversely that \eqref{eq:sixconres_1} is false. Then there must exist $r\in(0,R]$ and $s\in[0,\tau_3]$ such that
\begin{equation}\nonumber
\int_{[0,\theta r]}G(s,x)\dd x\geq\frac{Br}{2}.
\end{equation}
Applying Proposition \ref{pr:tightnessofsolutions} with $\eta=\frac12$ then implies \eqref{eq:sixconres_c} with $R_3=2\theta r$.
\end{proof}
\begin{proof}[Proof of Proposition \ref{pr:seventhconcentrationresult}]
By Lemma \ref{lm:fifthconcentrationresult} there are constants $B_2,R_2\in(0,\infty)$, depending only on $G(0,\cdot)$, and some $T_2\in(0,\frac12\tau_2]$ such that \eqref{eq:fifconres} holds.

Applying Lemma \ref{lm:sixthconcentrationresult} to the function $G(T_2+t,x)\dd x$ with $\theta\in(0,\frac23)$ such that $8\theta L\leq B_2T_2$ and $\tau_3\in(0,\frac12\tau_2]$, we immediately have \eqref{eq:sevconres} with $R_0:=R_3$ if we suppose that alternative (a) holds. Suppose alternatively that (b) holds, then in particular we have
\begin{equation}\label{eq:sevconres_1}
\int_{[0,r]}G(T_2+\tau_3,x)\dd x\geq C_\theta B_2T_2 \tau_3 r\log\big(\tfrac{\theta R_2}{r}\big)\text{ for all }r\in(0,\theta R_2].
\end{equation}
Setting then $R_0:=2\theta R_2 e^{-\beta}$, with $\beta\in(0,\infty)$ such that $\beta C_\theta B_2T_2 \tau_3\geq4L$, we obtain \eqref{eq:sevconres} if we fix $r=\frac12R_0$ in \eqref{eq:sevconres_1} and apply Proposition \ref{pr:tightnessofsolutions} with $\eta=\frac12$.
\end{proof}

\subsubsection{A lower bound on $\int\!\!_{[0,r]}G(t,x)\dd x$, as $r^\alpha$, $\alpha<\frac12$, for $t>0$}
We now prove two results. First we show the following.
\begin{pr}\label{pr:thirdconcentrationresult}
Let $G\in C([0,\infty):\mathcal{M}_+([0,\infty)))$ be a weak solution to \eqref{eq:def:solution}, let $m,R\in(0,\infty)$ and suppose that 
\begin{equation}\label{eq:scr_a}
\int_{[0,R]}G(t,x)\dd x\geq m\text{ for all }t\in[0,\infty).
\end{equation}
Given then any $\alpha\in(0,1)$ there exist a constant $T_*=T_*(\alpha)\in(0,\infty)$ such that
\begin{equation}\label{eq:trc}
\int_{[0,r]}G(t,x)\dd x\geq m\big(\tfrac{r}{2R}\big)^\alpha\text{ for all }r\in[0,R]\text{ and all }t\in[\tfrac{R}{m}T_*,\infty).
\end{equation}
\end{pr}
\noindent Combining this proposition with results obtained in previous subsections, we then derive the following estimate.
\begin{pr}\label{pr:finalconcentrationresult}
Let $G\in C([0,\infty):\mathcal{M}_+([0,\infty)))$ be a nonzero weak solution to \eqref{eq:def:solution} and let $\tau\in(0,1)$ be arbitrary. Given then any $\alpha\in(0,1)$ there exists some $R_*\in(0,\infty)$, depending on $G$, $\alpha$ and $\tau$, such that
\begin{equation}\label{eq:finconres}
\int_{[0,r]}G(t,x)\dd x\geq\tfrac{T_*}{\tau}(2R_*)^{1-\alpha}\,r^\alpha\!\text{ for all }r\in[0,R_*]\!\text{ and all }t\in[\tau,\infty),
\end{equation}
with $T_*=T_*(\alpha)$ as obtained in Proposition \ref{pr:thirdconcentrationresult}.
\end{pr}
We start by proving the following concentration lemma.
\begin{lm}\label{lm:firstconcentration}
Let $G\in C([0,\infty):\mathcal{M}_+([0,\infty)))$ be a weak solution to \eqref{eq:def:solution}, let $\delta\in(0,1)$ and suppose that
\begin{equation}\label{eq:fcl_a}
\int_{[0,1]}G(t,x)\dd x\geq1\text{ for all }t\in[0,\infty).
\end{equation}
Then there exists a constant $T_0=T_0(\delta)\in(0,1024\,\delta^{-3}]$ such that
\begin{equation}\label{eq:fcl_b}
\int_{[0,\frac14\delta]}G(t_0',x)\dd x\geq1-\tfrac12\delta\text{ for some }t_0'\in[0,T_0].
\end{equation}
\end{lm}
\begin{proof}
Suppose first that $\frac14\delta=2^{-n}$ for some $n\in\mathbb{N}$. As in the proofs on Lemmas \ref{lm:fourthconcentrationresult} and \ref{lm:sixthconcentrationresult}, applying Lemma \ref{lm:id;keylemma} with $c=\frac14\delta$ then implies that for all $t\in[0,\infty)$ we have
\begin{equation}\label{eq:fcl_1}
\int_{[0,\frac14\delta]}\!G(t,x)\dd x\geq\!\frac2\delta\int_0^t\iint_{[\frac14\delta,\infty)^2}\!\!\frac{G(s,x)G(s,y)}{\sqrt{xy}}\big(\tfrac14\delta-|x-y|\big)_+\dd x\dd y\,\dd s.
\end{equation}
We then restrict the domain of the double integral on the right hand side of \eqref{eq:fcl_1} to $\mathcal{D}:=\bigcup\!_{i=3}^{8\delta^{-1}}(\frac18\delta(i-1),\frac18\delta i]^2$. For $(x,y)\in\mathcal{D}$ we next have $(xy)^{-1/2}\geq1$ and $(\frac14\delta-|x-y|)_+\geq\frac18\delta$, and applying again \eqref{eq:HJ2} to combine the sums, we obtain
\begin{equation}\label{eq:fcl_3}
\int_{[0,\frac14\delta]}G(t,x)\dd x\geq\frac{\delta}{32}\int_0^t\Big({\textstyle\int_{(\frac14\delta,1]}G(s,x)\dd x}\Big)\!\Big.^2\dd s\text{ for all }t\in[0,\infty).
\end{equation}
We then define the set
\begin{equation}\nonumber
\mathcal{A}=\bigg\{s\in[0,\infty)\ \Big|\ \int_{[0,\frac14\delta]}G(s,x)\dd x<1-\tfrac12\delta\bigg\},
\end{equation}
and note that \eqref{eq:fcl_a} implies that for any $s\in\mathcal{A}$ we have $\int\!\!_{(\frac{1}{4}\delta,1]}G(s,x)\dd x\geq \frac12\delta$. Using this inequality to estimate the right hand side of \eqref{eq:fcl_3}, and the definition of $\mathcal{A}$ for the left hand side, it follows that if $[0,t)\subset\mathcal{A}$, then $1-\frac12\delta>\frac1{128}\delta^3t$. We then find that $\sup\{t\in[0,\infty):[0,t)\subset\mathcal{A}\}\leq 128\delta^{-3}(1-\frac12\delta)=:T_0(\delta)$ and the result follows for those $\delta\in(0,1)$ that are negative powers of $2$.

Noting finally that for any $t_0'\in[0,\infty)$ the left hand side of \eqref{eq:fcl_b} is nondecreasing in $\delta$ and the right hand side is decreasing, the result can be extended to all $\delta\in(0,1)$ by setting $T_0(\delta):=T_0(2^{2-n})$ with $n\in\mathbb{N}$ such that $\frac14\delta\in[2^{-n},2^{1-n})$.
\end{proof}
Next is a scaling result, the proof of which is an elementary rescaling of \eqref{eq:def:solution} that we omit here.
\begin{lm}\label{lm:scaling}
Let $G\in C([0,\infty):\mathcal{M}_+([0,\infty)))$ be a weak solution to \eqref{eq:def:solution} and let $\kappa,\lambda\in(0,\infty)$. The function $G_{\kappa,\lambda}\in C([0,\infty):\mathcal{M}_+([0,\infty)))$, defined by \eqref{eq:resckl}, is then also a weak solution to \eqref{eq:def:solution} and $\|G_{\kappa,\lambda}(0,\cdot)\|=\kappa\|G(0,\cdot)\|$.
\end{lm}
Combining Lemmas \ref{lm:firstconcentration} and \ref{lm:scaling} we obtain the following.
\begin{lm}\label{lm:secondconcentrationresult}
Let $G\in C([0,\infty):\mathcal{M}_+([0,\infty)))$ be a weak solution to \eqref{eq:def:solution}, let $m,R\in(0,\infty)$ and suppose that \eqref{eq:scr_a} holds. Let further $\delta\in(0,1)$ and let $T_0=T_0(\delta)\in(0,\infty)$ be as obtained in Lemma \ref{lm:firstconcentration}. Then there exists some $t_0\in[0,\frac{R}{m}T_0]$ such that
\begin{equation}\label{eq:scr_b}
\int_{[0,\frac12R]}G(t,x)\dd x\geq m(1-\delta)\text{ for all }t\in[t_0,\infty).
\end{equation}
\end{lm}
\begin{proof}
We first define the function $G_*(t,x)\dd x=m^{-1}G(\frac{R}{m}t,Rx)R\dd x$, which by Lemma \ref{lm:scaling} is a weak solution to \eqref{eq:def:solution}. With \eqref{eq:scr_a} we then obtain that
\begin{equation}\label{eq:scr_1}\nonumber
\int_{[0,1]}G_*(t,x)\dd x=m^{-1}\int_{[0,R]}G(\tfrac{R}{m}t,x)\dd x\geq1\text{ for all }t\in[0,\infty),
\end{equation}
and applying Lemma \ref{lm:firstconcentration} with some fixed $\delta\in(0,1)$, we obtain that
\begin{equation}\label{eq:scr_2}\nonumber
\int_{[0,\frac14\delta]}G_*(t_0',x)\dd x\geq1-\tfrac12\delta\text{ for some }t_0'\in[0,T_0(\delta)].
\end{equation}
Applying then Proposition \ref{pr:tightnessofsolutions} with $\eta=\frac12\delta$, it follows that
\begin{equation}\label{eq:scr_3}\nonumber
\int_{[0,\frac12R]}G(\tfrac{R}{m}t,x)\dd x=m\int_{[0,\frac12]}G_*(t,x)\dd x\geq m(1-\delta)\text{ for all }t\in[t_0',\infty),
\end{equation}
and $G$ indeed satisfies \eqref{eq:scr_b} with $t_0:=\frac{R}{m}t_0'\in[0,\frac{R}{m}T_0]$.
\end{proof}
Using Lemma \ref{lm:secondconcentrationresult} we can now prove Proposition \ref{pr:thirdconcentrationresult}.
\begin{proof}[Proof of Proposition \ref{pr:thirdconcentrationresult}]
Fix $\alpha\in(0,1)$ arbitrarily and set $\delta=1-2^{-\alpha}\in(0,\frac12)$. Choosing now any $r\in[0,R]$, there exist unique $\rho\in(R,2R]$ and $n\in\mathbb{N}$ such that $r=\frac{\rho}{2^{n}}$. By repeated application of Lemma \ref{lm:secondconcentrationresult} we then obtain
\begin{align}\nonumber
\int_{[0,r]}G(t,x)\dd x&\geq m(1-\delta)^n=m\big(\tfrac{r}{\rho}\big)^\alpha\text{ for all }t\in[t_r,\infty),\\\nonumber
\text{with }t_r&\leq \sum_{j=1}^n\frac{\frac{\rho}{2^{j-1}}T_0(\delta)}{m(1-\delta)^{j-1}}\leq\tfrac{R}{m}T_*(\alpha),
\end{align}
where $T_0(\delta)\in(0,\infty)$ is the constant obtained in Lemma \ref{lm:firstconcentration}, and where we define $T_*=T_*(\alpha):=2\,T_0(\delta)\sum_{j=0}^\infty\frac{1}{2^{(1-\alpha)j}}$, which is finite as $\alpha\in(0,1)$. Noting now that $T_*$ depends only on $\alpha$, it finally follows from our choice of $\rho$ that $G$ satisfies \eqref{eq:trc}.
\end{proof}
We can now combine Propositions \ref{pr:seventhconcentrationresult} and \ref{pr:thirdconcentrationresult} to prove Proposition \ref{pr:finalconcentrationresult}.
\begin{proof}[Proof of Proposition \ref{pr:finalconcentrationresult}]
For some arbitrarily fixed $\alpha\in(0,1)$, let $T_*=T_*(\alpha)\in(0,\infty)$ be the constant obtained in Proposition \ref{pr:thirdconcentrationresult} and let $L=\frac1\tau2T_*$. Setting further $\tau_2=\frac12\tau$, then by Proposition \ref{pr:seventhconcentrationresult} there exists some $R_0\in(0,\infty)$, depend\-ing on $G$, $L$ and $\tau$, such that \eqref{eq:sevconres} holds. Applying finally Proposition \ref{pr:thirdconcentrationresult} implies that $G$ satisfies \eqref{eq:finconres} with $R_*:=R_0$.
\end{proof}
\subsubsection{End of the proof of Proposition \ref{pr:instantdirac}}
\begin{proof}[Proof of Proposition \ref{pr:instantdirac}]
Without loss of generality we restrict ourselves to nonzero solutions and we assume that $\bar{t}=0$ and $\int\!\!_{\{0\}}G(0,x)\dd x=0$ (cf.~Remark \ref{rk:remark3_8}). The proof argues by contradiction; indeed, suppose that
\begin{equation}\label{eq:proofprop7_1}\nonumber
T:=\sup\bigg\{t\in[0,\infty)\ \Big|\ \int_{\{0\}}G(s,x)\dd x=0\text{ for all }s\in[0,t]\bigg\}>0.
\end{equation}
Note that $T$ may be infinite. Given now any $\alpha\in(0,\frac12)$, then applying Proposition \ref{pr:finalconcentrationresult} with $\tau=\frac13(T\wedge 1)$ there exists some $R_*\in(0,\infty)$, depending on $G$, $\alpha$ and $\tau$, such that
\begin{equation}\label{eq:proofinstdirac_1}
\int_{\tau}^{2\tau}\int_{(0,r]}G(s,x)\dd x\,\dd s\geq T_*(\alpha)(2R_*)^{1-\alpha}\,r^\alpha\text{ for all }r\in[0,R_*].
\end{equation}
Note that in the derivation of \eqref{eq:proofinstdirac_1} we use that $2\tau<T$, which implies that $\int\!\!_{[0,r]}G(s,x)\dd x=\int\!\!_{(0,r]}G(s,x)\dd x$ in the region of integration. On the other hand, Lemma \ref{lm:averagedupperbound} implies that
\begin{equation}\label{eq:proofinstdirac_2}
\int_{\tau}^{2\tau}\int_{(0,r]}G(s,x)\dd x\,\dd s\leq12\sqrt{\tau\|G(0,\cdot)\|}\sqrt{r}\text{ for all }r\in[0,\infty),
\end{equation}
so combining \eqref{eq:proofinstdirac_1} and \eqref{eq:proofinstdirac_2} we obtain
\begin{equation}\nonumber
0<\frac{T_*(\alpha)(2R_*)^{1-\alpha}}{12\sqrt{\tau\|G(0,\cdot)\|}}\leq r^{\frac12-\alpha}\text{ for all }r\in[0,R_*].
\end{equation}
Since $\alpha<\frac12$, this gives a contradiction if $r$ is sufficiently small. Therefore $T=0$ and the result follows using Proposition \ref{pr:monotonicityoforigin}.
\end{proof}
Notice that Proposition \ref{pr:instantdirac} allows us to classify all stationary weak solutions to \eqref{eq:def:solution} with finite total measure.
\begin{cor}\label{cor:uniquestationarysolutions}
The trivial solutions in the sense of Remark \ref{rk:deftrivsol} are the only stationary weak solutions to \eqref{eq:def:solution} in the sense of Definition \ref{def:notionofsolution}.
\end{cor}
\begin{proof}
From Corollary \ref{cor:trivialsolutions} it follows that trivial solutions are stationary. On the contrary, any nontrivial solution $G\in C([0,\infty):\mathcal{M}_+([0,\infty)))$ to \eqref{eq:def:solution} satisfies $\int\!\!_{(0,\infty)}G(\bar{t},x)\dd x>0$ for some $\bar{t}\in[0,\infty)$. By Proposition \ref{pr:instantdirac} it then follows that $\int\!\!_{\{0\}}G(t,x)\dd x>\int\!\!_{\{0\}}G(\bar{t},x)\dd x$ for all $t>\bar{t}$ and $G$ is nonstationary.
\end{proof}

\section{Self-similar solutions}\label{sec:self-similarity}
\subsection{Generalized self-similar energy solutions}
In this section we prove Theorem \ref{tm:self-similarsolutions}. We first state the problem that must be solved by $\Phi$ in \eqref{eq:ssesolution}, in order for $G$ therein to be a weak solution to \eqref{eq:def:solution}.
\begin{pr}\label{pr:rearragementresult}
Suppose there exists a nonnegative function $\Phi\in L^1(0,\infty;\linebreak(1+x))$ that for all $\varphi\in C^1([0,\infty])$ satisfies
\begin{equation}\label{eq:selfsimPhiprofile}
\int_{(0,\infty)}\tfrac12\big(x\varphi'(x)-\varphi(x)+\varphi(0)\big)\Phi(x)\dd x=\sint\frac{\Phi(x)\Phi(y)}{\sqrt{xy}}\Delta_{\varphi}(x,y)\dd x\dd y,
\end{equation}
and let $t_0\in(0,\infty)$ and $M\in[0,\infty)$ be such that $M\sqrt{t_0}\geq\|\Phi\|_{L^1(0,\infty)}$. Then the function $G\in C([0,\infty):\mathcal{M}_+([0,\infty)))$, defined by \eqref{eq:ssesolution}, is a weak solution to \eqref{eq:def:solution} in the sense of Definition \ref{def:notionofsolution}.
\end{pr}
\begin{rk}
The occurrence of the term $\varphi(0)$ in the left hand side may look puzzling at first sight; indeed, since $\Phi$ is an $L^1$-function one would not expect $\varphi(0)$ to appear. However, since $\Delta_{\varphi(\cdot)-\varphi(0)}\equiv\Delta_{\varphi(\cdot)}$ it turns out that using test functions $\varphi$ that take arbitrary values at $x=0$ gives no more information than using only test functions satisfying $\varphi(0)=0$. It is in fact more convenient to use formulation \eqref{eq:selfsimPhiprofile} with test functions satisfying $\varphi(0)\neq0$, since this will allow us to obtain information about the transfer of mass towards the origin of solutions to \eqref{eq:def:solution} of the form \eqref{eq:ssesolution}.
\end{rk}
\begin{proof}[Proof of Proposition \ref{pr:rearragementresult}]
Let $\Phi$, $t_0$ and $M$ be as in the statement of the proposition, let $G$ be defined by \eqref{eq:ssesolution} and let $\varphi,\psi\in C^1([0,\infty):C_c^1([0,\infty)))$ be two functions related by $\psi(s,x)=\varphi(s,x\sqrt{s+t_0})$ for all $(s,x)\in[0,\infty)^2$. Then
\begin{equation}\label{eq:tmsssol_6}\nonumber
\int_{\{0\}}\varphi(s,x)G(s,x)\dd x=\bigg(M-\frac{\|\Phi\|_{L^1(0,\infty)}}{\sqrt{s+{t_0}}}\bigg)\varphi(s,0)\text{ for all }s\in[0,\infty).
\end{equation}
Using the fundamental theorem of calculus we find for all $t\in[0,\infty)$ that
\begin{equation}\label{eq:tmsssol_7}
\begin{split}
&\int_{\{0\}}\varphi(t,x)G(t,x)\dd x-\int_{\{0\}}\varphi(0,x)G(0,x)\dd x\\
&\indent\!\!\!=\int_0^t\int_{\{0\}}\varphi_s(s,x)G(s,x)\dd x\,\dd s
+\int_0^t\int_{(0,\infty)}\frac{\psi(s,0)}{2(s+{t_0})^{3/2}}\Phi(x)\dd x\,\dd s.
\end{split}
\end{equation}
We also observe that
\begin{equation}\label{eq:tmsssol_8}\nonumber
\int_{(0,\infty)}\varphi(s,x)G(s,x)\dd x=\int_{(0,\infty)}\frac{\psi(s,x)}{\sqrt{s+{t_0}}}\Phi(x)\dd x\text{ for all }s\in[0,\infty),
\end{equation}
so as above we find for all $t\in[0,\infty)$ that
\begin{equation}\label{eq:tmsssol_9}
\begin{split}
&\int_{(0,\infty)}\varphi(t,x)G(t,x)\dd x-\int_{(0,\infty)}\varphi(0,x)G(0,x)\dd x\\
&\indent=\int_0^t\int_{(0,\infty)}\bigg(\frac{\psi_s(s,x)}{\sqrt{s+{t_0}}}-\frac{\psi(s,x)}{2(s+{t_0})^{3/2}}\bigg)\Phi(x)\dd x\,\dd s.
\end{split}
\end{equation}
Notice next that $\psi_s(s,x)=\varphi_s(s,x\sqrt{s+{t_0}})+\frac{x}{2(s+{t_0})}\psi_x(s,x)$ and that
\begin{equation}\label{eq:tmsssol_10}
\begin{split}
\int_{(0,\infty)}\frac{\varphi_s(s,x\sqrt{s+{t_0}})}{\sqrt{s+{t_0}}}\Phi(x)\dd x
&=\int_{(0,\infty)}\frac{\varphi_s(s,x)}{s+{t_0}}\Phi\big(\tfrac{x}{\sqrt{s+{t_0}}}\big)\dd x\\
&=\int_{(0,\infty)}\varphi_s(s,x)G(s,x)\dd x.
\end{split}
\end{equation}
Combining then \eqref{eq:tmsssol_7}, \eqref{eq:tmsssol_9} and \eqref{eq:tmsssol_10}, we find for all $t\in[0,\infty)$ that 
\begin{equation}\nonumber
\begin{split}
&\int_{[0,\infty)}\varphi(t,x)G(t,x)\dd x-\int_{[0,\infty)}\varphi(0,x)G(0,x)\dd x\\
&\indent=\int_0^t\bigg[\int_{[0,\infty)}\varphi_s(s,x)G(s,x)\dd x\\
&\indent\!\!\!+\frac1{(s+{t_0})^{3/2}}\int_{(0,\infty)}\tfrac12\big(x\psi_x(s,x)-\psi(s,x)+\psi(s,0)\big)\Phi(x)\dd x\bigg]\dd s.
\end{split}
\end{equation}
By \eqref{eq:selfsimPhiprofile}, the second term between square brackets now equals
\begin{equation}\label{eq:tmsssol_12}
\begin{split}
\frac1{(s+{t_0})^{3/2}}&\sint\frac{\Phi(x)\Phi(y)}{\sqrt{xy}}\Delta_{\psi(s,\cdot)}(x,y)\dd x\dd y\\=&\sint\frac{G(s,x)G(s,y)}{\sqrt{xy}}\Delta_{\varphi(s,\cdot)}(x,y)\dd x\dd y,
\end{split}
\end{equation}
where we have used the fact that for $\varphi(s,\cdot)\in C^1([0,\infty])$ the contribution of the Dirac measure at the origin to the integral on the right hand side of \eqref{eq:tmsssol_12} is zero (cf.~Remark \ref{rk:Taylorargument}). We conclude that $G$ is indeed a weak solution to \eqref{eq:def:solution}.
\end{proof}
\noindent In view of Proposition \ref{pr:rearragementresult}, the proof of Theorem \ref{tm:self-similarsolutions} now follows from the result below.
\begin{pr}\label{pr:selfsimexistenceofPhiprofile}
Given any $E\in[0,\infty)$ there exists at least one nonnegative function $\Phi\in C^{0,\alpha}((0,\infty))\cap L^1(0,\infty)$, $\alpha<\frac12$, with $\int\!\!_{(0,\infty)}x\Phi(x)\dd x=E$ that for all $\varphi\in C^1([0,\infty])$ satisfies \eqref{eq:selfsimPhiprofile}.
\end{pr}
Our strategy to prove Proposition \ref{pr:selfsimexistenceofPhiprofile} is as follows. The general idea is to obtain $\Phi$ as a fixed point of an evolution equation for which the steady states are those $\Phi$ that we are looking for. Since this evolution leaves the first moment invariant, it will be notationally convenient to introduce the energy profile $\Psi$, defined by $\Psi(x)=x\Phi(x)$ for $x\in[0,\infty)$. Notice that for $\varphi\in C_c^2((0,\infty])$, we can rewrite \eqref{eq:selfsimPhiprofile} as
\begin{equation}\label{eq:selfsimPsiprofile}
\int_{[0,\infty)}\tfrac12x\vartheta'(x)\Psi(x)\dd x=\sint\frac{\Psi(x)\Psi(y)}{(xy)^{3/2}}\Delta_{\cdot\vartheta(\cdot)}(x,y)\dd x\dd y,
\end{equation}
where $\vartheta(x)=\frac1x\varphi(x)$ and where, for the sake of brevity, in the right hand side of \eqref{eq:selfsimPsiprofile} we write $\Delta_{\cdot\vartheta(\cdot)}$ where we mean $\Delta_{\varphi}$ with $\varphi(x)=x\vartheta(x)$.

As a first step in the proof of Proposition \ref{pr:selfsimexistenceofPhiprofile}, we prove existence of $\Psi\in\mathcal{M}_+((0,\infty);(1+x))$ that satisfy \eqref{eq:selfsimPsiprofile} for all $\vartheta\in C_c^2((0,\infty])$. To this end we regularize \eqref{eq:selfsimPsiprofile} by means of a parameter $\varepsilon>0$ to control the singular terms on the right hand side. We then prove existence of solutions $\Psi_\varepsilon$ to the regularized problems, reformulating them as the steady states of suitable evolution equations that conserve the total measure. To be able to take the limit $\varepsilon\rightarrow0$ we prove certain a priori estimates, showing that the mass of $\Psi_\varepsilon$ does not concentrate at the origin, nor at infinity, which gives $\Psi$ as desired.

In the second step, we obtain that $\Psi\in C^{0,\alpha}((0,\infty))\cap L^1(0,\infty)$ for $\alpha<\frac12$, for which we need to show that $\int\!\!_{(0,1)}x^{-\beta}\Psi(x)\dd x<\infty$ for $\beta<\frac32$. In particular, these integral estimates allow us to show that the function $\Phi$, defined as $\Phi(x):=\frac1x\Psi(x)$, is integrable, and due to \eqref{eq:selfsimPsiprofile} satisfies \eqref{eq:selfsimPhiprofile} for all $\varphi\in C_c^2((0,\infty])$.

The final step in the proof of Proposition \ref{pr:selfsimexistenceofPhiprofile} is then to show that \eqref{eq:selfsimPhiprofile} holds for all $\varphi\in C^1([0,\infty])$, in particular that we can take test functions with $\varphi(0)\neq0$. This will again require the integral estimates proved in the second step.

\subsubsection{Existence of solutions to a regularized problem}
The goal of this subsection is to prove existence of solutions $\Psi_\varepsilon\in\mathcal{M}_+([\varepsilon,\infty))$ to a regularized version of \eqref{eq:selfsimPsiprofile}. To be able to precisely state the regularized problem under consideration, we require the following two definitions.
\begin{df}\label{def:regularization_1}
For arbitrary $\varepsilon\in(0,1)$ we define the following.
\begin{itemize}
\item{Let $\eta_\varepsilon\in C^\infty([0,\infty])$ be nondecreasing, satisfying $\eta_\varepsilon\equiv0$ on $[0,\varepsilon]$, $\eta_\varepsilon\equiv1$ on $[2\varepsilon,\infty)$ and $\eta_\varepsilon'\leq\frac2\varepsilon$ on $[0,\infty)$.}
\item{For $\varphi\in C([\varepsilon,\infty))$, let $\Delta_\varphi^\varepsilon:\{x\geq y\geq0\}\rightarrow\mathbb{R}$ be given by
\begin{equation}\label{eq:Deltaeps}
\Delta_\varphi^\varepsilon(x,y):=
\begin{cases}
\left.\begin{split}&(1-\eta_\varepsilon(x-y))(\varphi(2x)-2\varphi(x))\\&\hspace{16pt}+\eta_\varepsilon(x-y)\Delta_\varphi(x,y)\end{split}\right\}&\text{ if }y\geq\varepsilon,\\
0&\text{ if }y<\varepsilon,
\end{cases}
\end{equation}
with $\Delta_\varphi$ as in \eqref{eq:def;Delta}.}
\end{itemize}
\end{df}
\noindent The regularized problem will be solved in the following functional spaces.
\begin{df}\label{def:regularization_2}
For all $\varepsilon\in[0,1)$ and any $E\in[0,\infty)$ we define the set
\begin{equation}\nonumber
\begin{split}
\mathcal{I}^\varepsilon(E):=\mathcal{M}_+([\varepsilon,\infty))&\cap\Big\{{\textstyle \int_{[\varepsilon,\infty)}\mu(x)\dd x=E}\Big\}\\&\cap\Big\{{\textstyle\int_{(2,\infty)}\frac{(x-2)_+^2}{x}\mu(x)\dd x\leq36E^2\Big\}}.
\end{split}
\end{equation}
\end{df}
\begin{rk}
Note that for all $\varepsilon\in[0,1)$ and all $E\in[0,\infty)$ the set $\mathcal{I}^\varepsilon(E)$ is weakly-\st compact; indeed, the sets $\mathcal{I}^\varepsilon(E)$ are $\|\cdot\|$-bounded subsets of $\mathcal{M}_+([\varepsilon,\infty))$ and thereby precompact in the weak-\st topology by Banach-Alaoglu (cf.~\cite{B11}). Due to the tightness of $\mu\in\mathcal{I}^\varepsilon(E)$, provided by the last constraint, the sets $\mathcal{I}^\varepsilon(E)$ are weakly-\st closed, hence weakly-\st compact.
\end{rk}
We can now state the result that we will prove in this subsection.
\begin{pr}\label{pr:selfsimapproximatePsiprofile}
Given arbitrary $\varepsilon\in(0,1)$ and $E\in[0,\infty)$, there exists at least one $\Psi_\varepsilon\in\mathcal{I}^\varepsilon(E)$ that for all $\vartheta\in C^1([\varepsilon,\infty])$ satisfies
\begin{equation}\label{eq:selfsimPsiregprofile}
\int_{[0,\infty)}\tfrac12x\eta_\varepsilon(x)\vartheta'(x)\Psi_\varepsilon(x)\dd x=\sint\frac{\Psi_\varepsilon(x)\Psi_\varepsilon(y)}{(xy)^{3/2}}\Delta_{\cdot\vartheta(\cdot)}^\varepsilon(x,y)\dd x\dd y.
\end{equation}
\end{pr}
The proof of Proposition \ref{pr:selfsimapproximatePsiprofile} will be made through a reformulation of \eqref{eq:selfsimPsiregprofile} as the stationary points of an evolution semigroup. In fact, we need to further regularize the right hand side of \eqref{eq:selfsimPsiregprofile}, as can be seen in the following lemma.
\begin{pr}\label{pr:defsemigroup}
Let $\varepsilon\in(0,1)$ and $E\in[0,\infty)$ be arbitrary, let $\phi\in C_c^\infty((-1,1))$ be nonnegative with $\|\phi\|_{L^1(-1,1)}=1$ and for $a\in(0,1)$ define $\phi_a(x):=\frac1a\phi(\frac{x}{a})$. Now, for all $a\in(0,1)$ there exists a weakly-\st continuous semigroup $(S_a(t))_{t\geq0}$ on $\mathcal{I}^\varepsilon(E)$ such that given any $\Psi_0\in\mathcal{I}^\varepsilon(E)$, then $\Psi(t,x):=(S_a(t)\Psi_0)(x)$ satisfies
\begin{equation}\label{eq:selfsimPsiregprofdynamics}
\begin{split}
&\int_{[\varepsilon,\infty]}\vartheta(t,x)\Psi(t,x)\dd x-\int_{[\varepsilon,\infty]}\vartheta(0,x)\Psi_0(x)\dd x\\
&\indent\begin{split}=\int_0^t\bigg[&\int_{[\varepsilon,\infty]}\big(\vartheta_s(s,x)-\tfrac12x\eta_\varepsilon(x)\vartheta_x(s,x)\big)\Psi(s,x)\dd x\\
&+\iint_{\{x>y\geq\varepsilon\}}\frac{\Psi(s,x)(\phi_{a}\ast\Psi(s,\cdot))(y)}{(xy)^{3/2}}\Delta_{\cdot\vartheta(s,\cdot)}^\varepsilon(x,y)\dd x\dd y\bigg]\dd s,\end{split}
\end{split}
\end{equation}
for all $t\in[0,\infty)$ and all $\vartheta\in C^1([0,\infty):C^1([\varepsilon,\infty]))$.
\end{pr}
Notice that we do not claim that solutions to \eqref{eq:selfsimPsiregprofdynamics} are unique. Furthermore, for the sake of simplicity, we omit the dependence on $\varepsilon$ and $E$ in the notation of the semigroups $(S_a(t))_{t\geq0}$.
\begin{rk}
Note that since $\Psi(t,\cdot):=S_a(t)\Psi_0\in\mathcal{I}^\varepsilon(E)$ for all $t\in[0,\infty)$, we could replace the domains of integration $[\varepsilon,\infty]$ in \eqref{eq:selfsimPsiregprofdynamics} by $[\varepsilon,\infty)$. For further reference it is more convenient to keep the notation above.
\end{rk}
\begin{rk}
It is worth noticing that \eqref{eq:selfsimPsiregprofdynamics} admits an interpretation in terms of particle interactions along the lines of Remark \ref{rk:particleinterpretation}. Our regularization is essentially a thickening of the origin, by which we mean that now not only particles at the origin, but also particles of sizes smaller than a threshold level $\varepsilon>0$, do not interact with any other particle (cf.~Remark \ref{rk:remark3_9}). We wish to retain conservation of the energy of all interacting particles, i.e.~of the energy on $[\varepsilon,\infty)$ that is given by $\int\!\!_{[\varepsilon,\infty)}\Psi(\cdot,x)\dd x$. To that end we continuously adapt those interactions that lead to the creation of small particles. In particular, given a pairs of particles $\{x,y\}\subset[\varepsilon,\infty)$ they will interact at a rate $(xy)^{-1/2}$ to produce with probability $\frac12\eta_\varepsilon(|x-y|)$ the pair $\{x+y,x\wedge y\}$, with probability $\frac12\eta_\varepsilon(|x-y|)$ the pair $\{|x-y|,x\wedge y\}$, with probability $\frac12(1-\eta_\varepsilon(|x-y|))$ the pair $\{2x,x\wedge y\}$ or with probability $\frac12(1-\eta_\varepsilon(|x-y|))$ the pair $\{0,x\wedge y\}$. Note that this procedure does produce particles that are smaller that $\varepsilon$, but on average the energy in $[\varepsilon,\infty)$ is conserved. Finally we notice that by the cut off $\eta_\varepsilon$ in the second term on the right hand side of \eqref{eq:selfsimPsiregprofdynamics}, we avoid the transport of energy towards the region where the particles are smaller than $\varepsilon$.
\end{rk}
The proof of Proposition \ref{pr:defsemigroup} will be split into several lemmas. We first need to introduce some auxiliary notation. Given some fixed $\varepsilon\in(0,1)$, and with $\eta_\varepsilon$ as in Definition \ref{def:regularization_1}, we now define $u\in C^\infty(\mathbb{R}\times[0,\infty]:[0,\infty])$ to be the unique solution to
\begin{equation}\label{eq:selfsimdefofu}
\left\{\begin{split}u_t(t,x)&=\tfrac12u(t,x)\eta_\varepsilon(u(t,x)),\\u(0,x)&=x.\end{split}\right.
\end{equation}
For any $t\in\mathbb{R}$ the mapping $x\mapsto u(t,x)=:X$ defines a differomorphism from $[0,\infty]$ to $[0,\infty]$, where, since \eqref{eq:selfsimdefofu} is an autonomous system, we have for all $x\in[0,\infty]$ that $x=u(-t,X)$. We will use $u$ to obtain a mild solution reformulation of \eqref{eq:selfsimPsiregprofdynamics}, where the effect of the transport term therein is taken care of by integrating along characteristics. Note lastly that, by differentiating \eqref{eq:selfsimdefofu} with respect to the variable $x$, we have
\begin{equation}\nonumber
\partial_tu_x(t,x)=\tfrac12\big(\eta_\varepsilon(u)+u\eta_\varepsilon'(u)\big)u_x(t,x),
\end{equation}
where, by our assumptions on $\eta_\varepsilon$, the right hand side is nonnegative and bounded by $\frac52u_x(t,x)$. From this it then follows that the mapping $t\mapsto\|u_x(t,\cdot)\|_\infty$ is non\-decreasing on $\mathbb{R}$ and that we have
\begin{equation}\label{eq:ux}
\|u_x(t,\cdot)\|_\infty\leq 1\vee e^{5t/2}\text{ for all }t\in\mathbb{R}.
\end{equation}

\begin{df}\label{def:operator}
Let $\varepsilon\in(0,1)$ and $E\in[0,\infty)$ fixed, let $\phi\in C_c^\infty((-1,1))$ and $\phi_a$ be as in Proposition \ref{pr:defsemigroup} and let $\Psi_0\in\mathcal{I}^\varepsilon(E)$ be arbitrary. For all $T\in(0,\infty)$ and any $a\in(0,1)$ we then define the operator $\mathcal{T}_{a}$ from $C([0,T]:\mathcal{M}_+([\varepsilon,\infty]))$ into itself as
\begin{equation}\label{eq:defTa}
\begin{split}
\mathcal{T}_{a}[F](t,X)&=\Psi_0(X)e^{-\int_0^tA_{a}(s)[F(s,\cdot)](X)\dd s}\\&\indent+\int_0^te^{-\int_s^tA_{a}(\sigma)[F(\sigma,\cdot)](X)\dd\sigma}B_{a}(s)[F(s,\cdot)](X)\dd s,
\end{split}
\end{equation}
with $A_{a}(s):\mathcal{M}_+([\varepsilon,\infty])\rightarrow C([\varepsilon,\infty])$, for $s\in[0,T]$, given by
\begin{equation}\label{eq:defAa}
\begin{split}
A_{a}(s)[F](X)&=2\int_\varepsilon^X\frac{(\phi_{a}\ast F(u(s,\cdot)))(u(-s,Y))}{\sqrt{u(-s,X)u(-s,Y)}}\frac{u_x(-s,Y)}{u(-s,Y)}\dd Y\\
&\indent+\partial_s\big[\log(u_x(-s,X))\big],\Big.
\end{split}
\end{equation}
and where $B_{a}(s):\mathcal{M}_+([\varepsilon,\infty])\rightarrow \mathcal{M}_+([\varepsilon,\infty))$, with $s\in[0,T]$, is such that for all $s\in[0,T]$ and all $\varphi\in C([\varepsilon,\infty])$ we have
\begin{equation}\label{eq:defBa}
\begin{split}
&\int_{[\varepsilon,\infty)}\varphi(X)B_{a}(s)[F](X)\dd X\\
&\indent\begin{split}=\iint_{\{X>Y\geq\varepsilon\}}&\frac{F(X)(\phi_{a}\ast F(u(s,\cdot)))(u(-s,Y))}{(u(-s,X)u(-s,Y))^{3/2}}\indent\\
\times\,&\Xi_s^{\varphi}(u(-s,X),u(-s,Y))u_x(-s,X)u_x(-s,Y)\,\dd X\dd Y,\bigg.\end{split}
\end{split}
\end{equation}
with $\Xi_s^\varphi:\{x\geq y\geq0\}\rightarrow\mathbb{R}$ given by
\begin{equation}\label{eq:defXi}
\Xi_s^\varphi(x,y):=\begin{cases}
\begin{split}&(1-\eta_\varepsilon(x-y))\tilde\varphi(s,2x)\\
&\hspace{16pt}+\eta_\varepsilon(x-y)(\tilde\varphi(s,x+y)+\tilde\varphi(s,x-y))\end{split}&\text{ if }y\geq\varepsilon,\\
0&\text{ if }y<\varepsilon,
\end{cases}
\end{equation}
and $\tilde\varphi(s,z):=z\varphi(u(s,z))u_x(s,z)$.
\end{df}
\begin{rk}
For the readers convenience we present some identities to indicate the rationale of the computations. Given $\vartheta=\vartheta(t,x)\in C^1([0,T]:C^1([\varepsilon,\infty]))$ we define 
\begin{equation}\label{eq:4_18}
\varphi(t,X)=\vartheta(t,u(-t,X))u_x(-t,X)\text{ for all }(t,X)\in[0,T]\times[\varepsilon,\infty].
\end{equation}
For any $F\in C([0,T]:\mathcal{M}_+([\varepsilon,\infty]))$, we then first find for all $t\in[0,T]$ that
\begin{equation}\label{eq:sssimprelcomp_1}
\int_{[\varepsilon,\infty]}\varphi(t,X)F(t,X)\dd X=\int_{[\varepsilon,\infty]}\vartheta(t,x)\Psi(t,x)\dd x,
\end{equation}
where we have written $\Psi(t,x)=F(t,u(t,x))$. Furthermore, we have that
\begin{equation}\label{eq:sssimprelcomp_2}
\begin{split}
&\int_{[\varepsilon,\infty]}\!\varphi_t(t,X)F(t,X)\dd X-\int_{[\varepsilon,\infty]}\!\varphi(t,X)\partial_t\big[\log(u_x(-t,X))\big]F(t,X)\dd X\\
&\indent=\int_{[\varepsilon,\infty]}\big(\vartheta_t(t,x)-\tfrac12x\eta_\varepsilon(x)\vartheta_x(t,x)\big)\Psi(t,x)\dd x,
\end{split}
\end{equation}
which follows from \eqref{eq:selfsimdefofu} and the identity
\begin{equation}\nonumber
\begin{split}
\varphi_t(t,X)&=\vartheta_t(t,u(-t,X))u_x(-t,X)-u_t(-t,X)\vartheta_x(t,u(-t,X))u_x(-t,X)\\
&\indent-\vartheta(t,u(-t,X))u_{tx}(-t,X)(u_x(-t,X))^{-1}u_x(-t,X).
\end{split}
\end{equation}
Lastly, with $B_a$ as in \eqref{eq:defBa}, we remark for all $t\in[0,T]$ that
\begin{equation}\label{eq:defBabis}
\begin{split}
&\int_{[\varepsilon,\infty)}\varphi(t,X)B_a(t)[F(t,\cdot)](X)\dd X\\
&\indent=\iint_{\{x>y\geq\varepsilon\}}\frac{\Psi(t,x)(\phi_a\ast \Psi(t,\cdot))(y)}{(xy)^{3/2}}\Xi_t^{\varphi(t,\cdot)}(x,y)\dd x\dd y,
\end{split}
\end{equation}
where for $(x,y)\in\{x\geq y\geq\varepsilon\}$ we note that
\begin{equation}\nonumber
\begin{split}
\Xi_t^{\varphi(t,\cdot)}(x,y)&=(1-\eta_\varepsilon(x-y))2x\vartheta(t,2x)\\&\indent+\eta_\varepsilon(x-y)((x+y)\vartheta(t,x+y)+(x-y)\vartheta(t,x-y))\\
&=\Delta_{\cdot\vartheta(t,\cdot)}^\varepsilon(x,y)+2x\vartheta(t,x).
\end{split}
\end{equation}
\end{rk}
We will prove existence of the semigroup $(S_a(t))_{t\geq0}$ by means of a contractive fixed point argument.
\begin{lm}\label{lm:selfsimcontraction}
Given any $\varepsilon\in(0,1)$, $E\in[0,\infty)$, $\Psi_0\in\mathcal{I}^\varepsilon(E)$ and $a\in(0,1)$, there exists some $T=T(\varepsilon,E)\in(0,\infty)$ such that the operator $\mathcal{T}_{a}$, as defined in Definition \ref{def:operator}, is a contraction on
\begin{equation}\nonumber
\mathcal{S}:=\Big\{F\in C([0,T]:\mathcal{M}_+([\varepsilon,\infty]))\,\big|\,\|F\|_T\leq2E\Big\},
\end{equation}
with norm $\|F\|_T:=\sup_{t\in[0,T]}\|F(t,\cdot)\|$.
\end{lm}
\begin{proof}
The fact that the operators $A_{a}$, $B_{a}$ and $\mathcal{T}_{a}$ are well defined is straightforward as in Lemma \ref{lm:ex;regularizedexistence}.

We now show that $\mathcal{T}_a$ maps $\mathcal{S}$ into itself. Note thereto that $\mathcal{T}_a$ maps $C([0,T]:\mathcal{M}_+([\varepsilon,\infty]))$ into itself since $u_x\geq0$ on $\mathbb{R}\times[0,\infty]$, i.e.~nonnegativity is preserved. It remains to check that for $T\in(0,\infty)$ small enough we have $\|\mathcal{T}_a[F]\|_T\leq2E$ for any $F\in\mathcal{S}$; indeed, using $\varphi\equiv1$ in \eqref{eq:defBabis}, and noting that \eqref{eq:defXi} then implies that $0\leq\Xi_s^\varphi(x,y)\leq2x\|u_x(s,\cdot)\|_\infty$, we first obtain that
\begin{equation}\label{eq:4_23}
\big\|B_a(s)[F(s,\cdot)]\big\|\leq\tfrac2{\varepsilon^2}e^{5s/2}\|F(s,\cdot)\|^2,
\end{equation}
where we have used that, with $\Psi(s,x)=F(s,u(s,x))$ then, due to \eqref{eq:ux}, we have $\|\Psi(s,\cdot)\|\leq\|F(s,\cdot)\|$, and that in the domain of integration in the right hand side of \eqref{eq:defBabis} we have $x,y\geq\varepsilon$. With $\varphi\equiv1$ in \eqref{eq:defTa} we next find for $t\in[0,T]$ that
\begin{equation}\label{eq:4_24}\nonumber
\begin{split}
\big\|\mathcal{T}_a[F](t,\cdot)\big\|&\leq\big\|e^{-\int_0^tA_{a}(s)[F(s,\cdot)](\cdot)\dd s}\big\|_\infty\|\Psi_0\|\\&\indent\textstyle+\int_0^t\big\|e^{-\int_s^tA_{a}(\sigma)[F(\sigma,\cdot)](\cdot)\dd \sigma}\big\|_\infty\big\|B_a(s)[F(s,\cdot)]\big\|\dd s,
\end{split}
\end{equation}
so noting that, since the first term on the right hand side of \eqref{eq:defAa} is nonnegative, we have for $t\in[0,T]$ and $s\in[0,t]$ that
\begin{equation}\label{eq:4_25}
\textstyle\big|e^{-\int_s^tA_{a}(\sigma)[F(\sigma,\cdot)](X)\dd s}\big|\leq\frac{u_x(-s,X)}{u_x(-t,X)}=u_x(t-s,u(-t,X))\leq e^{5(t-s)/2},
\end{equation}
we obtain for all $F\in\mathcal{S}$ and all $T\in(0,\infty)$ that
\begin{equation}\label{eq:normTa}
\|\mathcal{T}_a[F]\|_T\leq\big(e^{5T/2}+\tfrac8{\varepsilon^2}Te^{5T/2}E\big)E.
\end{equation}
It now follows that the right hand side of \eqref{eq:normTa} is smaller than $2E$ for sufficiently small $T\in(0,\infty)$, hence $\mathcal{T}_{a}$ maps $\mathcal{S}$ into itself.

To finally show that $\mathcal{T}_{a}$ is contractive on $\mathcal{S}$ for small enough $T\in(0,\infty)$, we first rewrite \eqref{eq:defAa} as
\begin{equation}\label{eq:defAabis}
A_{a}(s)[F](X)-\partial_s\big[\log(u_x(-s,X))\big]=\frac2{\sqrt{x}}\int_\varepsilon^x\frac{(\phi_{a}\ast \Psi(s,\cdot))(y)}{y^{3/2}}\dd y,
\end{equation}
with again $\Psi(s,x)=F(s,X)$ and $X=u(s,x)$. Given then $F_1,F_2\in\mathcal{S}$ we obtain
\begin{equation}\label{eq:difAaF1AaF2}
\big\|A_{a}(s)[F_1(s,\cdot)](\cdot)-A_{a}(s)[F_2(s,\cdot)](\cdot)\big\|_\infty\leq\tfrac2{\varepsilon^2}\big\|F_1(s,\cdot)-F_2(s,\cdot)\big\|,
\end{equation}
where similar to the above we have now used the fact that $\|\Psi_1(s,\cdot)-\Psi_2(s,\cdot)\|\leq\|F_1(s,\cdot)-F_2(s,\cdot)\|$. We then combine \eqref{eq:difAaF1AaF2}, the fact that the first term on the right hand side of \eqref{eq:defAa} is nonnegative for $F\in\mathcal{M}_+([\varepsilon,\infty])$, the fact that $|e^{x-x_1}-e^{x-x_2}|\leq e^x|x_1-x_2|$ for $x_1,x_2\in[0,\infty)$ and $x\in\mathbb{R}$, and \eqref{eq:4_25} to find for $t\in[0,T]$ and $s\in[0,t]$ that
\begin{equation}\label{eq:4_29}
\begin{split}
&\textstyle\big\|e^{-\int_s^tA_{a}(\sigma)[F_1(\sigma,\cdot)](\cdot)\dd \sigma}-e^{-\int_s^tA_{a}(\sigma)[F_2(\sigma,\cdot)](\cdot)\dd \sigma}\big\|_\infty\\
&\indent\textstyle\leq e^{5(t-s)/2}\int_s^t\big\|A_{a}(\sigma)[F_1(\sigma,\cdot)](\cdot)-A_{a}(\sigma)[F_2(\sigma,\cdot)](\cdot)\big\|_\infty\dd\sigma\\
&\indent\textstyle\leq\frac2{\varepsilon^2}(t-s)e^{5(t-s)/2}\big\|F_1-F_2\big\|_T.
\end{split}
\end{equation}
Arguing as above, we also obtain that
\begin{equation}\label{eq:4_30}
\big\|B_a(s)[F_1(s,\cdot)]-B_a(s)[F_2(s,\cdot)]\big\|\leq\tfrac2{\varepsilon^2}e^{5s/2}\big(\|F_1\|_T+\|F_2\|_T\big)\big\|F_1-F_2\big\|_T,
\end{equation}
and combining \eqref{eq:4_23}, \eqref{eq:4_29} and \eqref{eq:4_30} we find that
\begin{equation}\label{eq:4_31}\nonumber
\big\|\mathcal{T}_a[F_1]-\mathcal{T}_a[F_2]\big\|_T\leq\tfrac2{\varepsilon^2}ETe^{5T/2}\big(3+\tfrac4{\varepsilon^2}ET\big)\big\|F_1-F_2\big\|_T,
\end{equation}
from which we conclude that $\mathcal{T}_a$ is indeed contractive on $\mathcal{S}$ if $T\in(0,\infty)$ is sufficiently small.
\end{proof}
We can now prove existence of a fixed point of $\mathcal{T}_a$, which we will use to obtain an evolution semigroup which solves \eqref{eq:selfsimPsiregprofdynamics}.
\begin{lm}\label{lm:selfsimweaksolutionsexist}
Let $\varepsilon\in(0,1)$, $E\in[0,\infty)$, $\Psi_0\in\mathcal{I}^\varepsilon(E)$ and ${a}\in(0,1)$ be arbitrary, let the operator $\mathcal{T}_{a}$ be as defined in Definition \ref{def:operator} and let $T\in(0,\infty)$ be as obtained in Lemma \ref{lm:selfsimcontraction}. Then $\mathcal{T}_{a}$ admits a unique fixed point $F\in C([0,T]:\linebreak\mathcal{M}_+([\varepsilon,\infty]))$, and defining $\Psi(t,x):=F(t,u(t,x))$ we obtain a function $\Psi\in C([0,T]:\mathcal{M}_+([\varepsilon,\infty]))$ that for all $\vartheta\in C^1([0,T]:C^1([\varepsilon,\infty]))$ and all $t\in[0,T]$ satisfies \eqref{eq:selfsimPsiregprofdynamics}.
\end{lm}
\begin{proof}
Banach's fixed point theorem implies existence of a unique $F\in \mathcal{S}$ such that $F\equiv\mathcal{T}_{a}[F]$ on $[0,T]\times[\varepsilon,\infty]$. Multiplying $F$ by $\varphi\in C^1([0,T]:C^1([\varepsilon,\infty]))$, integrating with respect to $X$, differentiating with respect to $t$, rearranging the terms to eliminate the terms containing exponentials and integrating the resulting formula, we find for all $t\in[0,T]$ that $F$ satisfies
\begin{equation}\label{eq:mildimpliesweak}
\begin{split}
&\int_{[\varepsilon,\infty]}\varphi(t,X)F(t,X)\dd X-\int_{[\varepsilon,\infty]}\varphi(0,X)F(0,X)\dd X\\
&\indent\begin{split}=\int_0^t\bigg[&\int_{[\varepsilon,\infty]}\varphi_s(s,X)F(s,X)\dd X+\int_{[\varepsilon,\infty)}\varphi(s,X)B_{a}[F](s,X)\dd X\\
&\indent-\int_{[\varepsilon,\infty]}\varphi(s,X)A_{a}[F](s,X)F(s,X)\dd X\bigg]\dd s.\end{split}
\end{split}
\end{equation}
Given now any $\vartheta\in C^1([0,T]:C^1([\varepsilon,\infty]))$, we define $\varphi$ by \eqref{eq:4_18} and recalling \eqref{eq:sssimprelcomp_1}, \eqref{eq:sssimprelcomp_2}, \eqref{eq:defBabis} and \eqref{eq:defAabis} we conclude that $\Psi$, as defined in the statement of the lemma, satisfies \eqref{eq:selfsimPsiregprofdynamics} for all $\vartheta\in C^1([0,T]:C^1([\varepsilon,\infty]))$ and all $t\in[0,T]$.
\end{proof}
We next show that the semigroup constructed above preserves the space $\mathcal{I}^\varepsilon(E)$. Actually, this space is preserved for any solution to \eqref{eq:selfsimPsiregprofdynamics}.
\begin{lm}\label{lm:selfsiminvariantset}
Let $\varepsilon\in(0,1)$, $E\in[0,\infty)$, $\Psi_0\in\mathcal{I}^\varepsilon(E)$ and ${a}\in(0,1)$ be arbi\-trary. Suppose that $\Psi\in C([0,T]:\mathcal{M}_+([\varepsilon,\infty]))$ satisfies \eqref{eq:selfsimPsiregprofdynamics} for all $t\in[0,T]$ and all $\vartheta\in C^1([0,T]:C^1([\varepsilon,\infty]))$, then $\Psi(t,\cdot)\in\mathcal{I}^\varepsilon(E)$ for all $t\in[0,T]$.
\end{lm}
\begin{proof}
Using $\vartheta\equiv1$ in \eqref{eq:selfsimPsiregprofdynamics}, it is immediate that the total measure is constant and it only remains to prove the uniform tightness condition. To that end we let $\vartheta(x)=\frac1x(x-2)_+^2$ and for ${c}\in(0,1)$ small we set $\vartheta_{c}(x)=\vartheta(x)/(1+{c}x)$, which we can use in \eqref{eq:selfsimPsiregprofdynamics}. We now first compute that
\begin{equation}\label{eq:ssis_1}
\textstyle-\frac12x\vartheta_{c}'(x)=-\frac12\Big(\frac{(x-2)_+^2}{x(1+{c}x)}+\frac{4(x-2)_+}{x(1+{c}x)}-\frac{{c}(x-2)_+^2}{(1+{c}x)^2}\Big)\leq-\frac12\vartheta_{c}(x)+\frac{{c}(x-2)_+^2}{2(1+{c}x)^2}.
\end{equation}
Observing then that the mapping $x\mapsto x^2-x\vartheta_{c}(x)$ is convex, and due to the fact that $\Delta_{\cdot\vartheta_{c}(\cdot)}^\varepsilon\equiv0$ on $\{\varepsilon\leq y\leq x<1\}$, we find that
\begin{equation}\label{eq:ssis_2}\nonumber
\Delta_{\cdot\vartheta_{c}(\cdot)}^\varepsilon(x,y)\leq\Delta_{(\cdot)^2}^\varepsilon(x,y)\times\mathbf{1}_{\{x\geq1\}}(x,y)\text{ on }\{x>y\geq\varepsilon\},
\end{equation}
and since we have $(2x)^2-2x^2=2x^2\leq2(3y)^2$ on $\{x>y\geq\varepsilon\}\cap\{|x-y|<2\varepsilon\}$ and $\Delta_{(\cdot)^2}(x,y)=2y^2$ we obtain for all $\mu\in\mathcal{M}_+([\varepsilon,\infty])$ that
\begin{equation}\label{eq:ssis_3}
\iint_{\{x>y\geq\varepsilon\}}\frac{\mu(x)(\phi_{a}\ast\mu)(y)}{(xy)^{3/2}}\Delta_{\cdot\vartheta_{c}(\cdot)}^\varepsilon(x,y)\dd x\dd y\leq18\Big({\textstyle\int_{[\varepsilon,\infty]}\mu(x)\dd x}\Big)\!\Big.^2.
\end{equation}
Applying \eqref{eq:ssis_3} in \eqref{eq:selfsimPsiregprofdynamics} we next find for all $t_1\in[0,T]$ and $t_2\in[t_1,T]$ that
\begin{equation}\label{eq:ssis_4}
\begin{split}
\int_{[\varepsilon,\infty]}\vartheta_{c}(x)\Psi(t_2,x)\dd x&\leq\int_{[\varepsilon,\infty]}\vartheta_{c}(x)\Psi(t_1,x)\dd x\\
&\hspace{8pt}+\frac12\int_{t_1}^{t_2}\bigg[\int_{[\varepsilon,\infty]}-x\vartheta_{c}'(x)\Psi(s,x)\dd x+36E^2\bigg]\dd s.
\end{split}
\end{equation}
Noting then that $\vartheta_{c}'\geq0$, then since $\Psi_0\in\mathcal{I}^\varepsilon(E)$ we obtain
\begin{equation}\label{eq:ssis_5}\nonumber
\int_{[\varepsilon,\infty]}\vartheta_{c}(x)\Psi(t,x)\dd x\leq(1+\tfrac{1}{2}T)36E^2=:B=B(E,T)\text{ for all }t\in[0,T],
\end{equation}
hence
\begin{equation}\label{eq:ssis_6}
\int_{[\varepsilon,\infty]}\tfrac{{c}(x-2)_+^2}{2(1+{c}x)^2}\Psi(\cdot,x)\dd x\leq c\int_{[\varepsilon,\infty]}\vartheta_c(x)\Psi(t,x)\dd x\leq B\,c.
\end{equation}
For almost all $t\in[0,T]$ we now find by \eqref{eq:ssis_1}, \eqref{eq:ssis_4} and \eqref{eq:ssis_6} that
\begin{equation}\label{eq:ssis_7}\nonumber
\partial_t\bigg[\int_{[\varepsilon,\infty]}\vartheta_{c}(x)\Psi(t,x)\dd x\bigg]\leq-\frac12\bigg(\int_{[\varepsilon,\infty]}\vartheta_{c}(x)\Psi(t,x)\dd x-36E^2\bigg)+B\,c.
\end{equation}
Therefore
\begin{equation}\nonumber
\partial_t\bigg[\bigg(\int_{[\varepsilon,\infty]}\vartheta_{c}(x)\Psi(t,x)\dd x-36E^2\bigg)e^{\frac12t}\bigg]\leq Be^{\frac12t}\,c,
\end{equation}
and integrating this inequality and taking the limit $c\rightarrow0$, it follows by monotone convergence that
\begin{equation}\label{eq:ssis_8}\nonumber
\int_{[\varepsilon,\infty)}\vartheta(x)\Psi(t,x)\dd x\leq36E^2\text{ for all }t\in[0,T],
\end{equation}
which proves the claim.
\end{proof}

\begin{proof}[Proof of Proposition \ref{pr:defsemigroup}]
Let $T\in(0,\infty)$ be as obtained in Lemma \ref{lm:selfsimcontraction}. Given any $\Psi_0\in\mathcal{I}^\varepsilon(E)$, then for $t\in[0,T]$ we set $S_a(t)\Psi_0:=\Psi(t,\cdot)$ with $\Psi\in C([0,T]:\mathcal{I}^\varepsilon(E))$ as obtained in Lemmas \ref{lm:selfsimweaksolutionsexist} and \ref{lm:selfsiminvariantset}. For arbitrary $t\in[0,\infty)$ we define inductively $S_a(t)\Psi_0:=S_a(t-nT)S_a(nT)\Psi_0$ for $t\in(nT,(n+1)T]$ with $n\in\mathbb{N}$. It is straightforward to check that $(S_a(t))_{t\geq0}$, as defined in this manner, is a semigroup as required in the statement of Proposition \ref{pr:defsemigroup}. In particular, the function $\Psi(t,\cdot):=S_a(t)\Psi_0$ satisfies \eqref{eq:selfsimPsiregprofdynamics} for all $\vartheta\in C^1([0,\infty):C^1([\varepsilon,\infty]))$ and all $t\in[0,\infty)$.

It only remains to show that the mapping $\Psi_0\mapsto S_a(T)\Psi_0$ is weakly-\st continuous on $\mathcal{I}^\varepsilon(E)$ for any fixed $T\in[0,\infty)$. Recalling now \eqref{eq:sssimprelcomp_1}, this is equivalent to proving that the mapping $\Psi_0\mapsto F(T,\cdot):=(S_a(T)\Psi_0)(u(-T,\cdot))$ is weakly-$\ast$ continuous. Given $\Psi_1,\Psi_2\in\mathcal{I}^\varepsilon(E)$, we define $F_i(s,X):=(S_a(s)\Psi_i)(u(-s,X))$ for $(s,X)\in[0,T]\times[\varepsilon,\infty]$ and $i\in\{1,2\}$, which are functions that satisfy \eqref{eq:mildimpliesweak} for all $\varphi\in C^1([0,T]:C^1([\varepsilon,\infty]))$ and all $t\in[0,T]$ (cf.~Proof of Lemma \ref{lm:selfsimweaksolutionsexist}). With this terminology we now need to prove that for any $\psi\in C([\varepsilon,\infty])$ and any $\delta\in(0,1)$ there exists a weakly-\st open set $\mathcal{U}$, depending on $\psi$ and $\delta$, such that if $\Psi_1-\Psi_2\in\mathcal{U}$, then 
\begin{equation}\label{eq:4_38}
\Big|\textstyle\int_{[\varepsilon,\infty)}\psi(X)F_1(t,X)\dd X-\int_{[\varepsilon,\infty)}\psi(X)F_2(t,X)\dd X\Big|<\delta.
\end{equation}
A density argument implies that it is sufficient to prove this implication only for $\psi\in C^1([\varepsilon,\infty])$. Using \eqref{eq:mildimpliesweak} we obtain
\begin{equation}\label{eq:4_39}
\textstyle\int_{[\varepsilon,\infty)}\psi(X)\big(F_1(T,X)-F_2(T,X)\big)\dd X=\int_{[\varepsilon,\infty)}\varphi(0,X)\big(\Psi_1(X)-\Psi_2(X)\big)\dd X,
\end{equation}
if $\varphi\in C^1([0,T]:C^1([\varepsilon,\infty]))$ satisfies $\varphi(T,\cdot)\equiv\psi$ on $[\varepsilon,\infty]$ and
\begin{equation}\label{eq:4_40}
\begin{split}
&\textstyle\int_{[\varepsilon,\infty)}\varphi_s(s,X)\big(F_1(s,X)-F_2(s,X)\big)\dd X\\
&\indent\begin{split}&=\int_{[\varepsilon,\infty)}\varphi(s,X)\big(A_{a}[F_1](s,X)F_1(s,X)-A_{a}[F_2](s,X)F_2(s,X)\big)\dd X\\&\indent-\int_{[\varepsilon,\infty)}\varphi(s,X)\big(B_{a}[F_1](s,X)-B_{a}[F_2](s,X)\big)\dd X,\end{split}
\end{split}
\end{equation}
for all $s\in[0,T]$. Rearranging the terms, it can be seen that the right hand side of \eqref{eq:4_40} equals
\begin{equation}\label{eq:sssemigroup_4}\nonumber
\textstyle\int_{[\varepsilon,\infty)}\varphi(s,X)\partial_s\big[\log(u_x(-s,X))\big]\,\big(F_1(s,X)-F_2(s,X)\big)\dd X-\big(Q_1-Q_2\big),
\end{equation}
with $Q_i=Q(F_i,\varphi,s)$, $i\in\{1,2\}$, given by
\begin{equation}\nonumber
Q_i:=\iint_{\{x>y\geq\varepsilon\}}\frac{F_i(s,u(s,x))(\phi_{a}\ast F_i(s,u(s,\cdot)))(y)}{(xy)^{3/2}}\Delta_{\tilde\varphi(s,\cdot)}^\varepsilon(x,y)\dd x\dd y,
\end{equation}
where $\tilde\varphi(s,z)=z\varphi(s,u(s,z))u_x(s,z)$ and where $\Delta_\varphi^\varepsilon$ is given by \eqref{eq:Deltaeps}. Using then the fact that $f_1\bar{f}_1-f_2\bar{f}_2=\frac12(f_1-f_2)(\bar{f}_1+\bar{f}_2)+\frac12(f_1+f_2)(\bar{f}_1-\bar{f}_2)$, we find through careful computation that
\begin{equation}\nonumber
Q_1-Q_2=\int_{[\varepsilon,\infty)}\big(K_1(s,X;\varphi)+K_2(s,X;\varphi)\big)\big(F_1(s,X)-F_2(s,X))\big)\dd X
\end{equation}
with $K_1(s,X;\varphi)=L_1(s,u(-s,X);F_1,\varphi)+L_1(s,u(-s,X);F_2,\varphi)$ and with\linebreak $K_2(s,X;\varphi)=L_2(s,u(-s,X);F_1,\varphi)+L_2(s,u(-s,X);F_2,\varphi)$, and where
\begin{equation}\nonumber
L_1(s,x;F_i,\varphi)=\frac{u_x(-s,u(s,x))}{2x^{3/2}}\int_\varepsilon^{x}\frac{(\phi_{a}\ast F_i(s,u(s,\cdot)))(y)}{y^{3/2}}\Delta_{\tilde\varphi(s,\cdot)}^\varepsilon(x,y)\dd y
\end{equation}
and
\begin{equation}\nonumber
\begin{split}
L_2(s,x;F_i,\varphi)&=\tfrac12u_x(-s,u(s,x))\\
&\indent\times\int_\varepsilon^\infty\phi_{a}(z-x)\int_{(z,\infty)}\frac{F_i(s,u(s,y))}{(yz)^{3/2}}\Delta_{\tilde\varphi(s,\cdot)}^\varepsilon(y,z)\dd y\dd z.
\end{split}
\end{equation}
We then obtain that \eqref{eq:4_40}, and hence \eqref{eq:4_39}, holds if $\varphi\in C^1([0,T]:C^1([\varepsilon,\infty]))$ is a solution to the reverse in time initial value problem
\begin{equation}\label{eq:4_46}
\left\{\begin{split}
\varphi_s(s,X)&=\varphi(s,X)\,\partial_s\log(u_x(-s,X))-K_1(s,X;\varphi)-K_2(s,X;\varphi)\\&=:\mathcal{L}(s)(\varphi),\\
\varphi(T,\cdot)&\equiv\psi,
\end{split}\right.
\end{equation}
on $[0,T]\times[\varepsilon,\infty]$. Unique existence of solutions $\varphi\in C^1([0,T]:C^1([\varepsilon,\infty]))$ to \eqref{eq:4_46} follows by means of a standard fixed point argument, since the mappings $\varphi\mapsto\mathcal{L}(s)(\varphi)$, with $s\in[0,T]$, are linear from $C^1([\varepsilon,\infty])$ into itself. Moreover, these mappings are uniformly bounded if $s\in[0,T]$, since $\|F_i\|_T\leq e^{5T/2}E$ for $i\in\{1,2\}$ by conservation of measure of $\Psi$ and \eqref{eq:ux}, and we therefore obtain that $\|\varphi(0,\cdot)\|_{C^1([\varepsilon,\infty])}\leq C(\varepsilon,E,T)\|\psi\|_{C^1([\varepsilon,\infty])}$.

In order to conclude the proof of weak-\st continuity of $S_a(T)$ we now construct a weakly-\st open set $\mathcal{U}$ such that if $\Psi_1-\Psi_2\in\mathcal{U}$ then \eqref{eq:4_38} holds. Since the $\|\cdot\|_{C^1}$-bounded functions are compact in the space of continuous functions (cf.~Arzel\`a-Ascoli), there exist finitely many functions $\omega_1,\ldots,\omega_n\in C^0([0,\infty])$ such that $\inf_{i\in\{1,\ldots,n\}}\|\varphi(0,\cdot)-\omega_i\|_\infty\leq\delta(3(E+1))^{-1}=:\delta_*$. We then define
\begin{equation}\nonumber
\mathcal{U}:=\Big\{\mu\in\mathcal{M}([\varepsilon,\infty])\,:\,\big|{\textstyle\int_{[\varepsilon,\infty]}\omega_i(x)\mu(x)\dd x}\big|<\delta_*\text{ for all }i\in\{1,\ldots,n\}\Big\}.
\end{equation}
We can then write the right hand side of \eqref{eq:4_39} as
\begin{equation}\nonumber
\int_{[\varepsilon,\infty]}\omega_i(x)\big(\Psi_1(x)-\Psi_2(x)\big)\dd x+\int_{[\varepsilon,\infty]}\big(\varphi(0,x)-\omega_i(x)\big)\big(\Psi_1(x)-\Psi_2(x)\big)\dd x,
\end{equation}
which we can estimate by $\delta_*+2\delta_*E<\delta$, choosing $i\in\{1,\ldots,n\}$ such that $\|\varphi(0,\cdot)-\omega_i\|_\infty\leq\delta_*$ and assuming that $\Psi_1-\Psi_2\in\mathcal{U}$. The proof of the proposition is then complete.
\end{proof}
The proof of Proposition \ref{pr:selfsimapproximatePsiprofile} now uses a method analogous to the one in \cite{EMR05} and \cite{GPV04}.
\begin{proof}[Proof of Proposition \ref{pr:selfsimapproximatePsiprofile}]
Let $\varepsilon\in(0,1)$ and $E\in[0,\infty)$ be fixed arbitrarily and for any $a\in(0,1)$ let $(S_a(t))_{t\geq0}$ be a weakly-\st continuous semigroup on $\mathcal{I}^\varepsilon(E)$ as obtained in Proposition \ref{pr:defsemigroup}.
For each $t\in[0,\infty)$ there now exists some $\Psi_t\in\mathcal{I}^\varepsilon(E)$ such that $S_a(t)\Psi_t\equiv\Psi_t$ on $[\varepsilon,\infty]$ (by the Schauder-Tychonoff fixed point theorem; see e.g.~\cite[Thm.~3.6.1]{E65}), and we find that $S_a(i2^{-m})\Psi_{2^{-n}}\equiv\Psi_{2^{-n}}$ for all $i,n\in\mathbb{N}$ and all $m\in\{1,\ldots,n\}$. By compactness there then exists a subsequence $\Psi_{2^{-n_k}}\wsc \Psi_{a}\in\mathcal{I}^\varepsilon(E)$ and by continuity we have $S_a(t)\Psi_{a}\equiv\Psi_{a}$ for all dyadic $t\in[0,\infty)$. As the dyadic numbers are dense in $[0,\infty)$, it now follows that $\Psi_{a}$ is stationary under $(S_a(t))_{t\geq0}$ and choosing test functions independent of $t$ it follows from \eqref{eq:selfsimPsiregprofdynamics} that $\Psi_{a}$ satisfies
\begin{equation}\nonumber
\int_{[0,\infty)}\tfrac12x\eta_\varepsilon(x)\vartheta'(x)\Psi_{a}(x)\dd x=\iint_{\{x>y\geq\varepsilon\}}\frac{\Psi_{a}(x)(\phi_{a}\ast\Psi_{a})(y)}{(xy)^{3/2}}\Delta_{\cdot\vartheta(\cdot)}^\varepsilon(x,y)\dd x\dd y
\end{equation}
for all $\vartheta\in C^1([\varepsilon,\infty])$. Lastly, by again compactness, there exists a subsequence ${a}_k\rightarrow0$ such that $\Psi_{{a}_n}\wsc\Psi_\varepsilon\in\mathcal{I}^\varepsilon(E)$, and arguing as in the proof of Lemma \ref{lm:ex;regularizedexistence} it follows that this $\Psi_\varepsilon$ satisfies \eqref{eq:selfsimPsiprofile} for all $\vartheta\in C^1([\varepsilon,\infty])$.
\end{proof}

\subsubsection{Existence of an energy profile $\Psi\in\mathcal{M}_+((0,\infty))$}
In this subsection we show that $\Psi_\varepsilon\wsc\Psi$ as $\varepsilon\rightarrow0$, where $\Psi\in\mathcal{I}^0(E)$ satisfies \eqref{eq:selfsimPsiprofile} for all $\vartheta\in C_c^2((0,\infty])$, and $\int\!\!_{\{0\}}\Psi(x)\dd x=0$. To prove this last point, the following measure theory result will be needed.
\begin{lm}\label{lm:selfsimuniformconstant}
If $\mu\in\mathcal{M}_+((0,\infty))$ is such that $\inf({\rm supp}(\mu))=:i\in(0,1)$, then there exists at least one $\sigma\in(i,1)$ such that
\begin{equation}\label{eq:ssunfcons}
\iint_{[i,\infty)^2}\frac{\mu(x)\mu(y)}{(xy)^{3/2}}(x+y-\sigma)_+\wedge(\sigma-|x-y|)_+\dd x\dd y\geq \frac{1}{54}\Big({\textstyle\int_{(0,1]}\mu(x)\dd x}\Big)\!\Big.^2.
\end{equation}
\end{lm}
\begin{proof}
Let $q\in(\frac12,1)$ and let $n\in\mathbb{Z}$ the smallest integer such that $i\in[\tfrac12q^n,q^n)$. For $j\in\{0,\ldots,n\}$ we now define the measures $\mu_j\in\mathcal{M}_+((0,1])$ as
\begin{equation}\label{eq:ssunfcons_1}
\mu_j(x)=\frac{q^j\mu(q^jx)}{\int_{(0,q^j]}\mu(x)\dd x},
\end{equation}
and for all $j\in\{0,\ldots,n\}$ we then have
\begin{align}
\textstyle\text{either }\int_{(q,1]}\mu_j(x)\dd x&\geq 1-q,\label{eq:ssunfcons_2}\\
\textstyle\text{or }\int_{(0,q]}\mu_j(x)\dd x&>q.\label{eq:ssunfcons_3}
\end{align}
Suppose now that $\nu\in\{0,\ldots,n\}$ is the largest integer such that \eqref{eq:ssunfcons_3} holds for all $j\in\{0,\ldots,\nu-1\}$. We then have
\begin{equation}\label{eq:ssunfcons_4}
q^\nu\leq\prod_{j=0}^{\nu-1}\int_{(0,q]}\mu_j(x)\dd x=\prod_{j=0}^{\nu-1}\frac{\int_{(0,q^{j+1}]}\mu(x)\dd x}{\int_{(0,q^j]}\mu(x)\dd x}=\frac{\int_{(0,q^\nu]}\mu(x)\dd x}{\int_{(0,1]}\mu(x)\dd x},
\end{equation}
with equality holding only if $\nu=0$. Using next \eqref{eq:ssunfcons_1} and \eqref{eq:ssunfcons_4} it follows that
\begin{equation}\label{eq:ssunfcons_5}
\int_{(cq^\nu,q^\nu]}\mu(x)\dd x\geq q^\nu\Big({\textstyle\int_{(0,1]}\mu(x)\dd x}\Big)\int_{(c,1]}\mu_\nu(x)\dd x\text{ for }c\in[0,1).
\end{equation}

Suppose now first that $i\in[\frac12q^\nu,q^\nu)$. Setting then $\sigma:=\frac12(i+q^\nu)\in(i,1)$ we find by restricting the domain of integration that the left hand side of \eqref{eq:ssunfcons} can be bounded from below by 
\begin{equation}\label{eq:ssunfcons_6}
\iint_{[i,q^\nu]^2}\frac{\mu(x)\mu(y)}{q^{3\nu}}\big(\tfrac32i-\tfrac12q^\nu\big)\dd x\dd y\geq\frac{1}{4q^{2\nu}}\Big({\textstyle\int_{[i,q^\nu]}\mu(x)\dd x}\Big)\!\Big.^2.
\end{equation}
Supposing alternatively that $i\in(0,\frac12q^\nu)$, then setting $\sigma:=\frac12q^\nu\in(i,1)$ we estimate the left hand side of \eqref{eq:ssunfcons} from below by
\begin{equation}\label{eq:ssunfcons_7}
\iint_{(q^{\nu+1},q^\nu]^2}\frac{\mu(x)\mu(y)}{q^{3\nu}}\big(q^{\nu+1}-\tfrac12q^\nu\big)\dd x\dd y=\frac{2q-1}{2q^{2\nu}}\Big({\textstyle\int_{(q^{\nu+1},q^\nu]}\mu(x)\dd x}\Big)\!\Big.^2.
\end{equation}
Using finally \eqref{eq:ssunfcons_5} and the fact that \eqref{eq:ssunfcons_2} holds for $j=\nu$, we bound the right hand sides of \eqref{eq:ssunfcons_6} and \eqref{eq:ssunfcons_7} from below by
\begin{equation}\nonumber
\min\!\big\{\tfrac14,\tfrac{2q-1}{2}\big\}(1-q)^2 \Big({\textstyle\int_{(0,1]}\mu(x)\dd x}\Big)\!\Big.^2,
\end{equation}
and \eqref{eq:ssunfcons} follows by maximizing over $q$.
\end{proof}

\begin{lm}\label{lm:selfsimuniformbound}
Let $\varepsilon\in(0,1)$ and $E\in[0,\infty)$ be arbitrary and suppose that $\Psi_\varepsilon\in\mathcal{I}^\varepsilon(E)$ satisfies \eqref{eq:selfsimPsiregprofile} for all $\vartheta\in C^1([\varepsilon,\infty])$. Given then any $\sigma\in(\varepsilon,1)$ it holds that
\begin{equation}\label{eq:lm;ssunfbnd}\nonumber
E\geq\iint_{[\varepsilon,\infty)^2}\frac{\Psi_\varepsilon(x)\Psi_\varepsilon(y)}{(xy)^{3/2}}(x+y-\sigma)_+\!\wedge\!(\sigma-|x-y|)_+\dd x\dd y.
\end{equation}
\end{lm}
\begin{proof}
Let $\sigma\in(\varepsilon,1)$ be arbitrary and set $\varphi(x):=(x-\sigma)_+$. For $x>y\geq0$ we then compute that $\Delta_\varphi(x,y)=(x+y-\sigma)_+\wedge(\sigma-|x-y|)_+$, where the right hand side is symmetric. Furthermore, by convexity of $\varphi$ we remark for $x>y\geq\varepsilon$ that $\Delta_\varphi^\varepsilon(x,y)\geq\Delta_\varphi(x,y)$, so that the right hand side of \eqref{eq:selfsimPsiregprofile}, with $\vartheta(x):=\frac1x\varphi(x)$, can be estimated from below by
\begin{equation}\label{eq:lowerbound}
\frac12\iint_{[\varepsilon,\infty)^2}\frac{\Psi_\varepsilon(x)\Psi_\varepsilon(y)}{(xy)^{3/2}}(x+y-\sigma)_+\!\wedge\!(\sigma-|x-y|)_+\dd x\dd y.
\end{equation}
Approximating finally $\vartheta$ by functions $\vartheta_n\in C^1([\varepsilon,\infty])$, satisfying $\|\vartheta_n'\|\leq1$ and $\vartheta_n\rightarrow\vartheta$ in $C([\varepsilon,\infty])$ as $n\rightarrow\infty$, we now bound the left hand side of \eqref{eq:selfsimPsiregprofile} from above by $\frac12E$, and the right hand side from below by \eqref{eq:lowerbound} by dominated convergence. This then completes the proof.
\end{proof}
\begin{lm}\label{lm:Deltadotthetawelldef}
Given $\vartheta\in C_c^2((0,\infty])$, the mapping $(x,y)\mapsto(xy)^{-3/2}\Delta_{\cdot\vartheta(\cdot)}(x,y)$ is uniformly continuous on $\{x\geq y\geq 0\}$. In particular, the right hand side of \eqref{eq:selfsimPsiprofile} is well defined for all $\Psi\in\mathcal{M}_+([0,\infty))$ and all $\vartheta\in C_c^2((0,\infty])$.
\end{lm}
\begin{proof}
Recalling \eqref{eq:S1E18}, we find for $x\geq y\geq 0$ that
\begin{equation}\nonumber
|\Delta_{\cdot\vartheta(\cdot)}(x,y)|\leq2\min\big\{4x\|\vartheta\|_\infty,y(2x\|\vartheta'\|_\infty+\|\vartheta\|_\infty),y^2(x\|\vartheta''\|_\infty+\|\vartheta'\|_\infty)\big\}.
\end{equation}
Noticing further that $\Delta_{\cdot\vartheta(\cdot)}\equiv0$ on $\{0\leq y\leq x<\frac12\inf({\rm supp}(\vartheta))\}$, we conclude that the claim is true.
\end{proof}
Now we are able to take the limit $\varepsilon\rightarrow0$.
\begin{pr}\label{pr:ssb;existcandprofile}
Given $E\in[0,\infty)$ there exists at least one $\Psi\in\mathcal{M}_+((0,\infty))$ with $\int\!\!_{(0,\infty)}\Psi(x)\dd x=E$ that for all $\vartheta\in C_c^2((0,\infty])$ satisfies \eqref{eq:selfsimPsiprofile}.
\end{pr}
\begin{proof}
Without loss of generality we can restrict ourselves to strictly positive $E\in(0,\infty)$. Given any $\varepsilon\in(0,1)$ there exists at least one $\Psi_\varepsilon\in\mathcal{I}^\varepsilon(E)\subset\mathcal{I}^0(E)$ that satisfies \eqref{eq:selfsimPsiregprofile} for all $\vartheta\in C^1([\varepsilon,\infty])$ (cf.~Proposition \ref{pr:selfsimapproximatePsiprofile}). Then, since $\mathcal{I}^0(E)$ is weakly-\st compact, we can find a decreasing sequence $\varepsilon_k\rightarrow0$ such that $\Psi_{\varepsilon_k}\wsc\Psi\in\mathcal{I}^0(E)$. Given now $\vartheta\in C_c^2((0,\infty])$, then for $\varepsilon_k$ small enough we have
\begin{equation}\nonumber
\int_{[\varepsilon_k,\infty)}\tfrac12x\eta_{\varepsilon_k}(x)\vartheta'(x)\Psi_{\varepsilon_k}(x)\dd x=\int_{(0,\infty)}\tfrac12x\vartheta'(x)\Psi_{\varepsilon_k}(x)\dd x,
\end{equation}
and convergence of the left hand side of \eqref{eq:selfsimPsiregprofile} to the one of \eqref{eq:selfsimPsiprofile} as $\varepsilon_k\rightarrow0$ follows using the fact that $\lim_{x\rightarrow\infty}\vartheta'(x)=0$ and the fact that the first moments of $\Psi_{\varepsilon_k}$ are uniformly bounded. Noticing next that
\begin{equation}\nonumber
(xy)^{-3/2}\Delta_{\cdot\vartheta(\cdot)}^{\varepsilon_k}(x,y)\xrightarrow{\varepsilon_k\rightarrow0}(xy)^{-3/2}\Delta_{\cdot\vartheta(\cdot)}(x,y),
\end{equation}
uniformly on $\{x\geq y>0\}$ (cf.~Lemma \ref{lm:Deltadotthetawelldef}), convergence of the right hand sides also follows, and $\Psi\in\mathcal{I}^0(E)$ satisfies \eqref{eq:selfsimPsiprofile} for all $\vartheta\in C_c^2((0,\infty])$.

We finally show that indeed $\Psi\in\mathcal{M}_+((0,\infty))$. If we suppose by contradiction that $\int\!\!_{\{0\}}\Psi(x)\dd x=:e\in(0,E]$, then there exists a sequence $\delta_k\rightarrow0$ such that
\begin{equation}\label{eq:ssnodiracatzero_1}
\textstyle\int_{(0,\delta_k]}\Psi_{\varepsilon_k}(x)\dd x\xrightarrow{k\rightarrow\infty}e.
\end{equation}
Defining now for any $\delta_{k}$ the measure $\mu_k\in\mathcal{M}_+((0,\infty))$ as $\mu_k(x):=\delta_k\Psi_{\varepsilon_k}(\delta_kx)$, there then exist $\sigma_k\in(\inf({\rm supp}(\mu_k)),1)\neq\emptyset$ such that \eqref{eq:ssunfcons} holds (cf.~Lemma \ref{lm:selfsimuniformconstant}). Changing now variables and writing $\tilde\sigma_k=\delta_k\sigma_k\in(\varepsilon_k,1)$, we obtain
\begin{equation}\label{eq:ssnodiracatzero_2}
\begin{split}
&\delta_k^2\iint_{[\varepsilon_k,\infty)^2}\frac{\Psi_{\varepsilon_k}(x)\Psi_{\varepsilon_k}(y)}{(xy)^{3/2}}(x+y-\tilde\sigma_k)_+\wedge(\tilde\sigma_k-|x-y|)_+\dd x\dd y\\&\indent\geq \frac{1}{54}\Big({\textstyle\int_{(0,\delta_k]}\Psi_{\varepsilon_k}(x)\dd x}\Big)\!\Big.^2\xrightarrow[\eqref{eq:ssnodiracatzero_1}]{k\rightarrow\infty}\tfrac1{54}e^2>0.
\end{split}
\end{equation}
However, using Lemma \ref{lm:selfsimuniformbound} we know that the left hand side of \eqref{eq:ssnodiracatzero_2} tends to zero as $k\rightarrow\infty$, which is a contradiction. The result then follows.
\end{proof}
\subsubsection{Properties of the energy profile and proof of Proposition \ref{pr:selfsimexistenceofPhiprofile}}
We will now show that $\Phi(x):=\frac1x\Psi(x)$, with $\Psi$ as obtained in the previous subsection, is a finite measure on $(0,\infty)$. Moreover, $\Phi\in C^{0,\alpha}((0,\infty))$ for $\alpha<\frac12$. We first prove an estimate similar to the one obtained in Lemma \ref{lm:selfsimuniformbound}.

\begin{lm}\label{lm:secondkeylemma}
Suppose $\Psi\in\mathcal{M}_+((0,\infty))$ satisfies \eqref{eq:selfsimPsiprofile} for all $\vartheta\in C_c^2((0,\infty])$, then for all $\sigma\in(0,\infty)$ it holds that
\begin{equation}\label{eq:skl}
\int_{[\sigma,\infty)}\tfrac{\sigma}{x}\Psi(x)\dd x\geq\iint_{(\sigma,\infty)^2}\frac{\Psi(x)\Psi(y)}{(xy)^{3/2}}\big(\sigma-|x-y|\big)_+\dd x\dd y.
\end{equation}
\end{lm}
\begin{proof}
Fix $\sigma\in(0,\infty)$ arbitrarily, let $\sigma_-\in(0,\sigma)$ and $\sigma_+\in(\sigma,2\sigma)$ be arbitrary and let $\varphi\in C^2([0,\infty))$ be convex, satisfying $\varphi\equiv0$ on $[0,\sigma_-)$ and $\varphi\equiv(\,\cdot\,-\sigma)$ on $(\sigma_+,\infty)$. Since now $\frac1\cdot\varphi(\cdot)\in C_c^2((0,\infty])$ we immediately find that
\begin{equation}\label{eq:skl_1}
\int_{(0,\infty)}\tfrac12\big(\varphi'(x)-\tfrac1x\varphi(x)\big)\Psi(x)\dd x=\sint\frac{\Psi(x)\Psi(y)}{(xy)^{3/2}}\Delta_\varphi(x,y)\dd x\dd y.
\end{equation}
Noticing then that $(x+y-\sigma)-2(x-\sigma)=(\sigma-(x-y))$ and that $\Delta_\varphi\geq0$ by convexity, we estimate the right hand side of \eqref{eq:skl_1} from below by
\begin{equation}\label{eq:skl_2}
\frac12\iint_{(\sigma^+,\infty)^2}\frac{\Psi(x)\Psi(y)}{(xy)^{3/2}}\big(\sigma-\tfrac\sigma{\sigma^-}|x-y|\big)_+\dd x\dd y.
\end{equation}
Choose now sequences $\sigma_-\rightarrow\sigma$ and $\sigma_+\rightarrow\sigma$ and corresponding convex functions $\varphi\in C^2([0,\infty))$, as above, that converge pointwise to $(\,\cdot\,-\sigma)_+$. Then \eqref{eq:skl_2} converges by monotone convergence to one half times the right hand side of \eqref{eq:skl}. For the left hand side we finally notice that
\begin{equation}\nonumber
\varphi'(x)-\tfrac1x\varphi(x)\rightarrow\begin{cases}0&\text{ on }[0,\sigma),\\\frac\sigma{x}&\text{ on }(\sigma,\infty),\end{cases}
\end{equation}
and we conclude, by dominated convergence, that the left hand side of \eqref{eq:skl_1} is bounded from above by one half times the left hand side of \eqref{eq:skl}.
\end{proof}
Next we prove an estimate on the behaviour of $\Psi$ near the origin.
\begin{lm}\label{lm:selfsimpowerL1bounds}
Suppose $\Psi\in\mathcal{M}_+((0,\infty))$ satisfies \eqref{eq:selfsimPsiprofile} for all $\vartheta\in C_c^2((0,\infty])$, then there exists a constant $C\in(0,\infty)$ such that
\begin{equation}\label{eq:sspl1bnd_a}
\textstyle\int_{(0,\delta]}\Psi(x)\dd x\leq C\delta\text{ for all }\delta\in(0,\infty),
\end{equation}
and moreover
\begin{equation}\label{eq:sspl1bnd_b}
\lim_{\delta\rightarrow0}\frac{1}{\delta^{1+c}}\int_{(0,\delta]}\Psi(x)\dd x=0\text{ for all }c\in(-\infty,\tfrac12).
\end{equation}
\end{lm}
\begin{proof}
Fixing $\theta\in[\sqrt3-1,1)$ arbitrarily, then $\theta^2-(1-\theta)\geq\frac12\theta^2$ and for all $\delta\in(0,\infty)$ we find that
\begin{equation}\label{eq:sspl1bnd_1}\nonumber
\iint_{(\theta\delta,\delta]^2}\frac{\Psi(x)\Psi(y)}{(xy)^{3/2}}\big(\theta^2\delta-|x-y|\big)\dd x\dd y\geq\tfrac12\theta^2\Big({\textstyle\frac1\delta\int_{(\theta\delta,\delta]}\Psi(x)\dd x}\Big)\!\Big.^2.
\end{equation}
Since this is a lower bound on the right hand side of \eqref{eq:skl} with $\sigma=\theta^2\delta$, we obtain by Lemma \ref{lm:secondkeylemma} that
\begin{equation}\label{eq:sspl1bnd_2}
\int_{[\theta^2\delta,\infty)}\tfrac{\theta^2\delta}{x}\Psi(x)\dd x\geq\tfrac12\theta^2\Big({\textstyle\frac1\delta\int_{(\theta\delta,\delta]}\Psi(x)\dd x}\Big)\!\Big.^2\text{ for all }\delta\in(0,\infty).
\end{equation}
Now estimating the left hand side of \eqref{eq:sspl1bnd_2} from above by $\|\Psi\|$ we first obtain for all $\delta\in(0,\infty)$ that $\int\!\!_{(\theta\delta,\delta]}\Psi(x)\dd x\leq\frac1\theta\sqrt{2\|\Psi\|}\delta$. Using then this inequality with $\delta$ replaced by $\theta^j\delta$, with $j\in\mathbb{N}$, we find that
\begin{equation}\label{eq:sspl1bnd_3}
\textstyle\int_{(0,\delta]}\Psi(x)\dd x=\displaystyle\sum_{j=0}^\infty\Big(\textstyle\int_{(\theta^{j+1}\delta,\theta^j\delta]}\Psi(x)\dd x\Big)\leq\frac{\sqrt{2\|\Psi\|}}{\theta(1-\theta)}\,\delta\text{ for all }\delta\in(0,\infty),
\end{equation}
which proves \eqref{eq:sspl1bnd_a}.

For the second claim we again fix $\theta\in[\sqrt3-1,1)$ arbitrarily, so that we can use \eqref{eq:sspl1bnd_2}, and we suppose that there exists some $c\in(0,\frac12)$ such that
\begin{equation}\label{eq:sspl1bnd_4}
\limsup_{\delta\rightarrow0}\frac{1}{\delta^{1+c}}\int_{(\theta\delta,\delta]}\Psi(x)\dd x=\infty.
\end{equation}
For fixed $\kappa\in(0,\infty)$, arbitrarily large, there then exists a sequence $\delta_k\rightarrow0$ such that $\int\!\!_{(\theta\delta_k,\delta_k]}\Psi(x)\dd x\geq\kappa\delta_k^{1+c}$ for each $\delta_k$, and, assuming that $\delta_k<1$, it follows by \eqref{eq:sspl1bnd_2} that
\begin{equation}\label{eq:sspl1bnd_5}\nonumber
\int_{[\theta^2\delta_k,1)}\tfrac{\theta^2\delta_k}{x}\Psi(x)\dd x\textstyle+\theta^2\delta_k\|\Psi\|\geq\tfrac12\theta^2\kappa^2\delta_k^{2c},
\end{equation}
where the second term on the left hand side is due to the integral over $[1,\infty)$. Absorbing the linear term into the right hand side, which is possible since $2c<1$, and cancelling the factor $\theta^2$, we obtain for $k$ sufficiently large that
\begin{equation}\label{eq:4_56}
\textstyle\tfrac14\kappa^2\delta_k^{2c}\leq\delta_k\int_{[\theta^2\delta_k,1)}\tfrac1{x}\Psi(x)\dd x.
\end{equation}
Given now any $\sigma\in(0,1)$ we define $n\in\mathbb{N}$ to be the unique integer such that $\sigma\in[2^{-n},2^{1-n})$. Using then \eqref{eq:sspl1bnd_a}, we find that
\begin{equation}\label{eq:sspl1bnd_6}
\begin{split}
\int_{[\sigma,1)}\tfrac{1}{x}\Psi(x)\dd x
&\leq\sum_{j=1}^{n}\int_{[2^{-j},2^{1-j})}\tfrac{1}{x}\Psi(x)\dd x\leq\sum_{j=1}^{n}2^j\textstyle\int_{(0,2^{1-j}]}\Psi(x)\dd x\\&
\leq\sum_{j=1}^{n}2^{j+(1-j)}C=2C\,n\leq\tilde{C}\,|\log(\sigma)|,
\end{split}
\end{equation}
and combining \eqref{eq:4_56} and \eqref{eq:sspl1bnd_6} we obtain
\begin{equation}\label{eq:4_57}
0<\kappa^2\leq4\tilde{C}\,\delta_k^{1-2c}|\log(\theta^2\delta_k)|,
\end{equation}
with $\tilde{C}$ independent of $\delta_k$. However, the right hand side of \eqref{eq:4_57} tends to zero as $\delta_k\rightarrow0$ since $2c<1$, which disproves \eqref{eq:sspl1bnd_4} by contradiction and implies that $\limsup_{\delta\rightarrow0}\frac1{\delta^{1+c}}\int\!\!_{(\theta\delta,\delta]}\Psi(x)\dd x<\infty$ for all $c\in(0,\frac12)$. Moreover, for $c<\bar{c}\in(0,\frac12)$ we find that
\begin{equation}\label{eq:sspl1bnd_8}
\limsup_{\delta\rightarrow0}\frac1{\delta^{1+c}}\int_{(\theta\delta,\delta]}\Psi(x)\dd x\leq\lim_{\delta\rightarrow0}\delta^{\bar{c}-c}\limsup_{\delta\rightarrow0}\frac1{\delta^{1+\bar{c}}}\int_{(\theta\delta,\delta]}\Psi(x)\dd x=0,
\end{equation}
and by nonnegativity of $\Psi$ it follows that
\begin{equation}\label{eq:sspl1bnd_9}
\lim_{\delta\rightarrow0}\frac1{\delta^{1+c}}\int_{(\theta\delta,\delta]}\Psi(x)\dd x=0,
\end{equation}
for all $c\in(-\infty,\tfrac12)$. Fixing then again some arbitrary $c\in(0,\frac12)$ we note that
\begin{equation}\label{eq:sspl1bnd_10}\nonumber
\int_{(0,\delta]}\Psi(x)\dd x=\sum_{j=0}^\infty\Big({\textstyle\int_{(\theta^{j+1}\delta,\theta^j\delta]}\Psi(x)\dd x}\Big)\leq \frac{K(\delta)\,\delta^{1+c}}{1-\theta^{1+c}}\text{ as }\delta\rightarrow0,
\end{equation}
with $K(\delta)\rightarrow0$ as $\delta\rightarrow0$, which implies that \eqref{eq:sspl1bnd_9} holds for $\theta=0$ and $c\in(0,\frac12)$, and noticing finally that \eqref{eq:sspl1bnd_8} also holds with $\theta=0$, indeed \eqref{eq:sspl1bnd_b} follows.
\end{proof}
Using the previous lemma we can now prove boundedness of moments of $\Psi$. 
\begin{lm}\label{lm:ssb;finitemoments}
Suppose $\Psi\in\mathcal{M}_+((0,\infty))$ satisfies \eqref{eq:selfsimPsiprofile} for all $\vartheta\in C_c^2((0,\infty])$, then $\Psi$ has finite $\alpha$-th moment for all $\alpha\in(-\frac32,\infty)$.
\end{lm}
\begin{proof}
We first consider $\alpha\in(-\frac32,0]$. Fixing $\theta\in(0,1)$ arbitrarily we then find the decomposition
\begin{equation}\label{eq:lm;ssfinmom_1}
\int_{(0,\delta]}x^\alpha\Psi(x)\dd x=\sum_{j=0}^\infty\int_{(\theta^{j+1}\delta,\theta^j\delta]}x^\alpha\Psi(x)\dd x\text{ for all }\delta\in(0,\infty).
\end{equation}
Fix now some $c\in(0,1/2)$ with $1+c>-\alpha$. Then by Lemma \ref{lm:selfsimpowerL1bounds} there exists a constant $\delta^*\in(0,\infty)$ such that $\int\!\!_{(\theta\delta,\delta]}\Psi(x)\dd x\leq \delta^{1+c}$ for all $\delta\in[0,\delta^*]$. For all $\delta\in[0,\delta^*]$ it then follows that the right hand side of \eqref{eq:lm;ssfinmom_1} is bounded by
\begin{equation}\label{eq:lm;ssfinmom_2}\nonumber
\sum_{j=0}^\infty\frac{\int_{(\theta\,\theta^j\delta,\theta^j\delta]}\Psi(x)\dd x}{(\theta\,\theta^j\delta)^{-\alpha}}\leq\sum_{j=0}^\infty\frac{(\theta^j\delta)^{1+c}}{(\theta\,\theta^j\delta)^{-\alpha}}
\leq\frac{\theta^\alpha \delta^\beta}{1-\theta^\beta},
\end{equation}
where $\beta=1+c+\alpha>0$. Bounding the remaining part of the moment integral by $\delta^\alpha\int\!\!_{(\delta,\infty)}\Psi(x)\dd x$, we conclude the result for $\alpha\in(-\frac32,0]$.

For $\alpha\in(0,\infty)$, it is now sufficient to show that each even integer moment is bounded. To that end, let $n\in2\mathbb{N}$ and consider for any $\varepsilon\in(0,1)$ the function $\vartheta_{n,\varepsilon}(x)=x^n/(1+\varepsilon x^n)$, which is in the $\|\cdot\|_{C^1}$-closure of $C_c^2((0,\infty])$ and hence admissible. Using then the fact that the mapping $x\mapsto x^{n+1}-x\vartheta_{n,\varepsilon}(x)$ is convex on $(0,\infty)$, we find that
\begin{equation}\nonumber
\Delta_{\cdot\vartheta_{n,\varepsilon}(\cdot)}(x,y)\leq\Delta_{(\cdot)^{n+1}}(x,y)\leq2^{n+1}x^{n-1}y^2\text{ on }\{x\geq y\geq0\},
\end{equation}
hence we obtain by \eqref{eq:selfsimPsiprofile} that
\begin{equation}\nonumber
\frac{n}{2}\int_{(0,\infty)}\frac{x^n}{(1+\varepsilon x^n)^2}\Psi(x)\dd x\leq2^{n+1}\sint\frac{\Psi(x)\Psi(y)}{(xy)^{3/2}}x^{n-1}y^2\dd x\dd y.
\end{equation}
The right hand side is now smaller than $2^{n+1}\|\Psi\|\int\!\!_{(0,\infty)}x^{n-2}\Psi(x)\dd x$. By monotone convergence we then obtain
\begin{equation}\nonumber
\int_{(0,\infty)}x^n\Psi(x)\dd x\leq\frac{2^{n+2}\,\|\Psi\|}{n}\int_{(0,\infty)}x^{n-2}\Psi(x)\dd x,
\end{equation}
and the result follows by induction since the zeroth moment is finite. 
\end{proof}
As a last step before proving Proposition \ref{pr:selfsimexistenceofPhiprofile}, we prove a regularity result.
\begin{pr}\label{pr:selfsimcontinuousenergydensity}
Suppose that the measure $\Psi\in\mathcal{M}_+((0,\infty))$ satisfies \eqref{eq:selfsimPsiprofile} for all $\vartheta\in C_c^2((0,\infty])$, then $\Psi$ is absolutely continuous with respect to Lebesgue measure and its Radon-Nikodym derivative is $\alpha$-H\"older continuous for $\alpha<\frac12$, i.e.~it is actually a density function $\Psi\in C^{0,\alpha}((0,\infty))\cap L^1(0,\infty)$.
\end{pr}
\begin{proof}
Let $\chi\in C_c^\infty((0,\infty))$, then using $\vartheta(x):=\frac1x\int_0^x\chi(z)\dd z$ in \eqref{eq:selfsimPsiprofile} we obtain
\begin{equation}\label{eq:selfsimPsiL1_1}
\begin{split}
&\int_{[0,\infty)}\chi(x)\Psi(x)\dd x=\int_{[0,\infty)}\vartheta(x)\Psi(x)\dd x\\&\indent\indent+2\sint\frac{\Psi(x)\Psi(y)}{(xy)^{3/2}}\Big[{\textstyle\int_x^{x+y}\chi(z)\dd z-\int_{x-y}^x\chi(z)\dd z}\Big]\dd x\dd y,
\end{split}
\end{equation}
(cf.~\eqref{eq:Deltaestimate1}). Using H\"older's inequality we now find for arbitrary $p\in[1,\infty]$ that
\begin{equation}\label{eq:selfsimPsiL1_2}
\textstyle\int_a^b\chi(z)\dd z\leq(b-a)^{1-\frac1p}\|\chi\|_{L^p(0,\infty)}\text{ and }\vartheta(x)\leq x^{-\frac1p}\|\chi\|_{L^p(0,\infty)},
\end{equation}
which, by Lemma \ref{lm:ssb;finitemoments}, for $p\in(1,\infty]$ yields
\begin{equation}\label{eq:selfsimPsiL1_3}
\textstyle\big|\int_{[0,\infty)}\chi(x)\Psi(x)\dd x\big|\leq C\|\chi\|_{L^p(0,\infty)}
\end{equation}
with
\begin{equation}\label{eq:selfsimPsiL1_4}
\textstyle C=\int_{[0,\infty)}x^{-\frac1p}\Psi(x)\dd x+2\Big(\int_{[0,\infty)}x^{-\frac{1+2p}{2p}}\Psi(x)\dd x\Big)\!\Big.^2<\infty.
\end{equation}
Indeed, the first term in \eqref{eq:selfsimPsiL1_4} follows by estimating the first term on the right hand side of \eqref{eq:selfsimPsiL1_1}, using the second estimate in \eqref{eq:selfsimPsiL1_2}. Using the first estimate in \eqref{eq:selfsimPsiL1_2} and the fact that $x\geq y$ in the domain of integration, the second term on the right hand side of \eqref{eq:selfsimPsiL1_1} is bounded by
\begin{equation}\nonumber
4\sint\frac{\Psi(x)\Psi(y)}{(xy)^{3/2}}(xy)^{\frac12(1-\frac1p)}\dd x\dd y\times\|\chi\|_{L^p(0,\infty)},
\end{equation}
which yields the second term in \eqref{eq:selfsimPsiL1_4}.

Since $C_c^\infty((0,\infty))$ is dense in $L^p(0,\infty)$, $p\in[1,\infty)$, it now follows by duality from \eqref{eq:selfsimPsiL1_3} that $\Psi\in L^q(0,\infty)$ for all $q\in(1,\infty)$. Moreover, since $\|\Psi\|<\infty$ we actually have $\Psi\in L^q(0,\infty)$ for all $q\in[1,\infty)$.\\

To prove our continuity claim, fix $\gamma\in(0,\frac12)$ and $0<a<b<\infty$ arbitrarily. We now first show that given any $\varphi\in C^\infty(\mathbb{R})$ with ${\rm supp}(\varphi)\subset[a,b]$ we have
\begin{equation}\label{eq:selfsimPsicont_1}
\bigg|\int_\mathbb{R}\varphi'(x)\Psi(x)\dd x\bigg|\leq C(\Psi,a,b,\gamma)\|\varphi\|_{H^\gamma(\mathbb{R})}.
\end{equation}
Indeed, if ${\rm supp}(\varphi)\subset[a,b]$ then ${\rm supp}(\Delta_\varphi)\subset S_1\cup S_2$ with
\begin{equation}\nonumber
S_1:=\{2b\geq x\geq y\vee\tfrac{a}{2},\ b\geq y\geq0\}\text{ and }S_2:=\{2y\geq x\geq y\geq b\}.
\end{equation}
Using H\"older's inequality we have for $y\in[0,b]$ the inequality
\begin{equation}\nonumber
\int_{y\vee\frac{a}{2}}^{2b}\frac{\Psi(x)}{x^{3/2}}|\Delta_\varphi(x,y)|\dd x\leq2(\tfrac{2}{a})^{3/2}\|\Psi\|_{L^2(\frac{a}{2},2b)}\|\varphi(\cdot+y)-\varphi(\cdot)\|_{L^2(\mathbb{R})},
\end{equation}
where the right hand side is bounded by $C_1(\Psi,a,b,\gamma)\|\varphi\|_{H^\gamma(\mathbb{R})}y^\gamma$ (cf.~\cite[Cor.~24]{S90}), hence
\begin{equation}\nonumber
\iint_{S_1}\frac{\Psi(x)\Psi(y)}{(xy)^{3/2}}\Delta_\varphi(x,y)\dd x\dd y\leq \Big({\textstyle C_1\int_0^by^{\gamma-3/2}\Psi(y)\dd y}\Big)\|\varphi\|_{H^\gamma(\mathbb{R})}.
\end{equation}
Noting that $\Delta_\varphi(x,y)=\varphi(x-y)$ on $S_2$, we first observe that by Young's inequality for convolutions we have
\begin{equation}\nonumber
\big\|\widehat{\Psi\ast\varphi}\big\|_{L^2(b,\infty)}^2:=\int_b^\infty\Big|{\textstyle\int_{b\vee\frac{x}{2}}^x\Psi(y)|\varphi(x-y)|\dd y}\Big|^2\dd x\leq\|\Psi\|_{L^1(0,\infty)}^2\|\varphi\|_{L^2(\mathbb{R})}^2,
\end{equation}
which then yields
\begin{equation}\nonumber
\iint_{S_2}\frac{\Psi(x)\Psi(y)}{(xy)^{3/2}}\Delta_\varphi(x,y)\dd x\dd y\leq b^{-3}\|\Psi\|_{L^2(b,\infty)}\big\|\widehat{\Psi\ast\varphi}\big\|_{L^2(b,\infty)}\leq C_2\|\varphi\|_{H^\gamma(\mathbb{R})}.
\end{equation}
Using finally $\vartheta(x):=\frac1x\varphi(x)$ in \eqref{eq:selfsimPsiprofile} we thus obtain
\begin{equation}\nonumber
\bigg|\int_\mathbb{R}\varphi'(x)\Psi(x)\dd x\bigg|\leq\tfrac1a\|\Psi\|_{L^2(a,b)}\|\varphi\|_{L^2(\mathbb{R})}+2(\tilde{C}_1+C_2)\|\varphi\|_{H^\gamma(\mathbb{R})},
\end{equation}
which implies that \eqref{eq:selfsimPsicont_1} holds. Let next $\zeta\in C^\infty(\mathbb{R})$ be a cut off function such that $\zeta\equiv0$ on $[a,b]^c$ and $\zeta\equiv1$ on $I_{a,b}:=(\frac{2a+b}{3},\frac{a+2b}{3})$, and define $\Theta\equiv\zeta\Psi$. Given then any $\varphi\in C_c^\infty(\mathbb{R})$ we have
\begin{equation}\nonumber
\bigg|\int_\mathbb{R}\varphi'(x)\Theta(x)\dd x\bigg|\leq\bigg|\int_\mathbb{R}(\varphi\zeta)'(x)\Psi(x)\dd x\bigg|+\|\zeta'\|_{L^4(a,b)}\|\Psi\|_{L^4(a,b)}\|\varphi\|_{L^2(\mathbb{R})},
\end{equation}
which is bounded by $\tilde{C}(\Psi,a,b,\zeta,\gamma)\|\varphi\|_{H^\gamma(\mathbb{R})}$, where we use \eqref{eq:selfsimPsicont_1} to estimate the first term on the right hand side by $C\|\varphi\zeta\|_{H^\gamma(\mathbb{R})}\leq C'\|\varphi\|_{H^\gamma(\mathbb{R})}$. We thus obtain
\begin{equation}\nonumber
\bigg|\int_\mathbb{R}\varphi'(x)\Theta(x)\dd x\bigg|\leq\tilde{C}\|\varphi\|_{H^\gamma(\mathbb{R})}\text{ for all }\varphi\in H^\gamma(\mathbb{R}),
\end{equation}
from which we get that $\Theta'\in H^{-\gamma}(\mathbb{R})\equiv(H^\gamma(\mathbb{R}))^*$. By \cite[Thm.~19.6(b)]{DK10} it then follows that $\Theta\in H^{1-\gamma}(\mathbb{R})\subset C^{0,\alpha}(\mathbb{R})$, $\alpha=\frac12-\gamma$, and since $\Theta\equiv\Psi$ on $I_{a,b}$ we have $\Psi\in C^{0,\alpha}(I_{a,b})$. Covering now $(0,\infty)$ with intervals $I_{a_i,b_i}$, the result follows.
\end{proof}

\begin{proof}[Proof of Proposition \ref{pr:selfsimexistenceofPhiprofile}]
Let $\Psi\in C^{0,\alpha}((0,\infty))\cap L^1(0,\infty)$, $\alpha<\frac12$, be a nonnegative function with $\int\!\!_{(0,\infty)}\Psi(x)\dd x=E$ that satisfies \eqref{eq:selfsimPsiprofile} for all $\vartheta\in C_c^2((0,\infty])$, which exists by Propositions \ref{pr:ssb;existcandprofile} and \ref{pr:selfsimcontinuousenergydensity}. We now define the function $\Phi\in C^{0,\alpha}((0,\infty))$ to be the unique one that satisfies $\Psi(x)=x\Phi(x)$ for all $x\in(0,\infty)$ and we show that this $\Phi$ satisfies the required conditions.

By definition we have $\int\!\!_{(0,\infty)}x\Phi(x)\dd x=E$ and since the first negative moment of $\Psi$ is bounded (cf.~Lemma \ref{lm:ssb;finitemoments}) it follows that $\Phi\in L^1(0,\infty)$. Let next $\varphi\in C_c^2((0,\infty])$ be arbitrary, then using $\vartheta(x):=\frac1x\varphi(x)$ in \eqref{eq:selfsimPsiprofile} we obtain by rearranging terms that
\begin{equation}\label{eq:selfsimPhiprof_1}
\int_{(0,\infty)}\tfrac12\big(x\varphi'(x)-\varphi(x)\big)\Phi(x)\dd x=\sint\frac{\Phi(x)\Phi(y)}{\sqrt{xy}}\Delta_\varphi(x,y)\dd x\dd y.
\end{equation}

If we now show that \eqref{eq:selfsimPhiprof_1} is also satisfied by functions $\varphi\in C^1([0,\infty))$ of the form $\varphi(x):=\omega(x)-\omega(0)-\omega'(0)x$ with $\omega\in C^1([0,\infty])$, then the result follows since $x\varphi'(x)-\varphi(x)=x\omega'(x)-\omega(x)+\omega(0)$ for all $x\in[0,\infty)$ and $\Delta_\varphi\equiv\Delta_\omega$. We thus set $\tilde\varphi_{a,A}(x):={\rm sign}(\varphi(x-a))(|\varphi(x-a)|\wedge A)$ for $x\geq a$ and $\tilde\varphi_{a,A}(x)=0$ for $x<a$, with $a\in(0,1)$ small and $A\in(1,\infty)$ large. Mollifying this equation we get admissible test functions $\tilde\varphi_{a,A,n}\in C_c^2((0,\infty])$ satisfying $\|\tilde\varphi_{a,A,n}'\|_\infty\leq C$, independent of $a$, $A$ and $n$, and such that $\tilde\varphi_{a,A,n}\rightarrow\tilde\varphi_{a,A}$ in $W^{1,p}(0,\infty)$, for any $p\in[1,\infty)$, as $n\rightarrow\infty$. Therefore \eqref{eq:selfsimPhiprof_1} holds with the test functions $\tilde\varphi_{a,A}$ and sending first $A\rightarrow\infty$ and then $a\rightarrow0$, it follows by dominated convergence that $\varphi$ satisfies \eqref{eq:selfsimPhiprof_1}, where we note that $\tilde\varphi_{a,\infty}\rightarrow\varphi$ in $C^1([0,\infty])$.
\end{proof}

\begin{proof}[Proof of Theorem \ref{tm:self-similarsolutions}]
The proof is an immediate consequence of Propositions \ref{pr:rearragementresult} and \ref{pr:selfsimexistenceofPhiprofile}.
\end{proof}

\subsection{Solutions invariant under the scaling $G\mapsto G_{1,\lambda}$ (cf.~\eqref{eq:resckl})}
\begin{proof}[Proof of Proposition \ref{prop:selfsimmasssol}]
Let $G\in C([0,\infty):\mathcal{M}_+([0,\infty)))$ be any weak solution to \eqref{eq:def:solution} with initially finite first moment. Then $\int\!\!_{(0,\infty)}xG(\cdot,x)\dd x$ is constant by conservation of initially finite first moments (cf.~Proposition \ref{pr:conservationlaws}) and\linebreak since $\int\!\!_{\{0\}}x\mu(x)\dd x=0$ for any measure $\mu\in\mathcal{M}_+([0,\infty))$. Computing now for any $t\in[0,\infty)$ that
\begin{equation}\nonumber
\int_{(0,\infty)}xG(t,x)\dd x=\frac1\lambda\int_{(0,\infty)}xG(\lambda t,x)\dd x\text{ for all }\lambda\in(0,\infty),
\end{equation}
we then find that $\int\!\!_{(0,\infty)}xG(\cdot,x)\dd x\equiv0$ on $[0,\infty)$, hence $G$ is trivial.
\end{proof}

\begin{ack}
We thank B.~Niethammer for comments that helped to clarify the structure of self-similar solutions to problems with multiple conserved quantities and remarks concerning the final form of this paper.

The authors acknowledge support through the {\em CRC 1060 The mathematics of emergent effects} at the University of Bonn, that is funded through the German Science Foundation (DFG).
\end{ack}



\begin{thebibliography}{}

\bibitem{BCEP08}
Benedetto, D., Castella, F., Esposito, R., Pulvirenti, M.; From the N-body Schr\"odinger Equation to the Quantum Boltzmann Equation: a Term-by-Term Convergence Result in the Weak Coupling Regime; Comm.~Math.~Phys.~277, 1-44 (2008)

\bibitem{B11}
Brezis, H.; Functional Analysis, Sobolev Spaces and Partial Differential Equations; Springer 2011

\bibitem{DK10}
Duistermaat, J.J., Kolk, J.A.C.; Distributions, Theory and Applications; Birkh\"auser 2010

\bibitem{DS63}
Dunford, N., Schwartz, J.T.; Linear Operators, Part I.~General Theory; John Wiley \& Sons, Inc.~1963

\bibitem{D06}
D\"uring, G., Josserand, C., Rica, S.; Weak Turbulence for a Vibrating Plate: Can One Hear a Kolmogorov Spectrum?; Phys.~Rev.~Lett.~97, 025503 (2006)

\bibitem{DNPZ92}
Dyachenko, S., Newell, A.C., Pushkarev, A., Zakharov, V.E.; Optical turbulence: weak turbulence, condensates and collapsing filaments in the nonlinear Schr\"odinger equation; Physica D, 57, 96-160 (1992)

\bibitem{E65}
Edwards, R.E.; Functional Analysis.~Theory and Applications; Holt, Rinehart and Winston, Inc.~1965

\bibitem{EMR05}
Escobedo, M., Mischler, S., Rodriguez Ricard, M.; On self-similarity and stationary problem for fragmentation and coagulation models; Ann.~Inst.~H. Poincar\'e Anal.~Non Lin\'eaire, 22, 99-125 (2005)

\bibitem{EV14a}
Escobedo, M., Vel\'azquez, J.J.~L.; Finite time blow-up and condensation for the bosonic Nordheim equation; Invent.~Math.~(online) (2014)

\bibitem{EV14b}
Escobedo, M., Vel\'azquez, J.J.~L.; On the theory of Weak Turbulence for the Nonlinear Schr\"odinger Equation; Memoirs AMS (to appear) arXiv: 1305.5746v1 [math-ph]

\bibitem{FL05}
Fournier, N., Lauren\c cot, P.; Existence of Self-Similar Solutions to Smoluchowski's Coagulation Equation; Comm.~Math.~Phys.~256, 589-609 (2005)


\bibitem{GSRT14}
Gallagher, I., Saint-Raymond, L., Texier, B.; From Newton to Boltzmann: Hard Spheres and Short-Range Potentials; EMS, Zurich Lectures in Advanced Mathematics 2014

\bibitem{GPV04}
Gamba, I.M., Panferov, V., Villani, C.; On the Boltzmann equation for diffusively excited granular media; Comm.~Math.~Phys., 246, 503-541 (2004)

\bibitem{HLP52}
Hardy, G.H., Littlewood, J.E., P\'olya, G.; Inequalities, 2nd edition; Cambridge University Press 1952

\bibitem{H62}
Hasselmann, K.; On the non-linear energy transfer in a gravity-wave spectrum.~Part 1.~General Theory; J.~Fluid Mech.~12, 481-500 (1962)

\bibitem{H63}
Hasselmann, K.; On the non-linear energy transfer in a gravity-wave spectrum.~Part 2.~Conservation theorems, wave-particle analogy, irreversibility; J.~Fluid Mech.~15, 273-281 (1963)

\bibitem{L75}
Lanford, III, O.E.; Time evolution of large classical systems; Lecture Notes in Physics 38, pp.~1-111, Springer 1975
 

\bibitem{L04}
Lu, X.; On Isotropic Distributional Solutions to the Boltzmann Equation for Bose-Einstein Particles; J.~Stat.~Phys.~116, 1597-1649 (2004)

\bibitem{L13}
Lu, X.; The Boltzmann Equation for Bose-Einstein Particles: Condensation in Finite Time; J.~Stat.~Phys.~150, 1138-1176 (2013)

\bibitem{L14}
Lu, X.; The Boltzmann Equation for Bose-Einstein Particles: Regularity and Condensation; J.~Stat.~Phys.~156, 493-545 (2014)

\bibitem{LS11}
Lukkarinen, J., Spohn, H.; Weakly nonlinear Schr\"odinger equation with random initial data; Invent.~Math.~183, 79-188 (2011)

\bibitem{N11}
Nazarenko, S.; Wave Turbulence; Lecture Notes in Physics 825, Springer 2011

\bibitem{N68}
Newell, A.C.; The closure problem in a system of random gravity waves; Reviews of Geophysics 6, 1-31 (1968)

\bibitem{N69}
Newell, A.C.; Rossby wave packet interactions; J.~Fluid Mech.~35, 255-271 (1969)

\bibitem{NV13}
Niethammer, B., Vel\'azquez, J.J.~L.; Self-similar Solutions with Fat Tails for Smoluchowski's Coagulation Equation with Locally Bounded Kernels; Comm.~Math.~Phys.~318, 505-532 (2013)

\bibitem{P29}
Peierls, R.; Zur kinetischen Theorie der W\"armeleitung in Kristallen; Ann. Phys.~395, 1055-1101 (1929)

\bibitem{PSS14}
Pulvirenti, M., Saffirio, C., Simonella, S.; On the validity of the Boltzmann equation for short range potentials; Rev.~Math.~Phys.~26, 1450001 (2014)


\bibitem{S90}
Simon, J.; Sobolev, Besov and Nikolskii Spaces: Imbeddings and Comparisons for Vector Valued Spaces on an Interval; Ann.~Mat.~Pura Appl.~(4) 157, 117-148 (1990)

\bibitem{S10}
Spohn, H.; Kinetics of the Bose-Einstein condensation; Phys.~D 239, 627-634 (2010)

\bibitem{Z67}
Zakharov, V.E.; Weak-turbulence Spectrum in a Plasma Without a Magnetic Field; Zh.~Eksp.~Teor.~Fiz.~51, 688-696 (1967) [Sov.~Phys.~JEPT 24, 455-459 (1967)]

\bibitem{Z72}
Zakharov, V.E.; Collaps of Langmuir Waves; Zh.~Eksp.~Teor.~Fiz.~62, 1745-1759 (1972) [Sov.~Phys.~JEPT 35, 908-914 (1972)]

\bibitem{ZF67b}
Zakharov, V.E., Filonenko, N.N.; Weak turbulence of capillary waves; Zh. Prikl.~Mekh.~Tekh.~Fiz.~8(5), 62-67 (1967) [J.~Appl.~Mech.~Tech.~Phys.~8(5), 37-40 (1967)]

\bibitem{ZLF92}
Zakharov, V.E., L'vov, V.S., Falkovich, G.; Kolmogorov spectra of turbulence I.~Wave turbulence; Springer 1992

\end{thebibliography}
\end{document}